\documentclass[11pt,a4paper]{article}
\usepackage[T1]{fontenc}
\usepackage{lmodern}
\usepackage[a4paper,margin=1in]{geometry}
\usepackage{amsmath,amssymb,amsthm,mathtools}
\usepackage{array,booktabs,tabularx,longtable,enumitem}
\usepackage{graphicx}
\usepackage{placeins}
\usepackage{tikz}
\usetikzlibrary{arrows.meta}
\usepackage[title]{appendix}
\usepackage[numbers,sort&compress]{natbib}
\usepackage{microtype}
\usepackage[hidelinks]{hyperref}
\usepackage[nameinlink,capitalise,noabbrev]{cleveref}
\setcitestyle{numbers,square}
\newtheoremstyle{thmstyleone}{7pt plus2pt minus1pt}{7pt plus2pt minus1pt}{\normalfont\itshape}{0pt}{\normalfont\bfseries}{.}{.5em}{}
\newtheoremstyle{thmstyletwo}{7pt plus2pt minus1pt}{7pt plus2pt minus1pt}{\normalfont}{0pt}{\normalfont\itshape}{.}{.5em}{}
\newtheoremstyle{thmstylethree}{7pt plus2pt minus1pt}{7pt plus2pt minus1pt}{\normalfont}{0pt}{\normalfont\bfseries}{.}{.5em}{}
\theoremstyle{thmstyleone}
\newtheorem{theorem}{Theorem}[section]
\newtheorem{introtheorem}{Theorem}

\newtheorem{lemma}[theorem]{Lemma}
\newtheorem{proposition}[theorem]{Proposition}
\newtheorem{corollary}[theorem]{Corollary}
\theoremstyle{thmstylethree}
\newtheorem{definition}[theorem]{Definition}
\newtheorem{example}[theorem]{Example}
\theoremstyle{thmstyletwo}
\newtheorem{remark}[theorem]{Remark}
\crefname{example}{Example}{Examples}
\Crefname{example}{Example}{Examples}
\hypersetup{
 bookmarksdepth=2,
 pdftitle={Random Clocks on Computation Paths: How Far Can They Reach, and When Does Stepping Back Help?},
 pdfauthor={Saman Farhat and Farzad Didehvar},
 pdfsubject={Observation complexity, crossing probabilities, and stochastic accessibility},
 pdfkeywords={random clocks, crossing probabilities, observation complexity, stochastic accessibility, sparse probability laws, Markov chains}
}
\usepackage{etoolbox}
\AtBeginEnvironment{appendices}{%
 \setcounter{equation}{0}\setcounter{figure}{0}\setcounter{table}{0}%
}
\title{Random Clocks on Computation Paths: How Far Can They Reach, and When Does Stepping Back Help?}
\author{Saman Farhat\quad Farzad Didehvar\\[0.6em]
 \normalsize Amirkabir University of Technology, Tehran, Iran\\[0.3em]
 \small\href{mailto:saman.farhat@aut.ac.ir}{saman.farhat@aut.ac.ir}, \href{mailto:didehvar@aut.ac.ir}{didehvar@aut.ac.ir}}
\date{}
\begin{document}
\maketitle
\begin{abstract}
We rarely redo a computation from scratch. A stored value can be looked up, a familiar pattern recalled, an answer predicted by a trained model, or a saved earlier state revisited. Each of these reaches a stage of a computation without going through the stages in between. We model this by a random clock, fixed in advance, that moves along a deterministic computation by jumping ahead or stepping back. Each observation costs one unit. We ask what can be recognised with polynomially many observations when such jumps are allowed, and which complexity classes arise. When the jumps are independent, the answer depends mainly on how likely very long jumps are. Explicit clocks of this kind give exactly $\mathbf{EXP}$, $\mathbf{2EXP}$, and each quasipolynomial class $\mathbf{QP}_k$, so single clocks realise a strict hierarchy between $\mathbf{P}$ and $\mathbf{EXP}$. Below $\mathbf{P}$ the picture changes. A clock that tends to drift backward, with long forward jumps very rare, recognises exactly the reverse-definite languages, whose membership is decided by a bounded prefix of the input. Between these and $\mathbf{P}$ we construct infinitely many pairwise incomparable classes. Backward moves can help a great deal. We construct families of clocks in which every clock that only moves forward stays within $\mathbf{P}$, while even occasional single steps back raise the class to $\mathbf{EXP}$ and beyond.
\end{abstract}


\section{Introduction}
\label{sec:introduction}

\subsection{The object, the motivation, and the question}
\label{sec:opening-motivation}

\paragraph{The object.}
A deterministic machine running on an input $x$ produces a path of
configurations $C^x_0,C^x_1,C^x_2,\dots$, one per step. We do not follow this
path step by step. Instead, an observer looks at it only at indices chosen by
a random clock. At each tick the clock may jump forward, step back, or stay
where it is, and the observer sees the configuration at the clock's current
index (\cref{fig:observation}). Once the machine accepts, its accepting
configuration repeats forever, so the input is recognised as soon as the clock
reaches or exceeds the first accepting index.

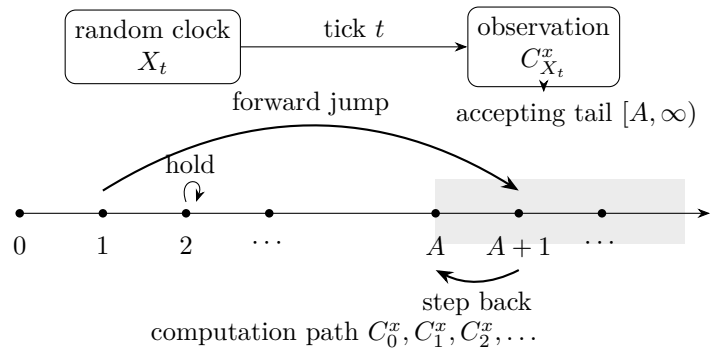
\begin{figure}[ht]
\centering
\begin{tikzpicture}[>=Stealth,font=\small,x=1.1cm,y=1cm]
 \fill[gray!15] (5,-.4) rectangle (8,.45);
 \draw[->] (0,0)--(8.3,0);
 \foreach \x/\lab in {0/0,1/1,2/2,3/{\cdots},5/{A},6/{A+1},7/{\cdots}}
  {\fill (\x,0) circle (1.6pt); \node[below=5pt] at (\x,0) {$\lab$};}
 \node[above=4pt] at (6.7,.85) {accepting tail $[A,\infty)$};
 \draw[->,thick] (1,.3) to[bend left=32] node[above] {forward jump} (6,.3);
 \draw[->,thick] (6,-.75) to[bend left=27] node[below] {step back} (5,-.75);
 \draw[->] (2,.15) to[out=110,in=70,looseness=7] (2.12,.15);
 \node at (2.05,.65) {hold};
 \node[draw,rounded corners,align=center] (clock) at (1.6,2.2) {random clock\\$X_t$};
 \node[draw,rounded corners,align=center] (obs) at (6.3,2.2) {observation\\$C^x_{X_t}$};
 \draw[->] (clock)--node[above] {tick $t$} (obs);
 \draw[->] (obs.south)--(6.3,1.6);
 \node[align=center] at (3.9,-1.6) {computation path $C^x_0,C^x_1,C^x_2,\ldots$};
\end{tikzpicture}
\caption{A tick updates the clock index $X_t$ and supplies the observation
$C^x_{X_t}$. Crossing means $X_t\ge A$ at some observed tick; persistent
acceptance makes that observation sufficient.}
\label{fig:observation}
\end{figure}
\FloatBarrier

\paragraph{Why skip steps?}
Computations are often reused rather than carried out from scratch. A program
can store a computed value and look it up when it is needed
again~\cite{Michie1968}. Studies of chess expertise likewise illustrate the
role of learned perceptual patterns and their recall~\cite{ChaseSimon1973}.
Language models provide another contrast: a prompt may elicit written
intermediate steps, or an answer directly~\cite{Wei2022}. These examples
motivate asking what is gained when access to a later state replaces
inspection of the intervening states. The extent of a shortcut may vary;
we represent this variation by a random clock, whose jump distribution may
also depend on the current position. We idealise access: a shortcut lands
on the true configuration. Uncertainty concerns how far one jumps, not
whether the result is correct. The examples motivate the abstraction;
they do not supply its configuration-access mechanism.

\paragraph{Why look back?}
A saved earlier state can be revisited, making another route forward
available. In our clocks a backward step returns to an earlier index; it does
not reverse the deterministic computation. It can return the observer to a
position from which a long jump starts. Backtracking in a reasoning trace is
an informal analogue. In the constructions of \cref{sec:backward}, large
forward jumps are available only at specified gate positions. A forward jump
can overshoot the next gate, while a step back by one can recover it. Within
the prescribed families, these returns give crossing rates that no permitted
forward-only clock achieves. For independent proposals, deleting backward
jumps never reduces reach (\cref{prop:deletion}).

\paragraph{Why count looks?}
Classical running time charges each executed transition. To measure access
along a fixed computation path, the \emph{jumpy observation model} instead
charges one unit per observed configuration, and nothing for producing
configurations, sampling the clock, or writing indices. This separates the
logical length of a computation from the number of observations needed to
recognise its outcome. Within this abstraction, threshold-crossing
probabilities determine the observation budget. The resulting classes
describe an observation resource, without asserting an ordinary machine
speedup or a physical implementation.

\paragraph{The question.}
The clock is fixed in advance, independently of the input and the machine;
a position-dependent clock uses the current index, not the contents of the
configuration. This restriction matters: an externally prescribed schedule
$X_t=2^t$ reaches index $2^{p(n)}$ after $p(n)$ observations, placing the
access resource in the schedule itself. For a fixed random clock, the
probability of reaching a threshold within $q$ ticks determines the languages
recognisable at confidence $\alpha$ (\cref{lem:reachability}). We ask which
jump laws and position-dependent kernels let polynomially many observations
recognise which languages, and when backward access enlarges this power
beyond every comparable forward-only clock in a specified family.

\paragraph{Background.}
Farzad Didehvar developed the programme of \emph{fuzzy-time computation}
in his 2023 manuscript and its 2026 edited version~\cite{didehvar2023tc,didehvar2026tc}.
These works motivate asking how uncertain access to logical time affects
computation. The present paper develops that motivation through a fixed
random clock and an explicit observation budget.

The present paper is a
further attempt to formalise these ideas. As a first step in this formulation,
we study one perspective: stochastic access to a fixed deterministic
computation path and the resulting observation complexity. A random clock
selects indices by moving forward, backward, or remaining in place. This
perspective separates the computation's logical length from the observations
needed to reach its accepting tail.
Clocks with independent proposals are held random walks with one fixed law
$\mu$, called \emph{fuzzy clocks} in \cref{def:kt-walk}; we also consider
position-dependent kernels. For a fixed clock, ``fuzzy'' concerns probabilistic
uncertainty in the observed time index, rather than fuzzy-set membership.
Probability specifies that clock's motion. Possibility cuts~\cite{Zadeh1978}
express incomplete specification by selecting allowed clock parameters;
\cref{sec:fuzzy} studies observation budgets common to all the permitted
parameters.

\subsection{A gate that a forward jump misses}
The doubling gates of \cref{def:gated} illustrate the accessibility question.
At a gate $g=2^k$, the clock moves to $2g+1$, just beyond the next gate
$2g$. A return by one reaches that gate, where the next large jump is
available. A forward-only clock cannot revisit the gate just passed.
If its forward displacement away from gates is bounded by a fixed integer
$w$, accessing another gate requires travel using those bounded steps. The proof
of \cref{thm:gated} charges each gate's gain to this intervening travel
and bounds the reach of every forward-only member by a linear function
of the tick budget. Length-one returns instead allow logarithmic crossing
budgets. Squaring gates give the corresponding polynomial forward-only
bound and doubly logarithmic budget (\cref{thm:sq-gated}).

This comparison concerns all forward-only kernels within the prescribed
gate families. The parity clock makes a weaker comparison: deletion
changes its class from exponential to polynomial time, but replacing
returns with forward re-entry also restores exponential time
(\cref{rem:forward-reentry}). Independent-proposal clocks satisfy a
deletion obstruction: suppressing negative proposals never decreases
their reach (\cref{prop:deletion}).

\subsection{Main results}
For a clock started at zero and a fixed confidence $0<\alpha<1$, let
$R_\alpha(q)$ be the greatest index crossed by tick $q$ with probability
at least $\alpha$. Persistent acceptance gives the reachability equivalence
of \cref{lem:reachability}: a recogniser belongs to the class exactly when
its accepting times are bounded by $R_\alpha(Q(|x|))$ for some polynomial
budget $Q$. Sequential input access makes small radii restrictive, while
large radii permit longer computations to be observed.

The two principal construction results concern the order of classes below
$\mathbf{P}$ and the necessity of backward access within gated families.
They are followed by the deletion comparison and the iid tail hierarchy.
The zero-drift example below is separate from the negative-drift antichain;
no ordering between them is established.

\begin{introtheorem}[Below $\mathbf{P}$: locality and an antichain; \cref{thm:negative-all-moments,thm:zero-drift-counterexample,thm:antichain}]
\label{thm:summary-below-p}
Negative drift with positive forward probability and all positive moments
finite gives exactly the reverse-definite languages. One computable
zero-drift law has a class strictly between the reverse-definite languages
and $\mathbf{P}$, and computable negative-drift laws give infinitely many
pairwise incomparable such classes, at every pair of confidences.

\end{introtheorem}

For the fixed-return gated clock of \cref{sec:gated}, $\beta\in(0,1]$ is the probability of a one-step backward return at a landing index.

\begin{introtheorem}[Backward moves are necessary, and may be sparse; \cref{thm:gated,thm:sq-gated,cor:density-zero,prop:vanishing-beta}]
\label{thm:summary-sparse}
In a family whose only large forward moves are doublings at powers of
two, every forward-only clock has class within $\mathbf{P}$, whereas steps
back by one, confined to a set of indices of density zero and possibly
with vanishing probability, give $\mathbf{EXP}$. The fixed-return gated clock
has crossing budget $(1+o(1))(1+\beta^{-1})\log_2T$ at every confidence.
With squaring gates, backward indices of counting order $\log\log N$
give $\mathbf{2EXP}$ while forward-only clocks stay within $\mathbf{P}$.

\end{introtheorem}

For the return-enabled parity clock of \cref{sec:activation}, $\beta\in(0,1/2]$ is the backward-return probability at a positive even index. Write $\ell_\alpha=-\ln(1-\alpha)$ for $0<\alpha<1$, where $\ln$ is the natural logarithm.

\begin{introtheorem}[When backward moves help; \cref{thm:monotone-rows,thm:activation,thm:sharp}]
\label{thm:summary-backward}
If the forward-displacement law is stochastically nondecreasing in
position, deleting backward moves never lowers the class. For a
parity-dependent clock, a step back by one reopens a heavy-sampling phase:
the class rises from $\mathbf{P}$ to $\mathbf{EXP}$ and the crossing budget falls from
$(1+o(1))T$ to $(1+o(1))(1+\beta^{-1})\ell_\alpha\log_2T$, where
$\ell_\alpha=-\ln(1-\alpha)$.

\end{introtheorem}

\begin{introtheorem}[Above $\mathbf{P}$: the tail determines the class; \cref{prop:tail-sandwich,thm:kt-exp,thm:qp-hierarchy,cor:2exp-law}]
\label{thm:summary-tail}
For independent jumps the reach $R_\alpha(q)$ lies between the
$\alpha$-quantile of the largest of $q$ proposals and $q$ times that
quantile. Consequently one fixed computable law has class exactly
$\mathbf{EXP}$; for each $k\ge1$ one fixed computable law has class exactly
$\mathbf{QP}_k=\operatorname{DTIME}(2^{O(\log^k n)})$; and one fixed computable law
has class $\mathbf{2EXP}$. All these classes are independent of the
confidence, and
$\mathbf{P}=\mathbf{QP}_1\subsetneq\mathbf{QP}_2\subsetneq\cdots
\subsetneq\mathbf{EXP}\subsetneq\mathbf{2EXP}$.

\end{introtheorem}

\noindent\textbf{Further results.}
Positive drift and the nondegenerate regular zero-drift regimes give $\mathbf{P}$
(\cref{thm:kt-collapse,prop:kt-zero-drift}); a polynomial forward tail can
overcome negative drift (\cref{prop:kt-polynomial-tail}). Computable iid
laws collectively give $\mathbf{REC}$, and one noncomputable law gives $\mathbf{RE}$
(\cref{cor:computable-union,prop:arbitrary-re}). Common budgets over
parameter sets depend on allowed-set infima (\cref{sec:fuzzy}).
\Cref{fig:landscape,tab:classification} collect the regimes.

\begin{figure}[ht]
\centering
\begin{tikzpicture}[>=Stealth,font=\footnotesize,x=1cm,y=.9cm]
 \node (re) at (0,10) {$\mathbf{RE}$};
 \node (rec) at (0,9) {$\mathbf{REC}$};
 \node (two) at (0,8) {$\mathbf{2EXP}$};
 \node (exp) at (0,7) {$\mathbf{EXP}$};
 \node (dots) at (0,6) {$\vdots$};
 \node (qp) at (0,5) {$\mathbf{QP}_2$};
 \node (p) at (0,4) {$\mathbf{P}=\mathbf{QP}_1$};
 \node[draw,rounded corners,align=center,text width=4.5cm,fill=gray!7] (ac) at (0,2.5)
  {one zero-drift class\\and a negative-drift antichain\\{\scriptsize\cref{thm:zero-drift-counterexample,thm:antichain}}};
 \node (rdef) at (0,1) {$\mathbf{RDEF}$};
 \node (frozen) at (0,0) {$\{\varnothing,\Sigma^*\}$};
 \foreach \lo/\hi in {frozen/rdef,rdef/ac,ac/p,p/qp,qp/dots,dots/exp,exp/two,rec/re}
  {\draw[->] (\lo)--node[right=2pt,inner sep=1pt,font=\scriptsize] {$\subsetneq$} (\hi);}
 \draw[->] (two)--node[right=2pt,inner sep=1pt,font=\scriptsize] {$\subseteq$} (rec);
 \node[right] at (2.8,10) {busy-beaver law};
 \node[right] at (2.8,9) {computable laws (union)};
 \node[right] at (2.8,8) {$K=2^{2^J}$; squaring gates};
 \node[right,align=left] at (2.8,7) {$\mu_{\mathrm{E}}$; parity $\eta_\beta$; gated $\eta^{\mathrm{gate}}_\beta$};
 \node[right] at (2.8,5) {$\mu_{[k]}$};
 \node[right,align=left] at (2.8,4) {positive drift / regular\\nondegenerate zero drift};
 \node[right] at (2.8,2.5) {sparse laws};
 \node[right] at (2.8,1) {negative drift, all positive moments};
 \node[right] at (2.8,0) {frozen law};
 \node[align=right,text width=2.5cm] (forward) at (-3.8,4.4) {forward-only gated\\clocks stay here};
 \draw[->] (forward.east)--(p.west);
\end{tikzpicture}
\caption{The landscape of observation classes. Each sparse class lies
strictly between $\mathbf{RDEF}$ and $\mathbf{P}$; the box asserts no ordering between
the zero-drift class and the negative-drift antichain.}
\label{fig:landscape}
\end{figure}
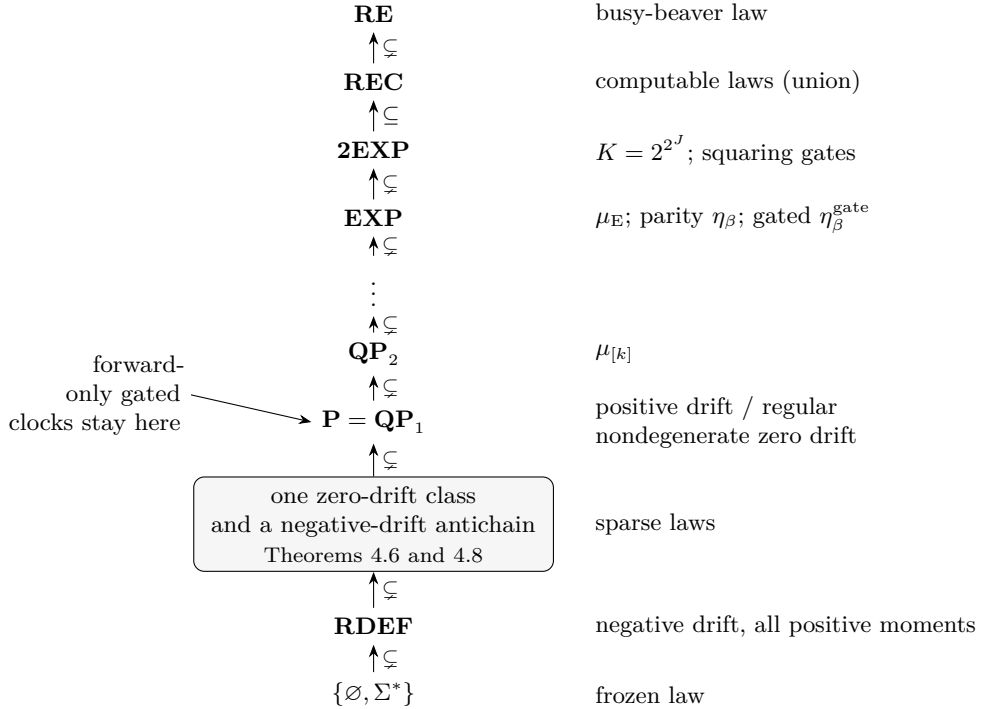

In \cref{tab:classification}, $\delta_a$ denotes the point-mass probability law concentrated at $a$, including $a=0$ and $a=1$. For an increment $K$, its positive and negative parts are $K^+=\max(K,0)$ and $K^-=\max(-K,0)$ (\cref{def:kt-moments}). The deletion kernel $\mathsf D\eta$ replaces backward moves by holds at their starting indices and keeps forward moves unchanged (\cref{def:deletion}).

\begin{table}[ht]
\centering
\footnotesize

\begin{tabularx}{\linewidth}{@{}Xp{.19\linewidth}p{.25\linewidth}@{}}
\toprule
Hypotheses & Class & Reference\\
\midrule
\multicolumn{3}{@{}l}{\emph{Independent jumps}}\\
Some finite positive moment & $\subseteq\mathbf{P}$ & \cref{thm:kt-positive-moment}\\
Integrable, positive mean & $\mathbf{P}$ & \cref{thm:kt-collapse}\\
Integrable zero drift; $\mu\ne\delta_0$; finite variance or bounded positive jumps & $\mathbf{P}$ & \cref{prop:kt-zero-drift}\\
Positive forward probability; $\mathbb{E}K^->\mathbb{E}K^+$; all positive moments finite & $\mathbf{RDEF}$ & \cref{thm:negative-all-moments}\\
Polynomial lower forward tail and some finite positive moment & $\mathbf{P}$ & \cref{prop:kt-polynomial-tail}\\
All positive moments infinite & $\mathbf{P}$ possible & \cref{prop:infinite-moments-p}\\
Computable law & $\subseteq\mathbf{REC}$ & \cref{prop:computable-rec}\\
\midrule
\multicolumn{3}{@{}l}{\emph{Position-dependent kernels}}\\
Monotone forward ability & $\mathcal{C}_{\eta,\alpha}\subseteq\mathcal{C}_{\mathsf D\eta,\alpha}$ & \cref{thm:monotone-rows}\\
Parity clock: no returns / fixed $\beta>0$ & $\mathbf{P}$ / $\mathbf{EXP}$ & \cref{thm:activation,thm:sharp}\\
Doubling gates: forward-only / returns & $\subseteq\mathbf{P}$ / $\mathbf{EXP}$ & \cref{thm:gated}\\
Density-zero backward-index support; return probability $m^{-c}$, $c>0$, at $2^m+1$ & $\mathbf{EXP}$ & \cref{cor:density-zero,prop:vanishing-beta}\\
Squaring gates: forward-only / returns & $\subseteq\mathbf{P}$ / $\mathbf{2EXP}$ & \cref{thm:sq-gated}\\
\bottomrule
\end{tabularx}
\caption{Selected fixed-confidence regimes, supplementing
\cref{fig:landscape}. Each row holds for $0<\alpha<1$;
common-budget results are in \cref{sec:fuzzy}.}
\label{tab:classification}
\label{tab:kernel-classification}
\end{table}

\subsection{Techniques and related work}
\label{sec:related-work}
Holding an inadmissible proposal at its starting position is a proposal-rejection
rule, a standard way to construct Markov kernels
\cite[Section~3.2]{LevinPeres2017}. It differs from Lindley's reflected
update, which resets an overshoot below zero to zero \cite{Lindley1952}.
The held walk dominates the reflected walk but is not monotone in its
starting position (\cref{rem:lindley-comparison}). Reflected-walk upper
bounds therefore require a separate argument for the held rule.

The elementary drift regimes use the strong law, the central limit theorem,
and moment bounds. Foster--Lyapunov methods
\cite[Chapters~11, 14 and~15]{MeynTweedie1993} and exponential drift
analysis \cite{Hajek1982} provide the standard framework for controlling
negative-drift walks. Pemantle--Rosenthal \cite{PemantleRosenthal1999}
prove uniform moment bounds under negative drift, including a one-sided
truncation version in their Corollary~3. Our proof of
\cref{thm:negative-all-moments} gives the held-rule crossing estimate
directly by truncating negative proposals in the conditional drift. It
allows infinite negative mean. The subsequent identification with
bounded-prefix recognition uses sequential locality, related to classical
one-tape arguments \cite{Hennie1965}, rather than a further stability theorem.

The relation between sums, maxima, and slowly varying tails is classical
\cite{Darling1952,Kasahara1986}. The independent-sum limit theorem in
Croydon--Fribergh--Kumagai \cite[Section~5]{CroydonFriberghKumagai2013}
is a relevant probability comparison. Here elementary proposal-maxima
bounds suffice for the explicit tails used in \cref{sec:kt-tails}.
The contribution is their realisation as fixed computable clocks with
specified observation classes. The deterministic time hierarchy theorem
supplies the strict separations between the familiar time classes
\cite[Section~3.1]{AroraBarak2009}.

The sparse constructions address a different realisation problem.
Their explicit laws have crossing radii that alternate between effective
fast and slow stages. Finite rational searches protect all previously
chosen slow stages from future atoms in the zero-drift construction.
A partition of tower scales gives the negative-drift antichain.
Distinguishing input cells supply accepting-time lower bounds, while
specific length languages give witnesses above the bounded-prefix class.
Thus the class comparisons are proved for the constructed laws;
an arbitrary resource-function hierarchy alone does not provide those
probabilistic realisations.

Restart strategies for Las Vegas algorithms use independent new runs
\cite{LSZ1993}. Our phase-reuse argument instead changes accessibility
within the clock on one fixed deterministic path. The deletion coupling
compares forward-displacement laws, and the gates compare all forward-only
members of specified kernel families. In the sparse gated examples,
density and counting orders concern the support of backward indices in
index space. They do not measure the frequency of backward moves along
a trajectory. Vanishing probabilities are specified at the doubling
landing indices (\cref{prop:vanishing-beta}).

The fuzzy-time computation background is described in
\cref{sec:opening-motivation}.

\subsection{Organisation}
\Cref{sec:preliminaries} defines the model and reachability dictionary;
\cref{sec:kt-collapse,sec:kt-taxonomy,sec:kt-tails} classify iid clocks,
and \cref{sec:backward} treats backward moves.
\Cref{sec:discussion} lists open problems. The appendices contain the
sparse constructions, effectiveness proofs, calibrations, and common-budget
results.
\space%
\section{Model and the reachability dictionary}
\label{sec:preliminaries}
\label{ch:kernel-time-computation}

The \emph{jumpy observation model} describes an observer who sees the
configuration at a randomly chosen index, with consecutive indices chosen
by a random clock that may jump. We define the computation path, clock,
and acceptance rule. The access convention and
sequential locality connect crossing estimates to language classes.

\subsection{Machines, languages, and classes}
The machine reads its input sequentially, which constrains clocks with small reach.
We fix the machine conventions and recall the complexity classes used below.

A \emph{language} over a finite alphabet $\Sigma$ is a set $L$ of finite
input words; a machine \emph{decides} $L$ when it halts and accepts
exactly the words of $L$. Polynomial time, written $\mathbf{P}$, is the class
of languages decidable with the number of executed transitions bounded by
a polynomial in the input length. Exponential time is
$\mathbf{EXP}=\bigcup_{k\ge1}\operatorname{DTIME}(2^{n^k})$, with the usual asymptotic
running-time convention. Equivalently, a language has a deterministic
decider with bound $2^{p(n)}$ for one fixed polynomial $p$, after
enlarging the bound to cover every input length. \label{txt:2exp}Doubly exponential time is
$\mathbf{2EXP}=\bigcup_{k\ge1}\operatorname{DTIME}(2^{2^{n^k}})$. As for
$\mathbf{EXP}$, a language belongs to $\mathbf{2EXP}$ exactly when it has a
deterministic decider with bound $2^{2^{p(n)}}$ for one fixed polynomial
$p$. \label{txt:qp}For an integer $k\ge1$ put
$\mathbf{QP}_k=\bigcup_{c\ge1}\operatorname{DTIME}(2^{c(\log_2(n+2))^k})$.
Thus $\mathbf{QP}_1=\mathbf{P}$, and every $\mathbf{QP}_k$ is contained in $\mathbf{EXP}$.
The decidable languages are denoted $\mathbf{REC}$, and the semidecidable (recursively
enumerable) languages, accepted by machines that halt on yes-instances
but need not halt on no-instances, are denoted $\mathbf{RE}$. Thus
$\mathbf{P}\subseteq\mathbf{REC}\subsetneq\mathbf{RE}$. For example, over $\{0,1\}$, the prefix language $101\,\Sigma^*$ is decided
in three transitions, while $L_{\mathrm{even}}=\{x:|x|\text{ is even}\}$
requires a full input scan. The reverse-definite languages $\mathbf{RDEF}$,
defined below, depend on only a bounded initial segment.

\begin{definition}[Machine and logical path]
\label{def:kt-config-path}
Fix a finite nonempty input alphabet $\Sigma$. A machine has a finite tape
alphabet $\Gamma\supseteq\Sigma$ and a blank $b\in\Gamma\setminus\Sigma$.
Its tape has cells $1,2,\ldots$. Initially it contains $x\in\Sigma^*$ in
cells $1,\ldots,|x|$ and blanks elsewhere. The head starts at cell $1$ in a
fixed state $q_0$, independent of $x$. A transition reads and writes the
scanned cell, changes state, and moves the head by $-1$, $0$, or $+1$; a
left move at cell $1$ leaves the head there. Input access is sequential.
There is no input-length oracle or free access to an end marker.

The transition function is deterministic and defined on every nonhalting
state and tape symbol. There are distinct accepting and rejecting halting
states $q_{\rm acc},q_{\rm rej}$. We allow $q_0$ to be either halting state;
this means that the machine halts before reading any symbol.
A configuration records the state, head position, and full tape content.
For each input $x$, let
\[
 \rho_M(x)=(C_0^x,C_1^x,\ldots)
\]
be the logical path obtained by executing transitions and repeating a
halting configuration forever. The first halting index is
$T_M(x)=\inf\{i\ge0:C_i^x\text{ is halting}\}$, with $\inf\varnothing=\infty$.
The first accepting index is $A_M(x)=T_M(x)$ if this halt accepts, and
$A_M(x)=\infty$ otherwise. Thus a rejecting halt has finite $T_M(x)$ but
infinite $A_M(x)$; a nonhalting path has both indices infinite.
\end{definition}

We use $\mathbb{N}=\{0,1,2,\ldots\}$. Budgets are polynomials
$Q:\mathbb{N}\to\mathbb{N}$ with nonnegative integer coefficients, including the zero
polynomial. Every polynomial upper bound has a majorant of this form on $\mathbb{N}$. We make the same enlargement for polynomial exponents in
exponential-time bounds, taking $p(0)\ge1$ when needed. Machine-model
scope is discussed in \cref{rem:kt-model-scope,rem:model-robustness}.

\subsection{The held walk and the access rule}

The held walk defines an \emph{independent-proposal random clock}, or \emph{fuzzy clock}.
One tick supplies a configuration under the following access rule.

\begin{definition}[Law and held walk]
\label{def:kt-time-measure}
\label{def:kt-walk}
Fix one probability law $\mu$ on $\mathbb{Z}$, independently of the input and the
machine. Let $K_1,K_2,\ldots$ be independent with law $\mu$. The logical
position at tick $t$ is $X_t$, where
\begin{equation}
 X_0=0,\qquad
 X_{t+1}=\begin{cases}
 X_t+K_{t+1},&X_t+K_{t+1}\ge0,\\
 X_t,&X_t+K_{t+1}<0.
 \end{cases}
 \label{eq:hold}
\end{equation}
An inadmissible negative jump is \emph{held} at its starting position. This
differs from the reflected update $\max(0,X_t+K_{t+1})$; the held rule is
not generally order-preserving between starting positions. For an integer
target $T\ge0$, put
\[
 \tau_T=\inf\{t\ge0:X_t\ge T\},\qquad \inf\varnothing=\infty.
\]
In the random-walk literature, $\tau_T$ is the first passage time into $[T,\infty)$.
We write $\Pr_i$ for the same held walk started at $i\ge0$, and
omit the subscript when $i=0$.
\end{definition}

\noindent\textbf{Conventions.}
Configurations are supplied at unit access cost, and the observer tests
whether the state accepts. Configuration generation, proposal sampling,
and index writing are excluded from that cost. Laws and parameters are
fixed in advance, with each parameter fixed throughout a run. Witnesses
and budget constants may depend on the language, clock, and confidence.
No effectiveness is assumed unless stated. These observation resources
make no claim of a physical implementation or ordinary machine speedup.

\label{rem:resource-models}
Classical complexity charges transitions; the present access convention charges observations.

\begin{definition}[One-sided acceptance]
\label{def:kt-ac1}
A budget of $Q$ ticks observes the configurations at ticks $0,\ldots,Q$:
there are $Q+1$ observations, including the free initial configuration.
Their indices can repeat. For $0<\alpha<1$, a language
$L\subseteq\Sigma^*$ belongs to $\mathbf P^*_{\mu,\alpha}$ if there are a
deterministic machine $M$ and a polynomial $Q$ such that
\begin{enumerate}[label=(\roman*),leftmargin=*]
\item for every $x\notin L$, $A_M(x)=\infty$;
\item for every $x\in L$, $A_M(x)<\infty$ and
$\Pr(\tau_{A_M(x)}\le Q(|x|))\ge\alpha$.
\end{enumerate}
Set $\mathbf P^*_\mu=\bigcap_{0<\alpha<1}\mathbf P^*_{\mu,\alpha}$.
Its quantifiers are $\forall\alpha\,\exists M\,\exists Q\,\forall x$: each
confidence may receive its own witness, and no single witness is required
to serve all levels.
\end{definition}

\label{rem:confidence-intersection}
Equalities at each fixed confidence pass to the intersection class.

\noindent\textbf{Kernel classes.}
For a time-homogeneous stochastic kernel $\eta$ on $\mathbb{N}$, fixed independently
of the machine and input and started at zero, define $\mathcal{C}_{\eta,\alpha}$
by \cref{def:kt-ac1} with the kernel's hitting probability in place of the
held law's probability. Set
$\mathcal{C}_\eta=\bigcap_{0<\alpha<1}\mathcal{C}_{\eta,\alpha}$.
One tick supplies the configuration at the current index under the same
access convention. For $\eta=\eta_\mu$ in \eqref{eq:kernel}, these classes
are $\mathbf P^*_{\mu,\alpha}$ and $\mathbf P^*_\mu$.

Write $\ell_\alpha=-\ln(1-\alpha)$ for the confidence logarithm; $\ln$ is the natural logarithm.

On an accepting input, the accepting indices form the tail
$[A_M(x),\infty)\cap\mathbb{N}$. Acceptance by tick $q$ records a visit to this
tail at some tick $0\le t\le q$, even if the clock later jumps below
$A_M(x)$. Soundness applies to the entire path of each no-instance,
including nonhalting paths.

\begin{lemma}[Reachable logical time]
\label{lem:reachability}
Fix a held law or a stochastic kernel on $\mathbb{N}$, started at zero, and
$0<\alpha<1$. For every $q\in\mathbb{N}$, the quantity
\[
 R_\alpha(q)=
 \max\{T\in\mathbb{N}:\Pr(\tau_T\le q)\ge\alpha\}
\]
is well defined and finite. It is nondecreasing in $q$, and
$R_\alpha(0)=0$. A language belongs to the fixed-confidence class
exactly when some deterministic recognizer $M$ accepts precisely that
language and some allowed polynomial $Q$ satisfies
\[
 A_M(x)\le R_\alpha(Q(|x|))\qquad(x\in L).
\]
Here $A_M(x)$ is finite on $L$ and infinite off $L$.
\end{lemma}
\begin{proof}
For finite $q$, the running maximum
$Z_q=\max_{0\le t\le q}X_t$ is surely finite: each of the finitely
many positions takes its value in $\mathbb{N}$. Thus
$\Pr(\tau_T\le q)=\Pr(Z_q\ge T)$ decreases to zero as the integer
$T$ increases. Target zero is always eligible. The eligible targets
therefore form a nonempty finite initial segment of $\mathbb{N}$, whose maximum
is $R_\alpha(q)$. Increasing $q$ can only enlarge this segment. At tick
zero the only eligible target is zero.

For a recognizer accepting exactly $L$, one-sided soundness already
holds off $L$. On each yes-instance, the initial-segment property gives
$\Pr(\tau_{A_M(x)}\le Q(|x|))\ge\alpha$ if and only if
$A_M(x)\le R_\alpha(Q(|x|))$.
\end{proof}

For a fixed clock and confidence, define the crossing budget by
\begin{equation}
 q_\alpha(T)=\min\{q\in\mathbb{N}:\Pr_0(\tau_T\le q)\ge\alpha\},
 \qquad \min\varnothing=\infty.
 \label{eq:quantile-def}
\end{equation}
The quantities used throughout the paper have the following meanings.
\begin{center}
\small
\begin{tabularx}{\linewidth}{@{}lX@{}}
\toprule
Term & Meaning\\
\midrule
Threshold $T$ & A logical index, usually $A_M(x)$\\
Crossing time $\tau_T$ & First tick with $X_t\ge T$\\
Crossing probability & $\Pr(\tau_T\le q)$\\
Crossing budget $q_\alpha(T)$ & Fewest ticks giving crossing probability at least $\alpha$\\
Reachable radius $R_\alpha(q)$ & Farthest threshold crossed by $q$ ticks with probability at least $\alpha$\\
\bottomrule
\end{tabularx}
\end{center}
The budget and radius are inverse views of the same crossing probabilities:
$q_\alpha(T)\le q$ exactly when $T\le R_\alpha(q)$.
\Cref{lem:reachability-dictionary} translates radius growth into class bounds.

\begin{definition}[Moments and drift]
\label{def:kt-moments}
For $K\sim\mu$, put $K^+=\max(K,0)$ and $K^-=\max(-K,0)$.
For $r>0$, the positive moment $m_r=\mathbb{E}[(K^+)^r]$ is allowed to be infinite.
We use the finite drift $v=\mathbb{E}K$ only when
$\mathbb{E}|K|=\mathbb{E}K^++\mathbb{E}K^-<\infty$.
Under this condition the variance is $\mathbb{E}[(K-v)^2]$, possibly infinite.
Otherwise an extended mean may be $+\infty$ or $-\infty$; if both parts
are infinite, $\mathbb{E}K$ is undefined. None of these cases is assigned a
finite drift. A positive exponential moment means
$\mathbb{E}e^{\theta K}<\infty$ for some $\theta>0$.
Since $K$ is integer valued, $\Pr(K>0)>0$ holds exactly when
$\mu(k_0)>0$ for some integer $k_0\ge1$.
\end{definition}

\begin{definition}[Computable law]
\label{def:kt-computable-law}
A probability law $\mu$ on $\mathbb{Z}$ is computable if an algorithm, on input
$(k,p)\in\mathbb{Z}\times\mathbb{N}$, returns a rational $d(k,p)$ with
$|d(k,p)-\mu(k)|\le2^{-p}$. The law has total mass one, as usual.
Thus the point probabilities are computable uniformly in $k$, with
certified errors. No bound on the running time of this algorithm is part
of the definition.
\end{definition}

This agrees with the standard notion of a computable probability measure
on the discrete integers; \cref{app:effectiveness} explains the
equivalence and derives an effective tail bound and a decidability result.

\subsection{The transition kernel and two basic bounds}
\label{sec:kt-measure-kernel}
The kernel identifies the held boundary rule, and two bounds compare its reach with free proposals.
A Markov kernel on $\mathbb{N}$ is a matrix of nonnegative transition probabilities
whose rows sum to one. For \eqref{eq:hold} it is
\begin{equation}
 \eta_\mu(i,j)=\begin{cases}
 \mu(j-i),&j\ne i,\\
 \mu(0)+\Pr(K<-i),&j=i.
 \end{cases}
 \label{eq:kernel}
\end{equation}
The row sum is $\sum_{k\ge-i}\mu(k)+\Pr(K<-i)=1$. Off the diagonal, the
transition probability depends only on the displacement; conversely, this
property characterizes the held kernels among stochastic kernels, with a
unique recovering law (\cref{prop:kt-hold-representation}).

\label{app:kernel}
\begin{remark}[Intrinsic representation and uniqueness]
\label{lem:kt-eta-recovery}
\label{prop:kt-hold-representation}
A stochastic kernel $\eta$ on $\mathbb{N}$ is $\eta_\mu$ for a unique probability
law $\mu$ on $\mathbb{Z}$ if and only if there is a function
$\varphi:\mathbb{Z}\setminus\{0\}\to[0,1]$ with
$\eta(i,j)=\varphi(j-i)$ whenever $i\ne j$.
\end{remark}
\begin{proof}
Necessity follows from \eqref{eq:kernel}. For sufficiency, row sums give $\sum_{k\ge-i,k\ne0}\varphi(k)\le1$; monotone convergence gives $\sum_{k\ne0}\varphi(k)\le1$~\cite[Theorem~1.5.7]{Durrett2019}.
Set $\mu(k)=\varphi(k)$ for $k\ne0$ and $\mu(0)=1-\sum_{k\ne0}\varphi(k)$; stochasticity determines $\eta(i,i)=\mu(0)+\sum_{k<-i}\mu(k)$.
For $k\ne0$, any row $i\ge\max(0,-k)$ recovers $\mu(k)=\eta(i,i+k)$, and normalization recovers $\mu(0)$, proving uniqueness; a row $i\ge w$ recovers every mass of a law supported in $[-w,w]$.
\end{proof}

The hold rule never moves farther down than an executed negative proposal.
It also never gains more than the sum of the positive proposals.

\begin{lemma}[Comparison with the free walk]
\label{lem:kt-domination}
With $S_t=\sum_{j=1}^tK_j$, a held walk started at $i\ge0$ satisfies
$X_t\ge i+S_t$ for every $t$. Started at zero, it also satisfies
\begin{equation}
 \max_{0\le s\le q}X_s\le\sum_{j=1}^qK_j^+.
 \label{eq:positive-sum}
\end{equation}
Consequently $\Pr_i(\tau_T\le q)\ge\Pr(S_q\ge T)$ for $i\ge0$.
\end{lemma}
\begin{proof}
For the lower bound we use induction on $t$; the case $t=0$ is $X_0=i=i+S_0$.
If the proposal $K_{t+1}$ is executed, then
$X_{t+1}=X_t+K_{t+1}\ge i+S_t+K_{t+1}=i+S_{t+1}$.
If it is held, then $X_t+K_{t+1}<0\le X_t$, so $K_{t+1}<0$ and
$X_{t+1}=X_t\ge i+S_t> i+S_{t+1}$.

For the upper bound, both cases give $X_{t+1}\le X_t+K^+_{t+1}$. Started at zero,
this yields $X_s\le\sum_{j=1}^s K^+_j$ for every $s$. The right-hand side is
nondecreasing in $s$, which proves \eqref{eq:positive-sum}.

Finally, with common proposals, $S_q\ge T$ implies $X_q\ge i+S_q\ge T$, hence
$\tau_T\le q$. This inclusion of events gives
$\Pr_i(\tau_T\le q)\ge\Pr(S_q\ge T)$.
\end{proof}

\begin{remark}[Comparison with the reflected walk]
\label{rem:lindley-comparison}
Under common proposals and $X_0=W_0=0$, the held walk dominates Lindley's walk $W_{t+1}=\max(0,W_t+K_{t+1})$: if $X_t\ge W_t$, the three cases $K_{t+1}\ge-W_t$, $-X_t\le K_{t+1}<-W_t$, and $K_{t+1}<-X_t$ preserve this order.
Writing $S_0=0$ and $S_t=\sum_{j=1}^tK_j$, this gives $X_t\ge W_t=S_t-\min_{0\le s\le t}S_s$, and reversing the iid increments shows that $W_t$ has the law of $\max_{0\le s\le t}S_s$ at each fixed $t$.
Upper bounds for $W$ therefore do not bound $X$, so \cref{sec:kt-taxonomy} proves held-walk drift estimates directly.
\end{remark}

\begin{lemma}[Elementary proposal-maxima comparison]
\label{lem:proposal-maxima}
For iid held proposals and $q\ge1$, put $J_q=\max_{1\le j\le q}K_j^+$.
For integer $T\ge1$,
\[
 \{J_q\ge T\}\subseteq\{\tau_T\le q\}
 \subseteq\{J_q\ge T/q\}.
\]
The integer quantile
$u_\alpha(q)=\max\{u\in\mathbb{N}:\Pr(J_q\ge u)\ge\alpha\}$ is finite and satisfies
$u_\alpha(q)\le R_\alpha(q)\le q u_\alpha(q)$.
For $q=0$ set $u_\alpha(0)=0$, agreeing with $R_\alpha(0)=0$.
For $q\ge1$ and every real $z>0$, define $\bar F(z)=\Pr(K^+\ge z)$.
Then
\[
 \Pr(J_q\ge z)=1-(1-\bar F(z))^q;
 \qquad \Pr(J_q\ge z)\ge\alpha
 \ \Longleftrightarrow\ \bar F(z)\ge1-(1-\alpha)^{1/q}.
\]
The tail threshold is real; the quantile remains over integer targets.
\end{lemma}
\begin{proof}
A proposal at least $T$ crosses from every nonnegative position.
Conversely, \eqref{eq:positive-sum} gives
$\max_{s\le q}X_s\le\sum_{j\le q}K_j^+\le qJ_q$.
The maximum $J_q$ is finite almost surely, so its eligible integer targets
form a nonempty finite initial segment, as in \cref{lem:reachability}.
If $T$ is reachable at confidence $\alpha$, then
$\lceil T/q\rceil$ is eligible for $J_q$, hence $T\le q u_\alpha(q)$.
The lower quantile bound follows from the first inclusion, also when
$u_\alpha(q)=0$. Independence gives the tail formula and threshold.
\end{proof}

The factor $q$ can matter: for $\delta_1$, $u_\alpha(q)=1$ but
$R_\alpha(q)=q$. This estimate alone does not classify arbitrary clocks.

\begin{lemma}[Supported forward blocks]
\label{lem:kt-block-crossing}
If $p_0=\mu(k_0)>0$ for some $k_0\ge1$, then for $T\ge1$,
$m=\lceil T/k_0\rceil$, $j\ge1$, and every $i\ge0$,
\[
 \Pr_i(\tau_T>jm)\le(1-p_0^m)^j\le e^{-jp_0^m}.
\]
Thus every fixed target has a constant tick budget at any fixed confidence.
\end{lemma}
\begin{proof}
Partition the proposals into $j$ disjoint blocks of length $m$.
The event that all proposals in a block equal $k_0$ has probability $p_0^m$.
The independent-trial formula gives failure probability $(1-p_0^m)^j$. Each
such block reaches $T$ from any nonnegative starting position, since all
its jumps execute and $mk_0\ge T$.
If no hit occurs, all these events fail. For confidence $\alpha$, take
$j=\lceil p_0^{-m}\ell_\alpha\rceil$.
\end{proof}

For nonnegative proposals, equality holds in \eqref{eq:positive-sum}.
For a law on $\{0,1\}$ with $\mu(1)=p>0$ and target $1$, the block
failure probability is $(1-p)^q$ from zero. Larger-target block budgets
can exceed the optimal budget.

\begin{proposition}[Frozen laws]
\label{prop:kt-stuck-walk}
If $\Pr(K>0)=0$, then $X_t=0$ almost surely for every $t$ and
$\mathbf P^*_{\mu,\alpha}=\mathbf P^*_\mu=\{\varnothing,\Sigma^*\}$
for every $0<\alpha<1$.
\end{proposition}
\begin{proof}
Every proposal at zero is either zero or held. A yes-instance must
therefore have $A_M(x)=0$. The initial state is independent of $x$: if it
accepts, soundness forces $L=\Sigma^*$; otherwise completeness forces
$L=\varnothing$. The immediately accepting and immediately rejecting
machines attain both languages.
\end{proof}

\subsection{Sequential locality and the reachability dictionary}
\label{sec:reachability-dictionary}

This is the standard sequential-input locality observation: a head moving
by at most one cell per transition cannot read beyond the first $T$ cells
in $T$ transitions; related crossing-sequence arguments for one-tape
machines go back to Hennie~\cite{Hennie1965}. We give the proof for our
tape convention.

\begin{lemma}[Prefix truncation]
\label{lem:kt-prefix-locality}
For a deterministic machine and an integer $T\ge0$, let $u$ be the prefix
of $x$ of length $\min(|x|,T)$. The runs on $x$ and $u$ have the same
states and head positions through index $T$. If either run first halts
at an index at most $T$, the other first halts at the same index and in
the same state.
\end{lemma}
\begin{proof}
If $|x|\le T$, then $u=x$. Otherwise the initial tapes agree in cells
$1,\ldots,T$. Before transition $t+1$, for $0\le t<T$, the head is at
most at cell $t+1\le T$. Induction shows that the state, head position,
and contents of these cells agree through index $T$. Cell $T+1$ may be
the final head position but is not read during these $T$ transitions.
In the truncated input its blank therefore causes no earlier difference.
The case $T=0$ compares initial states without reading a cell.
Equality of states through $T$ gives the assertion about first halting.
\end{proof}

\begin{lemma}[Sublinear accepting time is constant]
\label{lem:kt-sublinear-constant}
Let $g:\mathbb{N}\to[0,\infty)$ be nondecreasing and satisfy
$g(n)/n\to0$ as $n\to\infty$. Under the sequential tape conventions of
\cref{def:kt-config-path}, suppose $M$ accepts exactly $L$ and
$A_M(x)\le g(|x|)$ for every $x\in L$.
Then $A_M(x)\le t_{\rm loc}$ on $L$ for one integer $t_{\rm loc}$ depending only on $g$.
If $M$ is total and $T_M(x)\le g(|x|)$ on every input, its total
halting time is likewise bounded by $t_{\rm loc}$.
\end{lemma}
\begin{proof}
For $x\in L$ put $T=A_M(x)$ and truncate it to $u$ of length
$\min(|x|,T)$. By \cref{lem:kt-prefix-locality}, $M$ first accepts $u$
at index $T$. Soundness implies $u\in L$, so
\[
 T=A_M(u)\le g(|u|)\le g(T).
\]
Choose an integer $R_{\rm loc}\ge1$ such that $g(t)<t$ for every integer
$t\ge R_{\rm loc}$. Such an $R_{\rm loc}$ exists because $g(t)/t\to0$.
The displayed inequality forces $T<R_{\rm loc}$, so
$t_{\rm loc}=R_{\rm loc}-1$ works for every accepting input.
For a total machine use $T=T_M(x)$; its bound applies to $u$ regardless
of the halting state.
\end{proof}

The proof only needs $g(t)<t$ eventually. If $A_M(x)\le c|x|+b$ on accepting inputs with
$0\le c<1$ and $b\ge0$, prefix truncation gives $T\le cT+b$, hence
$A_M(x)\le\lfloor b/(1-c)\rfloor$. The same holds for total halting
time. At the linear boundary a scan to the first blank takes $|x|+1$
transitions, so unbounded times are possible.

\begin{corollary}[Logarithmic accepting time]
\label{lem:kt-log-constant}
If $M$ accepts exactly $L$ and
$A_M(x)\le c(\log_2(|x|+1)+1)$ on $L$ for a fixed $c>0$,
then its accepting time is bounded by a constant depending only on $c$.
The corresponding statement holds for total halting time bounded by this
expression on every input.
\end{corollary}
\begin{proof}
Apply \cref{lem:kt-sublinear-constant} to
$g(n)=c(\log_2(n+1)+1)$, which is nonnegative, nondecreasing and $o(n)$.
\end{proof}

\begin{definition}[Reverse-definite language]
\label{def:reverse-definite}
A language is reverse definite if, for some fixed integer $k\ge1$,
$A\subseteq\Sigma^k$, and $F\subseteq\Sigma^{<k}$,
\begin{equation}
 L=F\ \cup\ \bigcup_{v\in A}v\Sigma^*.
 \label{eq:reverse-definite}
\end{equation}
Write $\mathbf{RDEF}$ for this class.
Membership depends on the first $k$ symbols,
with shorter words handled separately. This is the classical reverse-definite
class~\cite[Section~4.2]{BrzozowskiLiu2012}. The prefix language
$101\,\Sigma^*$ of \cref{sec:preliminaries} is reverse definite with
$k=3$ and $A=\{101\}$; the even-length language is not, since for any $k$
the words $a^{2k}$ and $a^{2k+1}$ share their length-$k$ prefix and have
opposite membership.
\end{definition}

The original theory of definite events and automata is due to
Brzozowski and to Perles, Rabin, and
Shamir~\cite{Brzozowski1963,PerlesRabinShamir1963}. Definite events depend
on a bounded suffix; reversing words gives the bounded-prefix counterpart
used here.

\begin{lemma}[Bounded accepting time and bounded prefixes]
\label{lem:finite-prefix-class}
A language has a sound recognizer with a constant accepting-time bound
if and only if it is reverse definite. Every reverse-definite language
has a total constant-time decider. If a clock reaches every finite target
at confidence $\alpha$ with some finite budget, then
$\mathbf{RDEF}\subseteq\mathcal{C}_\alpha$, where $\mathcal{C}_\alpha$ denotes its class.
In particular this holds for any iid held law with $\Pr(K>0)>0$.
\end{lemma}
\begin{proof}
Suppose $M$ accepts exactly $L$ by a constant index $t_0$ on all
accepted inputs. Put $k=\max(1,t_0)$. Membership is equivalent to
acceptance by index $k$. By \cref{lem:kt-prefix-locality}, words of
length at least $k$ have the same outcome as their length-$k$ prefixes.
Thus \eqref{eq:reverse-definite} holds with
$A=L\cap\Sigma^k$ and $F=L\cap\Sigma^{<k}$.
Conversely, for fixed $k,A,F$, a finite-control decider reads at most the
first $k$ cells. If it encounters a blank after $j<k$ symbols, it decides
membership in $F$ in $j+1\le k$ transitions. If it reads $k$ symbols,
it decides whether their prefix is in $A$ in $k$ transitions. This also
handles the empty input by reading its first blank. Its time bound is
constant. A finite budget reaching that constant target at confidence
$\alpha$ observes every accepting tail; no-instance soundness is
inherited from the decider. For an iid law,
\cref{lem:kt-block-crossing} supplies the finite-target hypothesis.
\end{proof}

\begin{lemma}[Reachability dictionary]
\label{lem:reachability-dictionary}
Fix a clock on $\mathbb{N}$, started at zero, and a confidence $0<\alpha<1$.
Write $\mathcal{C}_\alpha$ for its class and $R_\alpha(q)$ for
\cref{lem:reachability}'s radius. Then the following hold.
\begin{enumerate}[label=(\roman*),leftmargin=*]
\item If $R_\alpha(q)\le C(q+1)^d$ for finite $C,d>0$, then
$\mathcal{C}_\alpha\subseteq\mathbf{P}$.
\item If $R_\alpha(q)\ge c q^\varepsilon$ eventually, with
$c,\varepsilon>0$, then $\mathbf{P}\subseteq\mathcal{C}_\alpha$.
\item If $1+R_\alpha(q)=O((q+1)^\varepsilon)$ for every
$\varepsilon>0$, and every finite target is reachable at confidence
$\alpha$ with a finite budget, then $\mathcal{C}_\alpha=\mathbf{RDEF}$.
\item If eventually
$2^{c q^\varepsilon}\le R_\alpha(q)\le2^{C(q+1)^d}$ for positive
$c,\varepsilon,C,d$, then $\mathcal{C}_\alpha=\mathbf{EXP}$.
\end{enumerate}
\end{lemma}
\begin{proof}
For a witness $(M,Q)$, \cref{lem:reachability} bounds its accepting
indices by $R_\alpha(Q(n))$. In (i), choose an integer polynomial
majorant for $C(Q(n)+1)^d$. Simulate $M$ through that logical index,
checking index zero, and reject on timeout. Completeness ensures acceptance on every
yes-instance, while soundness excludes acceptance on a no-instance. Reading the input length and polynomial simulation overhead
preserve polynomial time. This gives (i). The upper estimate in (iv)
gives an exponential-time decider in the same way. The finite number of
exceptional budgets is absorbed in the bound using the finite radii.

For (ii), let $H(n)$ be a polynomial logical-time bound of a decider.
Choose an integer polynomial $Q$ of sufficiently large degree and
coefficient that $cQ(n)^\varepsilon\ge H(n)$ and the eventual radius
estimate applies. Enlarge its constant term for the finitely many small
input lengths. The radius tends to infinity, so every fixed target also
has a finite adequate budget. \Cref{lem:reachability} now gives the
inclusion, with the same sound decider. For the lower inclusion in (iv),
use $H(n)=2^{p(n)}$ and choose $Q$ with $cQ(n)^\varepsilon\ge p(n)$.
Any classical machine-model simulation overhead can be absorbed in
$p$. This proves (iv).

For (iii), write $Q(n)\le a(n+1)^d$ with $d\ge1$ and $a\ge1$.
Choose $\varepsilon<1/d$. The subpolynomial estimate gives a
nondecreasing upper bound $g(n)=O((n+1)^{d\varepsilon})=o(n)$ for
$R_\alpha(Q(n))$, after enlarging a constant to cover small budgets.
By \cref{lem:kt-sublinear-constant}, every accepting time is then
bounded by a constant. \Cref{lem:finite-prefix-class} identifies the
language as reverse definite and, using the finite-target hypothesis,
proves the lower inclusion.
\end{proof}

\begin{remark}[Machine-model robustness]
\label{rem:kt-model-scope}
\label{rem:model-robustness}
Prefix truncation holds for finitely many initially blank work tapes and sequential input heads: in $T$ transitions each head reads only the first $T$ input cells.
Thus the $\mathbf{RDEF}$ results and the $n+1$ distinguishing-cell lower bound in \cref{app:zero-drift} remain valid without a length oracle.
The $\mathbf{P}$, $\mathbf{QP}_k$, $\mathbf{EXP}$, $\mathbf{2EXP}$, $\mathbf{REC}$, and $\mathbf{RE}$ classifications are invariant under polynomially related sequential models~\cite{AroraBarak2009}; overheads are absorbed by radius bounds or halting envelopes, and for $\mathbf{QP}_k$ only change the constant.
The sparse law's stage budgets $T_j^j$ dominate every fixed polynomial slowdown, preserving its separation from $L_{\mathrm{even}}$, although its exact class may depend on the machine model.
\end{remark}
\space%
\section{Independent jumps up to P}
\label{sec:kt-collapse}

This section bounds crossing probabilities by forward moments and proves
linear or quadratic crossing budgets under suitable drift assumptions.
The moment bound gives containment in $\mathbf{P}$; the budget estimates
give equality.

\begin{theorem}[Positive-moment containment]
\label{lem:kt-reach}
\label{thm:kt-positive-moment}
If $\mathbb{E}[(K^+)^r]<\infty$ for some $r>0$, then
$\mathbf P^*_{\mu,\alpha}\subseteq\mathbf{P}$ for every $0<\alpha<1$.
More precisely, for $0<r\le1$, $m_r=\mathbb{E}[(K^+)^r]<\infty$,
$T\ge1$, and $q\ge0$,
\begin{equation}
 \Pr(\tau_T\le q)\le\frac{q m_r}{T^r}.
 \label{eq:moment-crossing}
\end{equation}
For such $0<r\le1$, a witness $(M,Q)$ therefore has
$A_M(x)\le(m_r Q(|x|)/\alpha)^{1/r}$ on every yes-instance with
$A_M(x)>0$.
\end{theorem}
\begin{proof}
If a moment of order $r>1$ is finite, the first moment is finite because
$y\le1+y^r$ for $y\ge0$. We may therefore assume $0<r\le1$.
By \eqref{eq:positive-sum}, subadditivity of the $r$th power, and Markov's
inequality~\cite[Theorem~1.6.4]{Durrett2019},
\[
 \Pr(\tau_T\le q)
 \le\frac{\mathbb{E}[(\sum_{j=1}^qK_j^+)^r]}{T^r}
 \le\frac{\sum_{j=1}^q\mathbb{E}[(K_j^+)^r]}{T^r}
 =\frac{q m_r}{T^r}.
\]
Completeness gives the claimed bound
when $A_M(x)>0$. If $m_r=0$,
then $K^+=0$ almost surely, and \cref{prop:kt-stuck-walk} applies.

For $m_r>0$, \eqref{eq:moment-crossing} bounds the radius by
$(m_rq/\alpha)^{1/r}$. Apply
\cref{lem:reachability-dictionary}(i).
\end{proof}

A law with all positive moments infinite can still have class $\mathbf{P}$; see \cref{app:moment-boundary}.

The sharpness of the moment-bound coefficient is illustrated in \cref{app:calibration-1}.
Recall the crossing budget $q_\alpha(T)$ from \eqref{eq:quantile-def}.
All asymptotic comparisons keep the law and confidence fixed and let the
integer target $T$ tend to infinity; their constants may depend on both.
For $\mu=\delta_1$, $\tau_T=T$ from zero,
so the linear scale in the next lemma cannot be improved in general.

\subsection{Positive drift without a variance assumption}
Positive drift makes the free walk advance linearly with high probability.
The comparison lemma transfers that estimate to the held walk.

\begin{lemma}[Linear crossing budget]
\label{lem:kt-crossing}
If $\mathbb{E}|K|<\infty$ and $v=\mathbb{E}K>0$, then for each $0<\alpha<1$ there
is an integer $c_{\alpha,\rm drift}\ge1$ such that, for every $T\ge1$ and $i\ge0$,
\[
 \Pr_i(\tau_T\le c_{\alpha,\rm drift} T)\ge\alpha.
\]
\end{lemma}
\begin{proof}
The strong law of large numbers gives $S_q/q\to v$ almost surely,
and hence in probability~\cite[Theorem~2.4.1]{Durrett2019}.
Its integrability and independence hypotheses hold for the increments.
Choose an integer $N_\alpha\ge1$
such that $\Pr(S_q/q\ge v/2)\ge\alpha$ for every integer
$q\ge N_\alpha$. Choose an integer
$c_{\alpha,\rm drift}\ge\max\{1,N_\alpha,2/v\}$. For $T\ge1$ and
$q=c_{\alpha,\rm drift} T$, we have $q\ge N_\alpha$ and $vq/2\ge T$.
Therefore \cref{lem:kt-domination} gives
$\Pr_i(\tau_T\le q)\ge\Pr(S_q\ge T)\ge\alpha$.
The choice $c_{\alpha,\rm drift}\ge N_\alpha$ covers every small positive target.
\end{proof}

In particular $R_\alpha(q)\ge\lfloor q/c_{\alpha,\rm drift}\rfloor$
for every $q\in\mathbb{N}$, by monotonicity in the budget.

\begin{theorem}[Positive-drift equality]
\label{thm:kt-collapse}
If $\mathbb{E}|K|<\infty$ and $\mathbb{E}K>0$, then
$\mathbf P^*_{\mu,\alpha}=\mathbf P^*_\mu=\mathbf{P}$ for every
$0<\alpha<1$.
\end{theorem}
\begin{proof}
The first positive moment is finite, so
\cref{thm:kt-positive-moment} gives containment in $\mathbf{P}$.
Conversely, take a deterministic polynomial-time decider $M$ with integer
polynomial time bound $H$. For confidence $\alpha$ use
$Q(n)=c_{\alpha,\rm drift}(H(n)+1)$. The preceding lemma reaches $H(n)+1$, hence
every accepting index at most $H(n)$, with probability at least $\alpha$.
The same machine is sound on every no-instance. This proves equality.
\end{proof}

\subsection{Zero drift}
With zero drift and finite nonzero variance, quadratically many ticks give
a fixed positive crossing probability. Repeating blocks raises the
confidence; \cref{app:calibration-2} calibrates the scale and the alternatives.

\begin{lemma}[Quadratic crossing blocks]
\label{lem:kt-zero-crossing}
If $\mathbb{E}K=0$ and $0<\sigma^2=\mathbb{E}K^2<\infty$, there are an integer
$b\ge1$ and a number $p\in(0,1)$ such that for every $T\ge1$,
$i\ge0$, and integer $j\ge1$,
\[
 \Pr_i(\tau_T>jbT^2)\le(1-p)^j.
\]
\end{lemma}
\begin{proof}
Let $\Phi$ be the standard normal distribution function and put
$p=(1-\Phi(1))/2>0$. By the classical iid (Lindeberg--L\'evy) central
limit theorem~\cite[Theorem~3.4.1]{Durrett2019}, there is an integer
$N\ge1$ such that
\[
 \Pr\left(\frac{S_q}{\sigma\sqrt q}\ge1\right)\ge p
 \quad\text{for every }q\ge N.
\]
Choose an integer $b\ge\max\{1,N,\sigma^{-2}\}$. For $q=bT^2$,
$\sigma\sqrt q\ge T$ and $q\ge N$, including when $T=1$.
Thus $\Pr(S_q\ge T)\ge p$ and the free-walk comparison gives
$\Pr_i(\tau_T\le q)\ge p$ from every $i\ge0$.

At the end of a block, condition on the complete past and on not yet
having hit $T$. The current position is below $T$. The next block uses
fresh independent increments, so the same lower bound $p$ applies from
that position. The tower property~\cite[Theorem~4.1.13]{Durrett2019},
applied at these deterministic block
boundaries, proves the display. The block-hit events need not be independent.
The free-walk bound holds from each starting position, so it supplies
the uniform conditional estimate needed here.
\end{proof}

There is also a zero-drift case with infinite variance that follows from a
stopped-square argument. Only positive overshoots need to be bounded.

We write $\mathbb E_i$ for expectation under the law of the held walk started at $X_0=i$, in parallel with the notation $\Pr_i$ for probability. For nonnegative quantities $a$ and $b$, write $a\wedge b=\min\{a,b\}$. Thus $q\wedge\tau_T$ is the number of ticks elapsed until either tick $q$ or the first crossing of $T$, whichever occurs first, and $X_{q\wedge\tau_T}$ is the clock position at that time.

\begin{lemma}[positive jumps at zero drift]
\label{lem:kt-zero-bounded-positive}
Suppose $\mathbb{E}|K|<\infty$, $\mathbb{E}K=0$, $\mu\ne\delta_0$, and
$K\le w$ almost surely for an integer $w\ge1$. Negative jumps may be
unbounded. Then some integer $k_0\ge1$ has $\mu(k_0)>0$. Fix one and put
$\delta_\mu=k_0^2\mu(k_0)>0$. Then, for $0\le i<T$,
\[
 \mathbb{E}_i\tau_T\le\frac{(T-1+w)^2-i^2}{\delta_\mu},\qquad
 \Pr_i(\tau_T>q)\le\frac{(T-1+w)^2-i^2}{\delta_\mu q}
 \quad(T,q\ge1).
\]
Choose an integer
$b_{\rm sq}\ge\max\{1,\lceil2w^2/\delta_\mu\rceil\}$.
Then $\Pr_i(\tau_T>jb_{\rm sq}T^2)\le2^{-j}$ for $j\ge1$,
uniformly below $T$.
\end{lemma}
\begin{proof}
Zero mean and nondegeneracy give a supported $k_0>0$. Set
$\delta_\mu=\mu(k_0)k_0^2$. For a position $i\ge0$, the conditional squared gain is
\[
 \mathbb{E}[X_{t+1}^2-X_t^2\mid X_t=i]
 =\sum_{k\ge-i}(2ik+k^2)\mu(k)\ge \delta_\mu.
\]
Indeed $\sum_{k\ge-i}k\mu(k)=-\sum_{k<-i}k\mu(k)\ge0$, and the
quadratic term includes $\delta_\mu$. These sums are finite because executed jumps
lie in $[-i,w]$. For $\tau=\tau_T$, finite summation of conditional
gains and the tower property give
\[
 \mathbb{E}_i X_{q\wedge\tau}^2-i^2
 \ge \delta_\mu\sum_{s=0}^{q-1}\Pr_i(\tau>s)=\delta_\mu\mathbb{E}_i(q\wedge\tau).
\]
The stopped position is at most $T-1+w$. Monotone convergence as
$q\to\infty$ gives the expectation bound, and Markov's inequality gives
the tail bound. Since $T-1+w\le Tw$ for integers $T,w\ge1$, the expectation
is at most $w^2T^2/\delta_\mu$, uniformly below $T$.
Choose an integer
$b_{\rm sq}\ge\max\{1,\lceil2w^2/\delta_\mu\rceil\}$.
Markov's inequality gives success at least $1/2$ in each block of
$b_{\rm sq}T^2$ ticks, uniformly below the target. Fresh conditional blocks iterate as in \cref{lem:kt-zero-crossing}.
\end{proof}

\begin{theorem}[Zero-drift equality]
\label{prop:kt-zero-drift}
Suppose $\mathbb{E}|K|<\infty$, $\mathbb{E}K=0$, and $\mu\ne\delta_0$.
If either $\mathbb{E}K^2<\infty$ or the positive jumps are bounded, then
$\mathbf P^*_{\mu,\alpha}=\mathbf P^*_\mu=\mathbf{P}$ for every
$0<\alpha<1$.
\end{theorem}
\begin{proof}
Containment follows from the finite first positive moment.
In the finite-variance case, nondegeneracy gives $\sigma^2>0$.
Use \cref{lem:kt-zero-crossing} with block constant $b$ and success
bound $p$. In the bounded-positive case use
\cref{lem:kt-zero-bounded-positive} with block constant $b_{\rm sq}$
and success bound $1/2$. Denote the chosen constant by $b_0$ and its
success bound by $p_0>0$. For a polynomial-time
decider with bound $H$, choose an integer $j_\alpha\ge1$ with
$(1-p_0)^{j_\alpha}\le1-\alpha$. The integer polynomial
$Q(n)=j_\alpha b_0(H(n)+1)^2$ is a valid budget. Soundness is unchanged.
The excluded law $\delta_0$ is covered by \cref{prop:kt-stuck-walk}.
\end{proof}
\space%
\section{Below P: sequential locality, an intermediate class, and an antichain}
\label{sec:kt-taxonomy}

Small crossing radii force bounded-prefix recognition through sequential locality. Negative drift gives this class under all positive moments, while sparse laws give intermediate and incomparable classes.

\subsection{Negative drift gives \texorpdfstring{$\mathbf{RDEF}$}{RDEF}}
\label{sec:negative-all-moments}

Polynomial Lyapunov bounds give a subpolynomial radius under negative mean
and all positive moments~\cite[Chapters~11 and~14]{MeynTweedie1993}.
When $\mathbb{E}K^-=\infty$, the untruncated uniform increment-moment hypothesis
of Pemantle--Rosenthal does not apply; their Corollary~3 permits one-sided
truncation~\cite{PemantleRosenthal1999}. The proof below uses this principle
directly at the level of the held walk's conditional drift. The class
identification then follows from sequential locality.

For an event $E$, $\mathbf{1}_E$ is its indicator, equal to one on $E$ and zero otherwise; for a random variable $Y$, $\mathbb{E}[Y;E]$ denotes $\mathbb{E}[Y\mathbf{1}_E]$ whenever this expectation is defined.

\begin{theorem}[Negative mean and all positive moments]
\label{thm:negative-all-moments}
Suppose $\Pr(K>0)>0$,
$\mathbb{E}[(K^+)^s]<\infty$ for every $s>0$, and
$\mathbb{E}K^->\mathbb{E}K^+$. The last condition permits $\mathbb{E}K^-=\infty$,
hence extended mean $-\infty$. For every integer $r\ge2$ there is a
finite constant $C_r$ such that, from zero,
\begin{equation}
 \Pr(\tau_T\le q)\le C_r\frac{q+1}{(1+T)^{r-1}}
 \qquad(q\in\mathbb{N},\ T\ge1).
 \label{eq:polynomial-drift-crossing}
\end{equation}
Consequently, for every $0<\alpha<1$,
$\mathbf P^*_{\mu,\alpha}=\mathbf P^*_\mu=\mathbf{RDEF}$.
\end{theorem}
\begin{proof}
By monotone convergence~\cite[Theorem~1.5.7]{Durrett2019}, choose a finite integer $m\ge1$ with
\[
 \mathbb{E}[K^-;K^-\le m]>\mathbb{E}K^+.
\]
Replace every proposal $K<-m$ by zero, leaving the other proposals
unchanged, to obtain a law $\nu$. It has integrable strictly negative
mean
$v_\nu=\mathbb{E}K^+-\mathbb{E}[K^-;K^-\le m]<0$ and the same positive moments.
For $V_r(i)=(1+i)^r$, write $\Delta_\mu V_r(i)$ for its conditional
one-step drift at $i$ under the held rule. Every removed negative
proposal contributes either a decrease of $V_r$ when executed or zero
when held. Thus
\begin{equation}
 \Delta_\mu V_r(i)\le\Delta_\nu V_r(i)\qquad(i\ge0).
 \label{eq:truncated-drift-domination}
\end{equation}
Both drifts are finite at each state: executed negative proposals lie
in $[-i,-1]$, and positive proposals have finite $r$th moment.

For every fixed integer $k$, as $i\to\infty$,
\[
 \frac{((1+i+k)^r-(1+i)^r)\mathbf{1}_{\{k\ge-i\}}}
 {(1+i)^{r-1}}\longrightarrow rk.
\]
For $k\ge0$, expansion bounds the quotient by
\[
 \sum_{a=1}^r\binom ra\frac{k^a}{(1+i)^{a-1}}
 \le(2^r-1)(k+k^r),
\]
using $k^a\le k+k^r$ for $1\le a\le r$.
For $k<0$ the mean-value theorem bounds its absolute value by $r|k|$
when executed; otherwise it is zero. These bounds are integrable under
$\nu$, since its negative proposals are bounded and its positive
$r$th moment is finite. Dominated convergence~\cite[Theorem~1.5.8]{Durrett2019} therefore gives
\[
 \frac{\Delta_\nu V_r(i)}{(1+i)^{r-1}}\longrightarrow rv_\nu<0.
\]
Choose $c_r>0$ smaller than $-rv_\nu$. Beyond a finite state there is
then drift at most $-c_r(1+i)^{r-1}$. For the finitely many remaining
states enlarge a finite $B_r\ge0$. By
\eqref{eq:truncated-drift-domination}, for all $i\ge0$,
\begin{equation}
 \Delta_\mu V_r(i)\le B_r-c_r(1+i)^{r-1}.
 \label{eq:polynomial-lyapunov-drift}
\end{equation}
No negative moment of the original law is required.

The pathwise bound $X_t\le\sum_{s=1}^tK_s^+$ in
\cref{lem:kt-domination} ensures $\mathbb{E}V_r(X_t)<\infty$ for every finite
$t$. For example, for $t\ge1$ the inequality
$(\sum_{s=1}^tK_s^+)^r\le t^{r-1}\sum_{s=1}^t(K_s^+)^r$
proves this. Conditional expectations are thus integrable, and summing
\eqref{eq:polynomial-lyapunov-drift} through deterministic time $q+1$
gives
\[
 c_r\sum_{t=0}^q\mathbb{E}(1+X_t)^{r-1}
 \le V_r(0)+B_r(q+1)-\mathbb{E}V_r(X_{q+1})
 \le1+B_r(q+1).
\]
A union bound and Markov's inequality now yield
\[
 \Pr(\tau_T\le q)
 \le\sum_{t=0}^q\Pr(X_t\ge T)
 \le\frac{1+B_r(q+1)}{c_r(1+T)^{r-1}}.
\]
Taking $C_r=(1+B_r)/c_r$ proves
\eqref{eq:polynomial-drift-crossing}, including $q=0$.

At confidence $\alpha$, whenever $R_\alpha(q)\ge1$ the bound gives
$1+R_\alpha(q)\le(C_r(q+1)/\alpha)^{1/(r-1)}$.
The case $R_\alpha(q)=0$ is absorbed by a constant. Since $r$ can be
arbitrarily large, the radius is subpolynomial. Positive supported
blocks reach every fixed target by \cref{lem:kt-block-crossing}.
\Cref{lem:reachability-dictionary}(iii) proves the class equality.
\end{proof}

Exponential contraction gives a sharper logarithmic radius. The held-boundary correction is handled directly, as in geometric drift estimates~\cite[Chapter~15]{MeynTweedie1993,Hajek1982}; reflected-walk estimates~\cite{KieferWolfowitz1956} do not transfer by \cref{rem:lindley-comparison}.

\begin{lemma}[Exponential bound for the held walk]
\label{lem:kt-stationary-dom-general}
Suppose $g(\kappa):=\mathbb{E}e^{\kappa K}<1$ for some $\kappa>0$.
There is a constant $C_\mu\ge1$, depending only on $\mu$ and the chosen $\kappa$, such
that, from zero,
\begin{equation}
 \sup_{t\ge0}\mathbb{E}e^{\kappa X_t}\le C_\mu,
 \qquad
 \Pr(\tau_T\le q)\le C_\mu(q+1)e^{-\kappa T}
 \quad(T,q\in\mathbb{N}).
 \label{eq:exponential-crossing}
\end{equation}
\end{lemma}
\begin{proof}
Put $V(i)=e^{\kappa i}$. The exact conditional identity is
\begin{align*}
 \frac{\mathbb{E}[V(X_{t+1})\mid X_t=i]}{V(i)}
 &=\mathbb{E}[e^{\kappa K};K\ge-i]+\Pr(K<-i)\\
 &=g(\kappa)+\mathbb{E}[1-e^{\kappa K};K<-i]
 \le g(\kappa)+\Pr(K<-i).
\end{align*}
Set $\lambda=(1+g(\kappa))/2<1$. Since $\Pr(K<-i)\to0$, choose an
integer $R_{\rm tail}\ge1$ so that the last expression is at most $\lambda$ for
$i\ge R_{\rm tail}$. For the finitely many $i<R_{\rm tail}$, the conditional expectation is
finite; for example it is at most $e^{\kappa R_{\rm tail}}(g(\kappa)+1)$.
For some finite $b_{\rm bd}\ge0$ we therefore have
\[
 \mathbb{E}[V(X_{t+1})\mid X_t=i]\le\lambda V(i)+b_{\rm bd}
 \quad(i\ge0).
\]
Iterating from $V(0)=1$ gives
$\mathbb{E}V(X_t)\le\lambda^t+b_{\rm bd}\sum_{s=0}^{t-1}\lambda^s
\le1+b_{\rm bd}/(1-\lambda)=:C_\mu$.
Markov's inequality gives $\Pr(X_t\ge T)\le C_\mu e^{-\kappa T}$.
Boole's inequality (the union bound) over $0\le t\le q$ proves the hitting
bound, also at $T=0$ and $q=0$.
\end{proof}

\begin{corollary}[Reverse-definite classification]
\label{thm:kt-reverse-tfae-general}
Suppose $\mathbb{E}e^{\kappa K}<1$ for some $\kappa>0$, and
$\Pr(K>0)>0$. For every $0<\alpha<1$,
\[
 \mathbf P^*_{\mu,\alpha}=\mathbf P^*_\mu=\mathbf{RDEF}.
\]
Equivalently, a language $L$ belongs to this class if and only if some
machine accepts exactly $L$ and has accepting logical time bounded by a
constant on every $x\in L$. Replacing this constant bound by
$c(\log_2(|x|+1)+1)$ for a fixed $c$ gives the same class.
\end{corollary}
\begin{proof}
The exponential bound makes every positive moment finite, and Jensen's inequality gives $\mathbb{E}K<0$ when $\mathbb{E}K^-$ is finite; otherwise $\mathbb{E}K^-=\infty>\mathbb{E}K^+$.
Thus \cref{thm:negative-all-moments} gives the class equality, with the sharper logarithmic radius supplied by \cref{lem:kt-stationary-dom-general}.
\Cref{lem:kt-stationary-dom-general} bounds the radius by
$\kappa^{-1}\ln(C_\mu(q+1)/\alpha)$, which is subpolynomial.
Positive jump probability gives finite-target budgets by
\cref{lem:kt-block-crossing}. Thus
\cref{lem:reachability-dictionary}(iii) proves the equality.
\Cref{lem:finite-prefix-class} identifies constant accepting time with
reverse definiteness. Constant time implies the stated logarithmic bound,
and \cref{lem:kt-log-constant} proves the converse.
\end{proof}

With a positive exponential moment, $(g(\kappa)-1)/\kappa\to\mathbb{E}K<0$ as $\kappa\downarrow0$, including extended mean $-\infty$, by dominated and monotone convergence.
If $g(\kappa_0)=1$ and $K\not\equiv0$, strict Cauchy--Schwarz gives $g(\kappa_0/2)<1$ and forces positive forward probability; frozen laws are covered by \cref{prop:kt-stuck-walk}.
The exponential and non-exponential calibrations are in \cref{app:calibration-3,app:calibration-4}.

\subsection{A polynomial tail can overcome negative drift}
\label{sec:polynomial-tail}

Drift alone does not classify laws with unbounded support. A sufficiently
large positive proposal reaches a target from any nonnegative position,
regardless of the earlier negative proposals.

\begin{proposition}[Polynomial lower bounds on the positive tail]
\label{prop:kt-polynomial-tail}
Suppose there are $a,c>0$ such that
$\Pr(K\ge T)\ge cT^{-a}$ for all sufficiently large integers $T$.
Then $\mathbf{P}\subseteq\mathbf P^*_{\mu,\alpha}$ for every $0<\alpha<1$.
If also $\mathbb{E}[(K^+)^r]<\infty$ for some $r>0$, equality holds at every
confidence and for $\mathbf P^*_\mu$.
\end{proposition}
\begin{proof}
The hypothesis gives unbounded positive jumps, so by decreasing $c>0$
we can make the bound hold for every integer $T\ge1$. If any of $q$
proposals is at least $T$, the walk hits $T$ from any nonnegative start.
Independence therefore gives
\[
 \Pr_i(\tau_T>q)\le(1-\Pr(K\ge T))^q\le e^{-qcT^{-a}}.
\]
The budget $\lceil c^{-1}T^a \ell_\alpha\rceil$ suffices.
Substituting the logical-time bound of a polynomial-time decider and
choosing a dominating integer polynomial proves the inclusion. Target
zero requires no ticks. The moment assumption gives the other inclusion
by \cref{thm:kt-positive-moment}.
\end{proof}

If $\Pr(K\ge T)$ is bounded above and below by positive
constant multiples of $T^{-a}$ for some fixed $a>0$, the class is
$\mathbf{P}$. For $0<r<\min(a,1)$, the tail-sum identity,
an application of Tonelli's theorem to nonnegative summands
\cite[Theorem~1.7.2]{Durrett2019}, is
\[
 \mathbb{E}[(K^+)^r]
 =\sum_{m=1}^\infty\bigl(m^r-(m-1)^r\bigr)\Pr(K\ge m).
\]
Since $m^r-(m-1)^r=O(m^{r-1})$, the moment is finite.
The same conclusion applies to any tail asymptotic to $cT^{-a}$.
A power-law positive tail with infinite mean therefore does not by itself
produce more than $\mathbf{P}$.

In the following example, $\zeta(2)=\sum_{m=1}^{\infty}m^{-2}=\pi^2/6$ is the Riemann zeta value at $2$.

\begin{example}[Negative drift with class $\mathbf{P}$]
\label{ex:kt-negative-polynomial-tail}
Consider the law
\[
 \mu(-1)=\frac34,\qquad
 \mu(j)=\frac14\left(\frac1{j^2}-\frac1{(j+1)^2}\right)
 \quad(j\ge1).
\]
The positive masses telescope to $1/4$, and for every integer $T\ge1$,
\[
 \Pr(K\ge T)=\frac1{4T^2},\qquad
 \mathbb{E}K^+=\sum_{T\ge1}\Pr(K\ge T)=\frac{\zeta(2)}4.
\]
Thus $\mathbb{E}|K|=(3+\zeta(2))/4<\infty$ and
$\mathbb{E}K=(\zeta(2)-3)/4<0$ (for example $\zeta(2)<2$ by integral comparison).
The first positive moment is finite. By \cref{prop:kt-polynomial-tail},
$\mathbf P^*_{\mu,\alpha}=\mathbf P^*_\mu=\mathbf{P}$ at every confidence.
This law has no finite second positive moment. Negative mean alone
therefore cannot replace the tail hypotheses in either
\cref{thm:kt-reverse-tfae-general} or \cref{thm:negative-all-moments}. Reversing the roles in the first mixture of
\cref{ex:zero-drift-calibration} gives another computable zero-mean law,
$K=W$ with probability $1/(1+m)$ and $K=-1$ otherwise: it has infinite
variance and unbounded positive jumps, yet its tail
$\Pr(K\ge T)=T^{-2}/(1+m)$ and finite first moment give class $\mathbf{P}$
by \cref{prop:kt-polynomial-tail}. Together with
\cref{thm:zero-drift-counterexample}, this shows that the zero-drift
alternatives are sufficient conditions, not a necessary characterization.
\end{example}

See \cref{app:calibration-5} for the order of the one-jump budget.

\subsection{A zero-drift class strictly between \texorpdfstring{$\mathbf{RDEF}$ and $\mathbf{P}$}{RDEF and P}}
\label{sec:zero-drift-statement}

A sparse zero-drift law lies strictly between bounded-prefix recognition and polynomial time. Its accepting-time radius oscillates between growth scales, and the proof is in \cref{app:zero-drift}.

Fix a finite nonempty alphabet $\Sigma$ and let $L_{\mathrm{even}}=\{x\in\Sigma^*:|x|\text{ is even}\}$.
The next theorem shows that zero drift alone does not give $\mathbf{P}$.

\begin{theorem}[A strictly intermediate zero-drift class]
\label{thm:zero-drift-counterexample}
There is one computable law $\mu$, independent of the finite nonempty
input alphabet $\Sigma$, supported on $\{-1,0\}\cup\{N_j:j\ge1\}$, for strictly increasing positive integers
$N_j$, such that
\[
 \mathbb{E}|K|=1,\qquad \mathbb{E}K=0,\qquad \mathbb{E}K^2=\infty,
 \qquad
 L_{\mathrm{even}}\notin\mathbf P^*_{\mu,\alpha}
 \quad(0<\alpha<1).
\]
For every finite nonempty $\Sigma$ and every $0<\alpha<1$,
$\mathbf{RDEF}\subsetneq\mathbf P^*_{\mu,\alpha}\subsetneq\mathbf{P}$, and likewise
$\mathbf{RDEF}\subsetneq\mathbf P^*_\mu\subsetneq\mathbf{P}$.
\end{theorem}

\begin{remark}[Oscillating radius]
\label{rem:sparse-oscillation}
The radius of this law oscillates between growth scales. At
$q_j=T_j^j$, \eqref{eq:zero-subsequence-crossing} gives
$R_\alpha(q_j)<T_j=q_j^{1/j}$ for large $j$. Along
$q=\lceil2\ell_\alpha N_j^{3/2}\rceil$,
\eqref{eq:sparse-large-radius} gives $R_\alpha(q)\ge c_\alpha q^{2/3}$
for a constant $c_\alpha>0$. Thus neither the eventual positive-power
lower bound nor the subpolynomial upper bound in
\cref{lem:reachability-dictionary} holds. Its polynomial upper-bound
criterion still applies, but none of its exact-class criteria does;
the exact class remains open.
\end{remark}
\space%
\subsection{Infinitely many incomparable classes}
\label{sec:antichain}

\Cref{thm:zero-drift-counterexample} gives one class strictly between
$\mathbf{RDEF}$ and $\mathbf{P}$. Sparse negative-drift laws give infinitely many,
pairwise incomparable.

Let $\mathrm{tw}_1=32$ and $\mathrm{tw}_{j+1}=2^{\mathrm{tw}_j}$.
For $e\in\mathbb{N}$ let $J_e=\{2^e(2m+1):m\in\mathbb{N}\}$; these sets
partition $\{1,2,\dots\}$. Define the law $\mu^{\langle e\rangle}$ by
\[
 \begin{split}
 \mu^{\langle e\rangle}(-1)&=\tfrac12,\qquad
 \mu^{\langle e\rangle}(\mathrm{tw}_j^3)=\mathrm{tw}_j^{-4}
 \quad(j\in J_e),\\
 \mu^{\langle e\rangle}(0)&=\tfrac12-
 \sum_{j\in J_e}\mathrm{tw}_j^{-4},
 \end{split}
\]
with zero mass elsewhere. For an integer $j_0\ge1$ let
\[
 L^{\langle e\rangle}_{j_0}
 =\{x\in\Sigma^*:|x|=\mathrm{tw}_j
 \text{ for some }j\in J_e\text{ with }j\ge j_0\}.
\]

The sparse laws separate from one another at every pair of confidences.

\begin{theorem}[An antichain of classes below $\mathbf{P}$]
\label{thm:antichain}
Each $\mu^{\langle e\rangle}$ is computable, has $\mathbb{E}K<0$, and,
for $s>0$, has $\mathbb{E}[(K^+)^s]<\infty$ exactly when $s<4/3$.
Fix a finite nonempty alphabet $\Sigma$. There is $j_*$ such that
for every $j_0\ge j_*$, every $e\ne e'$, and all
$\alpha,\alpha'\in(0,1)$,
\[
 \begin{split}
 L^{\langle e\rangle}_{j_0}
 &\in\mathbf P^*_{\mu^{\langle e\rangle}}
 \subseteq\mathbf P^*_{\mu^{\langle e\rangle},\alpha},\\
 L^{\langle e\rangle}_{j_0}
 &\notin\mathbf P^*_{\mu^{\langle e'\rangle},\alpha'}.
 \end{split}
\]
Consequently, for any confidences, the classes
$\mathbf P^*_{\mu^{\langle e\rangle},\alpha}$, $e\in\mathbb{N}$, are
pairwise incomparable, and each satisfies
\[
 \mathbf{RDEF}\subsetneq\mathbf P^*_{\mu^{\langle e\rangle},\alpha}
 \subsetneq\mathbf{P}.
\]
The same holds for the classes $\mathbf P^*_{\mu^{\langle e\rangle}}$.
\end{theorem}

These laws have negative mean and finite positive moments of every
order below $4/3$, a regime that \cref{thm:negative-all-moments}
(all positive moments finite, class $\mathbf{RDEF}$) and
\cref{ex:kt-negative-polynomial-tail} (class $\mathbf{P}$) leave open.
Fixed-law classes below $\mathbf{P}$ therefore do not form a chain.
The proof is in \cref{app:antichain}.
\space%

\subsection{Classes below P need negative proposals}
Nontrivial iid classes below $\mathbf{P}$ require negative proposals. The deletion comparison explains their role.

\begin{remark}[Negative proposals for classes below $\mathbf{P}$]
\label{rem:below-P-negative}
If $\Pr(K<0)=0$ and $\Pr(K>0)=p>0$, then $X_q=S_q$ is at
least the number of positive proposals among the first $q$, which is
$\mathrm{Bin}(q,p)$. Hence
$R_\alpha(q)\ge\lfloor pq/2\rfloor$ for all large $q$, and
$\mathbf{P}\subseteq\mathbf P^*_{\mu,\alpha}$ by
\cref{lem:reachability-dictionary}(ii).
Negative proposals never enlarge an iid class (\cref{prop:deletion}),
but an iid law with positive forward probability and a class strictly
below $\mathbf{P}$ requires them. This applies to the laws of
\cref{thm:zero-drift-counterexample,thm:negative-all-moments,thm:antichain}.
\end{remark}

\space%
\section{Above P: the tail determines the class}
\label{sec:kt-tails}

Tail quantiles organise the iid classes above $\mathbf{P}$. We first give the tail sandwich, then construct exponential, quasipolynomial and doubly exponential classes and identify the effectiveness boundary.

\subsection{Tail sandwich}
\label{sec:tail-sandwich}

For a nondecreasing function $f:\mathbb{N}\to\mathbb{N}$ let $\mathcal T(f)$ be the
set of languages $L$ for which some deterministic machine $M$ accepts
exactly $L$ and some allowed polynomial $Q$ satisfies
$A_M(x)\le f(Q(|x|))$ for every $x\in L$. By \cref{lem:reachability},
$\mathbf P^*_{\mu,\alpha}=\mathcal T(R_\alpha)$.

A proposal-tail quantile controls the class up to the multiplicative tick budget.

\begin{proposition}[Tail sandwich]
\label{prop:tail-sandwich}
For every iid law $\mu$ and $0<\alpha<1$, with $u_\alpha$ as in
\cref{lem:proposal-maxima},
\[
 \mathcal T(u_\alpha)\subseteq\mathbf P^*_{\mu,\alpha}
 \subseteq\mathcal T(q\mapsto q\,u_\alpha(q)).
\]
\end{proposition}

\begin{proof}
\Cref{lem:proposal-maxima} gives
$u_\alpha(q)\le R_\alpha(q)\le q\,u_\alpha(q)$ for $q\ge1$,
and all three vanish at $q=0$. The function $u_\alpha$ is
nondecreasing in $q$, because $\Pr(J_q\ge u)$ is nondecreasing in $q$.
Apply $\mathcal T$, which is monotone in $f$.
\end{proof}

Thus the class of an iid law is determined, up to the factor $q$,
by the quantiles of a single proposal tail. The next laws exploit this.

\subsection{One law with class \texorpdfstring{$\mathbf{EXP}$}{EXP}}
A tail of order $1/\log T$ crosses exponentially distant indices with a polynomial budget. The same tail bounds the class from above.

\begin{definition}[The exponential-time law]
\label{def:kt-exp-law}
Let $\mu_{\mathrm{E}}(2^j)=1/(j(j+1))$ for integers $j\ge1$, with zero mass
elsewhere. For a real $z>0$, write
$h(z)=\max\{1,\lceil\log_2 z\rceil\}$.
\end{definition}

The law is normalized because
$\sum_{j=1}^N1/(j(j+1))=1-1/(N+1)$.
Its exact tail, including $0<z\le2$, is
\begin{equation}
 \Pr(K\ge z)=\sum_{j\ge h(z)}\frac1{j(j+1)}=\frac1{h(z)}.
 \label{eq:exp-tail}
\end{equation}
The ceiling matters at powers of two. Every positive
moment is infinite, since $2^{rj}/(j(j+1))$ does not tend to zero for
$r>0$. The extended mean is $+\infty$; there is no finite drift.

\begin{lemma}[Two bounds for the exponential-time law]
\label{lem:kt-exp-crossing}
For $\mu=\mu_{\mathrm{E}}$, $T\ge1$, and integer $q\ge1$,
\begin{equation}
 \Pr(\tau_T>q)\le e^{-q/h(T)},\qquad
 \Pr(\tau_T\le q)\le\frac{q}{h(T/q)}.
 \label{eq:exp-crossing}
\end{equation}
If $\Pr(\tau_T\le q)\ge\alpha>0$, then
\begin{equation}
 T\le q\,2^{q/\alpha}.
 \label{eq:exp-extraction}
\end{equation}
For $q=0$, positive confidence is possible only at $T=0$.
\end{lemma}
\begin{proof}
Apply \cref{lem:proposal-maxima} with the tail \eqref{eq:exp-tail}.
The first inclusion and $1-u\le e^{-u}$ give the failure bound;
the second inclusion and a union bound give the success bound.
At positive confidence the latter
implies $h(T/q)\le q/\alpha$. Since
$\log_2(T/q)\le h(T/q)$ for every positive $T/q$, exponentiation gives
\eqref{eq:exp-extraction}.
\end{proof}

For a polynomial budget $Q$, \eqref{eq:exp-extraction} bounds each
accepting index by $Q(n)\,2^{Q(n)/\alpha}$. Conversely, the failure bound
gives exponential reachability: if $k_\alpha\ge\max\{1,\ell_\alpha\}$ is an integer,
then $R_\alpha(q)\ge2^{\lfloor q/k_\alpha\rfloor}$ for $q\ge k_\alpha$,
while \eqref{eq:exp-extraction} gives $R_\alpha(q)\le q2^{q/\alpha}$
for $q\ge1$.

\begin{proposition}[Sharp fixed-confidence scale for $\mu_{\mathrm{E}}$]
\label{prop:exp-sharp-scale}
For $T,q\ge1$,
\begin{equation}
 1-\left(1-\frac1{h(T)}\right)^q
 \le\Pr(\tau_T\le q)
 \le1-\left(1-\frac1{h(T/q)}\right)^q.
 \label{eq:exp-sandwich}
\end{equation}
For every fixed $c>0$ and $0<\alpha<1$, as $T\to\infty$,
\[
 \Pr(\tau_T\le\lfloor c\log_2T\rfloor)\longrightarrow1-e^{-c},
 \qquad
 \frac{q_\alpha(T)}{\log_2T}\longrightarrow \ell_\alpha.
\]
\end{proposition}
\begin{proof}
\Cref{lem:proposal-maxima} and \eqref{eq:exp-tail} give the sandwich,
including $T\le2$ and $T/q\le2$. For $\ell_T=\log_2T$ and $q=\lfloor c\ell_T\rfloor$,
we have
$h(T)=\ell_T+O(1)$ and $h(T/q)=\ell_T-\log_2q+O(1)$.
Since $\log_2q=O(\log \ell_T)$, both $h(T)/\ell_T$ and $h(T/q)/\ell_T$
tend to $1$.
Thus both failure probabilities tend to $e^{-c}$.
Fix
$0<\varepsilon<\ell_\alpha$. The limits at $\ell_\alpha-\varepsilon$ and
$\ell_\alpha+\varepsilon$ put $q_\alpha(T)$ between the corresponding
integer budgets for all sufficiently large $T$; divide by $\ell_T$ and let
$\varepsilon\downarrow0$.
\end{proof}

The logarithmic crossing budget gives exactly exponential-time recognition.

\begin{theorem}[Exact exponential-time class]
\label{thm:kt-exp}
For every $0<\alpha<1$,
\[
 \mathbf P^*_{\mu_{\mathrm{E}},\alpha}=\mathbf P^*_{\mu_{\mathrm{E}}}=\mathbf{EXP}.
\]
\end{theorem}
\begin{proof}
The bounds immediately preceding \cref{prop:exp-sharp-scale} give
an eventual exponential lower radius bound and
$R_\alpha(q)\le q2^{q/\alpha}$ for $q\ge1$. The latter is bounded by
$2^{C(q+1)}$ for a fixed $C$. Apply
\cref{lem:reachability-dictionary}(iv), separately at each confidence.
\end{proof}

The law $\mu_{\mathrm{E}}$ has computable rational masses and an effective tail. Exact sampling and its infinite expected writing cost are discussed in \cref{app:calibration-6}.

\subsection{The hierarchy \texorpdfstring{$\mathbf{QP}_k$ and $\mathbf{2EXP}$}{QP and 2EXP}}
\label{sec:hierarchy}

Varying one tail produces infinitely many fixed-law classes. A further computable tail reaches doubly exponential time.

For functions of the target $T$, write $f(T)\asymp g(T)$ as $T\to\infty$ if $c\,g(T)\le f(T)\le C\,g(T)$ for some constants $0<c\le C<\infty$ and all sufficiently large $T$. The constants are independent of $T$ and may depend on fixed parameters unless uniformity is expressly asserted.

\begin{definition}[Hierarchy laws]
\label{def:hierarchy-laws}
For an integer $k\ge1$ let $J$ satisfy $\Pr(J=j)=2^{-j}$ for
$j\ge1$, and let $\mu_{[k]}$ be the law of $K=2^{J^k}$.
Put $h_k(T)=1$ for $T\le2$ and
$h_k(T)=\lceil(\log_2T)^{1/k}\rceil$ for $T>2$.
Then $\Pr(K\ge T)=\Pr(J\ge h_k(T))=2^{1-h_k(T)}$, and for
every $T\ge1$
\begin{equation}
 2^{-(\log_2T)^{1/k}}\le\Pr(K\ge T)
 \le2^{1-(\log_2T)^{1/k}}.
 \label{eq:hk-tail}
\end{equation}
The point masses are computable rational numbers, so $\mu_{[k]}$ is
computable. The law $\mu_{[1]}$ has
$\Pr(K\ge T)\asymp T^{-1}$ and class $\mathbf{P}$ by
\cref{prop:kt-polynomial-tail}; its positive moments of orders
$0<r<1$ are finite.
\end{definition}

Each hierarchy law gives exactly one quasipolynomial-time level, independently of confidence.

\begin{theorem}[A hierarchy of fixed-law classes]
\label{thm:qp-hierarchy}
For every integer $k\ge1$ and every $0<\alpha<1$,
\[
 \mathbf P^*_{\mu_{[k]},\alpha}=\mathbf P^*_{\mu_{[k]}}
 =\mathbf{QP}_k.
\]
Consequently
$\mathbf{P}=\mathbf{QP}_1\subsetneq\mathbf{QP}_2\subsetneq\cdots
\subsetneq\mathbf{EXP}$: computable iid laws realise infinitely many distinct
fixed classes between $\mathbf{P}$ and $\mathbf{EXP}$, none of which depends on
the confidence.
\end{theorem}

\begin{proof}
\emph{Quantile bounds.} By \cref{lem:proposal-maxima}, an integer
$u\ge1$ is eligible for $u_\alpha(q)$ exactly when
$\Pr(K\ge u)\ge p_q:=1-(1-\alpha)^{1/q}$.
Bernoulli's inequality and $1-e^{-x}\le x$ give
$\alpha/q\le p_q\le\ell_\alpha/q$.
If $q\ge\ell_\alpha$, the integer
$u=\lfloor2^{(\log_2(q/\ell_\alpha))^k}\rfloor\ge1$ satisfies
$(\log_2u)^{1/k}\le\log_2(q/\ell_\alpha)$, hence
$\Pr(K\ge u)\ge\ell_\alpha/q\ge p_q$ by \eqref{eq:hk-tail}.
Conversely, an eligible $u\ge1$ satisfies
$2^{1-(\log_2u)^{1/k}}\ge\alpha/q$, so
$\log_2u\le(\log_2(2q/\alpha))^k$. Hence
\begin{equation}
 \begin{split}
 \bigl\lfloor2^{(\log_2(q/\ell_\alpha))^k}\bigr\rfloor
 &\le R_\alpha(q)\le q\,2^{(\log_2(2q/\alpha))^k}\qquad
 (q\ge\ell_\alpha,\ q\ge1).
 \end{split}
 \label{eq:qp-radius}
\end{equation}
The upper bound holds for every $q\ge1$.

\emph{Upper inclusion.} Let $(M,Q)$ be a witness and choose integers
$a,d\ge1$ with $Q(n)\le a(n+2)^d$. Put
$c_1=d+\max\{0,\log_2(2a/\alpha)\}$.
Since $\log_2(n+2)\ge1$, the upper bound gives, whenever $Q(n)\ge1$,
\[
 \begin{split}
 A_M(x)&\le Q(n)\,2^{(\log_2(2Q(n)/\alpha))^k}\\
 &\le2^{c_1\log_2(n+2)}\,2^{(c_1\log_2(n+2))^k}\\
 &\le2^{2c_1^k(\log_2(n+2))^k}
 \end{split}
\]
for $x\in L$ of length $n$; when $Q(n)=0$, $A_M(x)=0$.
A decider computes $\lceil\log_2(n+2)\rceil$, which is at most
$2\log_2(n+2)$, and simulates $M$ through
$2^{c_2\lceil\log_2(n+2)\rceil^k}$ steps for a fixed integer
$c_2\ge2c_1^k$. It accepts exactly when that simulation accepts,
including at index zero. Soundness makes it a decider for $L$, and its
running time on the model of \cref{def:kt-config-path} is polynomial
in that bound. Hence $L\in\mathbf{QP}_k$.

\emph{Lower inclusion.} Let a decider run within
$2^{c(\log_2(n+2))^k}$ transitions, with $c\ge1$ an integer.
Put $Q(n)=\lceil\ell_\alpha\rceil(n+2)^c$, an allowed polynomial.
Then $Q(n)\ge\ell_\alpha$ and
$\log_2(Q(n)/\ell_\alpha)\ge c\log_2(n+2)$, so \eqref{eq:qp-radius}
gives
\[
 \begin{split}
 R_\alpha(Q(n))
 &\ge\lfloor2^{c^k(\log_2(n+2))^k}\rfloor\\
 &\ge\lfloor2^{c(\log_2(n+2))^k}\rfloor,
 \end{split}
\]
which bounds every accepting index. \Cref{lem:reachability} gives
$L\in\mathbf P^*_{\mu_{[k]},\alpha}$, with a witness depending on
$\alpha$ only through $Q$.

\emph{Strictness.} Use integer time bounds
\[
 B(n)=\lceil\log_2(n+2)\rceil,\qquad
 f_k(n)=2^{B(n)^k\lceil\sqrt{B(n)}\rceil}.
\]
These functions and their squares are time-constructible on multitape
machines: their integer exponents are computable in polynomial time in
$n+2$, and a binary counter supplies a clock in $O(f_k(n))$ time.
Every language in $\mathbf{QP}_k$ lies in $\operatorname{DTIME}(f_k)$, also on
multitape machines. Since $f_k\log f_k=o(f_k^2)$, the multitape time
hierarchy theorem~\cite[Section~3.1]{AroraBarak2009} gives a language
in $\operatorname{DTIME}(f_k^2)\setminus\operatorname{DTIME}(f_k)$.
The quadratic single-tape simulation places it in
$\operatorname{DTIME}(f_k^4)\subseteq\mathbf{QP}_{k+1}$, since
$\log_2 f_k(n)=\Theta((\log_2(n+2))^{k+1/2})$.
The same argument with $f(n)=2^n$ separates every $\mathbf{QP}_k$
from $\mathbf{EXP}$.
\end{proof}

\begin{corollary}[A fixed law with class $\mathbf{2EXP}$]
\label{cor:2exp-law}
Let $J$ have the law \eqref{eq:heavy},
$\Pr(J=j)=1/(j(j+1))$, and let $\mu$ be the law of $K=2^{2^J}$.
This law is computable, and
$\mathbf P^*_{\mu,\alpha}=\mathbf P^*_\mu=\mathbf{2EXP}$
for every $0<\alpha<1$.
\end{corollary}

\begin{proof}
Here $\Pr(K\ge T)=1/h(\log_2T)$ for $T>1$, with $h$ from
\cref{def:kt-exp-law}. The increasing atoms and their rational masses
are computable. If $q\ge\ell_\alpha$, $q\ge1$, and
$m=\lfloor q/\ell_\alpha\rfloor$, the target $2^{2^m}$ has tail
$1/m\ge\ell_\alpha/q$, so
$R_\alpha(q)\ge2^{2^{\lfloor q/\ell_\alpha\rfloor}}$.
Conversely, for $q\ge1$, an eligible $u\ge2$ has
$\log_2\log_2u\le h(\log_2u)\le q/\alpha$, so
$R_\alpha(q)\le q\,2^{2^{q/\alpha}}$ by \cref{lem:proposal-maxima}.
For a decider within $2^{2^{p(n)}}$ transitions, take an allowed
integer polynomial $p\ge1$ and budget
$Q(n)=\lceil\ell_\alpha\rceil p(n)$; its radius reaches that bound.
Conversely, simulating a witness through
$Q(n)2^{2^{Q(n)/\alpha}}$ steps, with $1/\alpha$ replaced by a
larger integer and index zero checked separately, gives a doubly
exponential-time decider. Polynomial simulation overhead is absorbed
in the polynomial exponent. Equality at each confidence also gives
the intersection equality.
\end{proof}

\space%

\subsection{The effectiveness boundary}
\label{sec:effectiveness-statement}

Computability of the law bounds its class by decidability. We state the
individual-law results here and prove them in \cref{app:effectiveness}.

\begin{proposition}[Decidability for computable laws]
\label{prop:computable-rec}
If $\mu$ is computable in the sense of \cref{def:kt-computable-law}, then
$\mathbf P^*_{\mu,\alpha}\subseteq\mathbf{REC}$ for every fixed
$0<\alpha<1$, where $\mathbf{REC}$ denotes the decidable languages.
No moment, drift, or sampling-cost assumption is used.
\end{proposition}

Computable masses yield a computable cutoff $b_{\rm tail}(q)$ for each
finite budget $q$. Making $\Pr(K>b_{\rm tail}(q))$ sufficiently small
excludes targets beyond $qb_{\rm tail}(q)$ at the given confidence. A
decider simulates the witnessing machine through that logical-time bound.

\begin{corollary}[Computable clocks collectively capture decidability]
\label{cor:computable-union}
For every fixed finite nonempty input alphabet and every $0<\alpha<1$,
\[
 \bigcup_{\mu\ \mathrm{computable}}\mathbf P^*_{\mu,\alpha}=\mathbf{REC},
 \qquad
 \bigcup_{\mu\ \mathrm{computable}}\mathbf P^*_\mu=\mathbf{REC}.
\]
The witnessing law may depend on the language.
\end{corollary}

\begin{proposition}[An arbitrary law attaining all semidecidable languages]
\label{prop:arbitrary-re}
For every fixed finite nonempty input alphabet there is one law $\mu$
on the positive integers such that, for every $0<\alpha<1$,
\[
 \mathbf P^*_{\mu,\alpha}=\mathbf P^*_\mu=\mathbf{RE}.
\]
This law is not computable.
\end{proposition}

The law places mass $1/(j(j+1))$ on a busy-beaver-type halting-time
envelope $(\mathcal E_j)$. Its tail satisfies
$\Pr(K\ge\mathcal E_j)=1/j$, so a single proposal reaches or exceeds
$\mathcal E_j$ with probability $1/j$. The law is noncomputable by
\cref{prop:computable-rec}, since $\mathbf{RE}$ contains undecidable languages.
\space%
\section{Backward moves}
\label{sec:backward}

Deleting backward moves never worsens reachability for iid clocks (\cref{prop:deletion}) or for kernels with monotone forward ability (\cref{thm:monotone-rows}). A length-one return can reopen a heavy-sampling phase and change the class from polynomial to exponential time. In the gated family, backward moves are necessary to exceed $\mathbf{P}$ (\cref{sec:gated}).

We use the kernel classes defined after \cref{def:kt-ac1}, with the
clock started at zero. Superscripts on probabilities and expectations
specify a kernel parameter when needed; subscripts specify the starting
index. For an iid law, we also write
$\mathcal{C}_{\mu,\alpha}=\mathcal{C}_{\eta_\mu,\alpha}=\mathbf P^*_{\mu,\alpha}$.

\subsection{Deletion and monotone forward ability}
Deleting backward moves preserves the forward-displacement law at each index. A monotone ordering of these laws couples the deleted clock above the original.

\begin{definition}[Deletion of backward moves]
\label{def:deletion}
For a stochastic kernel $\eta$ on $\mathbb{N}$, define $\mathsf D\eta$ by
$(\mathsf D\eta)(i,j)=\eta(i,j)$ for $j>i$,
$(\mathsf D\eta)(i,i)=\sum_{j\le i}\eta(i,j)$, and
$(\mathsf D\eta)(i,j)=0$ for $j<i$.
We call $\eta$ \emph{forward-only} if $\mathsf D\eta=\eta$.
\end{definition}

\begin{definition}[Forward ability]
\label{def:forward-ability}
For a stochastic kernel $\eta$ on $\mathbb{N}$ and $i\in\mathbb{N}$, the
\emph{forward-displacement law} $\varphi^\eta_i$ on $\mathbb{N}$ is
\[
 \varphi^\eta_i(k)=\eta(i,i+k)\quad(k\ge1),\qquad
 \varphi^\eta_i(0)=\sum_{j\le i}\eta(i,j).
\]
Thus $(\mathsf D\eta)(i,i+k)=\varphi^\eta_i(k)$ for every $k\ge0$,
and $\eta$ and $\mathsf D\eta$ have the same forward-displacement laws.
We say that $\eta$ has \emph{monotone forward ability} if
$\varphi^\eta_i\preceq_{\mathrm{st}}\varphi^\eta_{i'}$ whenever $i\le i'$,
that is, $\varphi^\eta_i([k,\infty))\le\varphi^\eta_{i'}([k,\infty))$
for all $k\in\mathbb{N}$.
\end{definition}

Monotone forward ability ensures that returns cannot improve reachability.

\begin{theorem}[Backward moves need non-monotone forward ability]
\label{thm:monotone-rows}
If $\eta$ has monotone forward ability, then for all $q,T\in\mathbb{N}$,
with both clocks started at zero,
\[
 \Pr_\eta(\tau_T\le q)\le\Pr_{\mathsf D\eta}(\tau_T\le q).
\]
Consequently $\mathcal{C}_{\eta,\alpha}\subseteq\mathcal{C}_{\mathsf D\eta,\alpha}$
and the crossing budget of $\eta$ is at least that of $\mathsf D\eta$,
for every $0<\alpha<1$. If every member of a family $\mathcal F$ has
monotone forward ability, the same inclusion holds for common budgets
over $\mathcal F$ and over $\{\mathsf D\eta:\eta\in\mathcal F\}$.
\end{theorem}

\begin{proof}
Write $\varphi_i=\varphi^\eta_i$ and let
$F_i^{-1}(u)=\min\{k:\varphi_i([0,k])\ge u\}$ for $0<u<1$.
Monotone forward ability gives $F_i^{-1}(u)\le F_{i'}^{-1}(u)$
for $i\le i'$. Let $U_1,U_2,\dots$ and $V_1,V_2,\dots$ be
independent uniform variables on $(0,1)$. Put $X_0=Y_0=0$ and
$Y_{t+1}=Y_t+F_{Y_t}^{-1}(U_{t+1})$.
Let $k=F_{X_t}^{-1}(U_{t+1})$. If $k\ge1$ put $X_{t+1}=X_t+k$.
If $k=0$, choose $X_{t+1}\in\{0,\dots,X_t\}$ using $V_{t+1}$,
with probabilities $\eta(X_t,j)/\varphi_{X_t}(0)$; this case has
probability zero when $\varphi_{X_t}(0)=0$. Then $X$ is a Markov
chain with kernel $\eta$ and $Y$ one with kernel $\mathsf D\eta$.

We show $X_t\le Y_t$ by induction. In both cases
$X_{t+1}\le X_t+F_{X_t}^{-1}(U_{t+1})$, so
\[
 X_{t+1}\le X_t+F_{X_t}^{-1}(U_{t+1})
 \le Y_t+F_{Y_t}^{-1}(U_{t+1})=Y_{t+1}.
\]
The chain $Y$ is nondecreasing, so $\max_{s\le t}X_s\le Y_t$ and
$\{\tau_T(X)\le q\}\subseteq\{\tau_T(Y)\le q\}$.
The class and budget statements follow from \cref{lem:reachability}
with the same witnesses, and the common-budget statement follows by
taking infima over $\mathcal F$.
\end{proof}

\begin{corollary}[Independent-proposal deletion principle]
\label{prop:deletion}
For every independent-proposal law $\mu$ and every $0<\alpha<1$,
\[
 \mathcal{C}_{\mu,\alpha}\subseteq\mathcal{C}_{\mu^+,\alpha},
\]
where $\mu^+$ is the law of $K^+=\max(K,0)$.
Consequently, if deleting the negative proposals already yields a class
contained in $\mathbf{P}$, adding those iid negative proposals cannot produce
a larger class.
\end{corollary}
\begin{proof}
For held kernels, $\varphi^{\eta_\mu}_i$ is the law of $K^+$ at every position and $\mathsf D\eta_\mu=\eta_{\mu^+}$, so \cref{thm:monotone-rows} applies.
Alternatively, use the same proposals for both clocks. The negative-suppressed position
is $X_t^+=\sum_{s\le t}K_s^+$. Each actual increment of the held walk is at
most $K_s^+$. By the pathwise bound of \cref{lem:kt-domination},
$X_t\le X_t^+$ for every $t$. Hence $\tau_T(X^+)\le\tau_T(X)$
pathwise. The same recognizing machine and budget therefore work for the
negative-suppressed clock.
\end{proof}

Both constructions below violate forward monotonicity: the parity rows alternate a heavy law with bounded steps, while a gate's doubling exceeds the next row's bound once $g+1>w$.
The next proposition shows that bounded positive displacements cannot raise the class beyond $\mathbf{P}$.

\begin{proposition}[Bounded forward steps cannot reach beyond $\mathbf{P}$]
\label{prop:bounded-forward}
For any kernel, position-dependent or not, whose positive displacements
are all at most a fixed $w$, $\mathcal{C}_{\eta,\alpha}\subseteq\mathbf{P}$.
\end{proposition}
\begin{proof}
For every path, $\max_{t\le q}X_t\le wq$. A witness can therefore accept
a yes-instance only at a logical index at most $wQ(|x|)$. Simulate its
deterministic path for this many transitions and reject on timeout.
Soundness makes this a polynomial-time decider.
\end{proof}

\subsection{Phase reuse in the parity clock}
A return by one restores the phase that samples the heavy tail. Reusing that phase changes the class from polynomial to exponential time.
\label{sec:activation}

Let $Y=2^J$, where
\begin{equation}
 \Pr(J=m)=\frac1{m(m+1)},\qquad m=1,2,\ldots.
 \label{eq:heavy}
\end{equation}
Recall $h(z)=\max\{1,\lceil\log_2 z\rceil\}$ from
\cref{def:kt-exp-law}. Telescoping gives, for every real $z>0$,
\begin{equation}
 \Pr(Y\ge z)=\frac1{h(z)}.
 \label{eq:tail}
\end{equation}
This is the positive-jump law $\mu_{\mathrm{E}}$ of \cref{def:kt-exp-law}.

Fix $\beta\in[0,1/2]$. Define a time-homogeneous Markov kernel $\eta_\beta$
on $\mathbb{N}$ by the complete transition table in \cref{tab:kernel}.

\begin{table}[ht]
\centering
\caption{Transition probabilities of the activation kernel $\eta_\beta$}
\label{tab:kernel}
\begin{tabularx}{\textwidth}{@{}p{0.26\textwidth}Xp{0.28\textwidth}@{}}
\toprule
Current index $i$ & Next index & Probability\\
\midrule
$i=0$ & $1$ & $1$\\[3pt]
$i\ge1$ odd & $i+1+2^m$, each $m\ge1$ & $1/[m(m+1)]$\\[3pt]
$i\ge2$ even & $i+2$ & $1/2$\\
 & $i-1$ & $\beta$\\
 & $i$ & $1/2-\beta$\\
\bottomrule
\end{tabularx}
\end{table}

In this notation $\eta_0=\mathsf D\eta_\beta$.

The kernel is fixed independently of the input and machine. Its entries
are rational for rational $\beta$ and computable for computable $\beta$.
At positive indices, its rows depend only on parity. The deletion
$\eta_0$ replaces each return by a hold.
\Cref{prop:kt-hold-representation} identifies this parity-dependent
kernel as outside the iid held family: its odd and even rows have
different probabilities for the same positive displacements.
\Cref{prop:deletion} explains why merely adding iid negative proposals
cannot produce this increase.

\begin{figure}[ht]
\centering
\includegraphics[width=\linewidth]{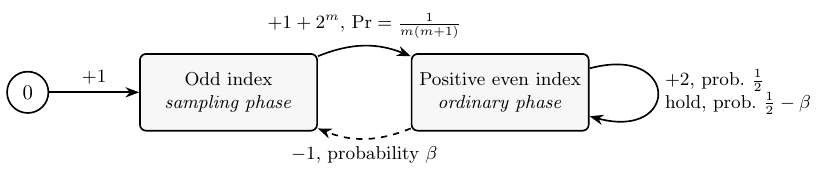}
\caption{The two phases of the clock $\eta_\beta$. From an odd index the
clock samples $Y=2^J$ with law $\mu_{\mathrm{E}}$ and moves by $1+Y$; from a positive even index it
moves forward by two or, with probability $\beta$, returns by one. A
backward step (dashed) restores the phase in which a fresh heavy jump is
available; the index itself keeps growing after each completed cycle}
\label{fig:phases}
\end{figure}

At an odd index a fresh heavy jump lands at an even index; a return by one makes that sample available again. From any state below $T$, two ticks give probability at least $\beta/h(T)$ of crossing, while the return and the following jump's $+1$ cancel.

A length-one return changes the parity clock's class from $\mathbf{P}$ to $\mathbf{EXP}$.

\begin{theorem}[Backward-enabled class transition]
\label{thm:activation}
For every fixed $0<\alpha<1$ and every fixed $\beta\in(0,1/2]$,
\[
 \mathcal{C}_{\eta_0,\alpha}=\mathbf{P},
 \qquad
 \mathcal{C}_{\eta_\beta,\alpha}=\mathbf{EXP}.
\]
The two clocks differ only in that $\eta_\beta$ may step back by one
from a positive even index, giving a strict increase from $\mathbf{P}$ to
$\mathbf{EXP}$.
At $\beta=0$ only the first odd sampling state is visited, after which
the walk stays even. At $\beta>0$ length-one returns repeatedly reopen
the odd sampling phase. The increase is phase reuse through accessibility.
\end{theorem}

\begin{remark}[Forward re-entry in the parity family]
\label{rem:forward-reentry}
\Cref{thm:activation} compares $\eta_\beta$ with its deletion
$\mathsf D\eta_\beta=\eta_0$. It does not show that backward displacement
is needed. Replace each move $i\to i-1$ from a positive even index by
$i\to i+1$. The resulting kernel is forward-only, and it still enters
an odd sampling index with probability $\beta$ at each even-state tick.
The proofs of \cref{lem:amplify,lem:ceiling} apply with only notational
changes: from an even index the step to $i+1$ followed by a jump
$2^m\ge T$ crosses $T$, and every positive displacement is still at most
$1+Y_t$. Its class is therefore $\mathbf{EXP}$ as well. \Cref{sec:gated} gives a
family in which backward displacement is necessary.
\end{remark}

The deterministic time hierarchy theorem gives strictness. The proof
uses the access rule for a fixed kernel and applies to every real $\beta$
in the parameter interval.

\begin{lemma}
\label{lem:baseline}
For $\eta_0$, a single copy of $Y$ can be chosen such that
\[
 \max_{t\le q}X_t\le2q+Y\qquad(q\ge0).
\]
Moreover $\mathbb{E}_i\tau_T\le2T$ for every $0\le i<T$.
\end{lemma}
\begin{proof}
Starting at zero, the first two positions are $1$ and $2+Y$; all later
increments are $0$ or $2$. The maximum bound follows, with an unused
independent $Y$ if $q<2$.

For the expectation bound, couple each tick with an independent fair
Bernoulli variable $\xi_t$. At even indices let $\xi_t$ select the $+2$ move.
At odd indices and at zero the mandatory positive increment is at least
$\xi_t$. The clock's position is therefore at least $\sum_{s\le t}\xi_s$. The
time to $T$ Bernoulli successes has mean $2T$ and bounds $\tau_T$.
This argument also works from every starting index below $T$.
\end{proof}

\begin{proof}[Proof of \cref{thm:activation} for $\beta=0$]
Fix $\alpha$. Choose a constant $b_{\alpha,\rm tail}=2^r$ with
$\Pr(Y>b_{\alpha,\rm tail})=1/(r+1)<\alpha$. If $T>2q+b_{\alpha,\rm tail}$, then
\[
 \Pr^{(0)}(\tau_T\le q)\le\Pr(Y>b_{\alpha,\rm tail})<\alpha.
\]
An accepting index for a witness of budget $Q$ is consequently at most
$2Q(|x|)+b_{\alpha,\rm tail}$. Simulating the witnessing machine up to this
polynomial bound decides the language. This proves
$\mathcal{C}_{\eta_0,\alpha}\subseteq\mathbf{P}$.

Conversely, let a polynomial-time decider have logical-time bound $H(n)$.
\Cref{lem:baseline} and Markov's inequality give
$\Pr(\tau_T>q)\le2T/q$. Taking an integer constant
$k_{\alpha,\rm lin}\ge2/(1-\alpha)$ and budget
$k_{\alpha,\rm lin}(H(n)+1)$ supplies completeness; soundness is inherited from the
decider. Hence $\mathbf{P}\subseteq\mathcal{C}_{\eta_0,\alpha}$.
\end{proof}

\begin{lemma}[Uniform two-tick amplification]
\label{lem:amplify}
Let $\beta>0$, $T\ge1$, and $0\le i<T$. For every integer $r\ge0$,
\begin{equation}
 \Pr_i(\tau_T>2r)
 \le\left(1-\frac\beta{h(T)}\right)^r
 \le\exp\left(-\frac{\beta r}{h(T)}\right).
 \label{eq:amplify}
\end{equation}
In particular, with $q_\alpha(T)$ as in \eqref{eq:quantile-def},
\begin{equation}
 q_\alpha(T)\le
 2\left\lceil\frac{h(T)}\beta \ell_\alpha\right\rceil.
 \label{eq:enabledquant}
\end{equation}
\end{lemma}
\begin{proof}
At an odd $i<T$, a heavy jump with $Y\ge T-i-1$ crosses the target on the
next tick; since $T-i-1\le T$, this has probability at least $1/h(T)$
(and probability one when $T-i-1\le0$). At
a positive even $i<T$, the two-event sequence ``backward by one, then a
heavy jump'' has probability $\beta/h(T)$ and crosses within two ticks:
from $i-1$ (odd) the jump $2^m\ge T$ reaches at least $i+2^m\ge T$. At
zero the first tick goes to $1$ and the second offers a heavy jump, again
with probability $1/h(T)$. Thus from every state below $T$ the next two
ticks cross with probability at least $\beta/h(T)$.

Condition successively on the past at the deterministic boundaries
$0,2,\ldots,2r$. On failure the position is still below $T$, so the same
conditional bound applies. Iteration gives the first inequality;
$1-u\le e^{-u}$ gives the second. Solving for failure probability
$1-\alpha$ yields \eqref{eq:enabledquant}.
\end{proof}

For a decider with logical bound $H(n)=2^{p(n)}$, take an integer
polynomial $p\ge1$ and an integer
$k_{\alpha,\beta}\ge\max\{1,\beta^{-1}\ell_\alpha\}$.
The allowed budget $Q(n)=2k_{\alpha,\beta}p(n)$ contains
$r=k_{\alpha,\beta}p(n)$ two-tick blocks. Since $h(H(n))=p(n)$,
\eqref{eq:amplify} bounds failure by $e^{-\beta k_{\alpha,\beta}}\le1-\alpha$.
The same machine is sound, so $\mathbf{EXP}\subseteq\mathcal{C}_{\eta_\beta,\alpha}$.

\begin{lemma}[Exponential ceiling]
\label{lem:ceiling}
For every $\beta\in[0,1/2]$, $q\ge1$, and $T>q$,
\begin{equation}
 \Pr^{(\beta)}(\tau_T\le q)
 \le \frac{q}{h\!\left(\frac{T-q}{q}\right)}.
 \label{eq:ceiling}
\end{equation}
Thus success probability at least $\alpha$ implies
\begin{equation}
 T\le q\bigl(1+2^{q/\alpha}\bigr).
 \label{eq:logicalbound}
\end{equation}
\end{lemma}
\begin{proof}
Choose iid copies $Y_1,\ldots,Y_q$, one for each tick, using a copy for
the actual heavy jump only when the current index is odd. Every positive
displacement is at most $1+Y_t$: this covers the initial $+1$, the
even-state $+2$, and the odd-state $1+2^m$. Consequently
\[
 \max_{t\le q}X_t\le q+\sum_{t=1}^qY_t.
\]
If the target is reached, some $Y_t\ge(T-q)/q$; otherwise the sum is less
than $T-q$. The union bound and \eqref{eq:tail} prove \eqref{eq:ceiling}.
Success at level $\alpha$ forces
$h((T-q)/q)\le q/\alpha$, hence $(T-q)/q\le2^{q/\alpha}$.
If $T\le q$, \eqref{eq:logicalbound} also holds.
\end{proof}

\begin{proof}[Completion of \cref{thm:activation}]
For a witness of polynomial budget $Q$, every accepting index is bounded
by $Q(n)(1+2^{Q(n)/\alpha})$ when $Q(n)\ge1$. Replace $1/\alpha$ by any
larger hardcoded integer and simulate the deterministic path through that
bound, rejecting on timeout. This is an exponential-time decider, by
soundness. Thus $\mathcal{C}_{\eta_\beta,\alpha}\subseteq\mathbf{EXP}$; the preceding
amplification argument proves the opposite inclusion for $\beta>0$.
\end{proof}

\label{sec:sharp}

Recall $q_\alpha(T)$ from \eqref{eq:quantile-def}; we write
$q^{(\beta)}_\alpha$ for the clock $\eta_\beta$.

The same phase reuse gives exact leading constants for the crossing budget.

\begin{theorem}[Sharp linear and logarithmic scales]
\label{thm:sharp}
Fix $0<\alpha<1$. As integer $T\to\infty$,
\[
 \frac{q^{(0)}_\alpha(T)}{T}\longrightarrow1,
 \qquad
 \frac{q^{(\beta)}_\alpha(T)}{\log_2T}
 \longrightarrow
 \left(1+\frac1\beta\right)\ell_\alpha
 \quad(\beta>0).
\]
\end{theorem}
\begin{proof}
For $\beta=0$ and $q\ge2$, the clock is nondecreasing and
\[
 X_q=2+Y+2F_{q-2},\qquad F_{q-2}\sim\operatorname{Bin}(q-2,1/2).
\]
For $q=\lfloor cT\rfloor$, $Y/T\to0$ in probability and
$2F_{q-2}/T\to c$ in probability. Thus the crossing probability tends to
zero if $c<1$, and to one if $c>1$. Squeezing the fixed-confidence
quantile between these budgets gives the first limit.

For $\beta>0$, the first heavy jump completes at tick 2. After the
$j$th heavy jump, the number $G_j$ of even-state ticks up to and
including the next backward move has distribution $\operatorname{Geom}(\beta)$ on
$\{1,2,\ldots\}$. One further tick completes the next heavy jump.
The cycle lengths $\varrho_j=G_j+1$ are iid, with mean
$d_\beta=1+\beta^{-1}$. They can be generated independently of the
samples $Y_1,Y_2,\ldots$: every even-state row has the same
probabilities, and each heavy displacement changes an odd index to an
even one regardless of its magnitude.

Let $N_q$ count heavy samples completed by tick $q$. It is zero for
$q<2$, and for $q\ge2$,
\[
 N_q=1+\max\left\{m\ge0:\sum_{j=1}^m\varrho_j\le q-2\right\}.
\]
With $S_m^{\varrho}=\sum_{j=1}^m\varrho_j$ and $m=N_q-1$, one has
$S_m^{\varrho}\le q-2<S_{m+1}^{\varrho}$. The strong law for the integrable cycle
lengths~\cite[Theorem~2.4.1]{Durrett2019} implies
$S_m^{\varrho}/m\to d_\beta$ and hence $N_q/q\to1/d_\beta$ almost surely.
This count is independent of all heavy samples.

At each completed-cycle endpoint the return $-1$ and the next forward
$+1$ cancel. The position is $2+\sum_{r\le N_q}Y_r$ plus the
completed even-state $+2$ steps. At a pending odd state it is one less
than the preceding even position. A completed sample $Y_r\ge T$
therefore forces a hit, while the running maximum is at most
$2q+\sum_{r\le N_q}Y_r$. Since $N_q\le q$, a hit at $T>2q$
forces some completed sample to be at least $(T-2q)/q$. Independence
of $N_q$ and the samples gives, for $T>2q$ and $q\ge1$,
\begin{equation}
 1-\mathbb{E}\left(1-\frac1{h(T)}\right)^{N_q}
 \le\Pr(\tau_T\le q)
 \le1-\mathbb{E}\left(1-\frac1{h((T-2q)/q)}\right)^{N_q}.
 \label{eq:sharp-sandwich}
\end{equation}
When $N_q=0$ the conditional failure probability is one, including the
convention $0^0=1$.

Put $L_T=\log_2T$ and $q=\lfloor cL_T\rfloor$, $c>0$ fixed.
For large $T$ the sandwich applies, and both denominators $H_T=h(T)$
and $H_T=h((T-2q)/q)$ satisfy $H_T/L_T\to1$: subtracting $2q$
and dividing by $q$ changes the logarithm by only $O(\log L_T)$.
For either denominator, $N_q/L_T\to c/d_\beta$ in probability and
$N_q/H_T^2\le q/H_T^2\to0$. Since $\ln$ is natural,
\[
 N_q\ln(1-1/H_T)=-N_q/H_T+O(N_q/H_T^2)
 \ \longrightarrow\ -c/d_\beta
\]
in probability. The powers lie in $[0,1]$. For any error margin
$\delta>0$, their expected distance from $e^{-c/d_\beta}$ is at most
$\delta+\Pr(|(1-1/H_T)^{N_q}-e^{-c/d_\beta}|>\delta)$.
It tends to at most $\delta$, and then to zero as $\delta\downarrow0$.
Thus both bounds converge and
\[
 \Pr(\tau_T\le\lfloor c\log_2T\rfloor)\longrightarrow1-e^{-c/d_\beta}.
\]
The target logarithm is base two, while $\ell_\alpha=-\ln(1-\alpha)$
uses the natural logarithm. Set $c_*=d_\beta\ell_\alpha$ and fix
$0<\varepsilon<c_*$. The limiting probabilities at
$c_*-\varepsilon$ and $c_*+\varepsilon$ have strictly positive gaps
below and above $\alpha$. Choose probability-error margins smaller than
half those gaps before taking $T$ large. The actual probabilities are
then respectively below and above $\alpha$. Monotonicity in the integer
budget places $q_\alpha^{(\beta)}(T)$ between the corresponding floor
budgets; dividing by $L_T$ removes their at-most-one rounding errors.
Letting $\varepsilon\downarrow0$ proves the second limit.
\end{proof}

\subsection{Backward steps are necessary in a gated family}
\label{sec:gated}

Sparse gates provide the only large forward moves.
We now give a family of position-dependent clocks in which no
forward-only member exceeds $\mathbf{P}$, however its bounded forward moves
are arranged, while steps back by one reach $\mathbf{EXP}$.

\begin{definition}[Gated family]
\label{def:gated}
Let $\Lambda=\{2^k:k\ge1\}$ be the set of \emph{gates}. For an integer
$w\ge1$, let $\mathcal K_w$ be the set of time-homogeneous stochastic
kernels $\eta$ on $\mathbb{N}$ such that
\begin{enumerate}[label=(\alph*),leftmargin=*]
\item $\eta(g,2g+1)=1$ for every gate $g\in\Lambda$;
\item for every $i\notin\Lambda$, $\eta(i,j)=0$ whenever $j>i+w$.
\end{enumerate}
\end{definition}

At a non-gate index the row may depend on $i$ arbitrarily and may contain
backward moves of any length; only its forward displacement is bounded.
A gate doubles the index and lands one step beyond the next gate $2g$.
The forward-only members of $\mathcal K_w$ include every kernel obtained
from a member by deleting its backward moves, or by replacing them with
forward moves of length at most $w$.

For $\beta\in(0,1]$, define
$\eta^{\mathrm{gate}}_\beta\in\mathcal K_1$ as follows. Put
$\eta^{\mathrm{gate}}_\beta(0,1)=\eta^{\mathrm{gate}}_\beta(1,2)=1$ and
$\eta^{\mathrm{gate}}_\beta(g,2g+1)=1$ for $g\in\Lambda$. At each landing
index $s_m=2^m+1$ with $m\ge2$,
\[
 \eta^{\mathrm{gate}}_\beta(s_m,s_m-1)=\beta,\qquad
 \eta^{\mathrm{gate}}_\beta(s_m,s_m)=1-\beta,
\]
and $\eta^{\mathrm{gate}}_\beta(i,i+1)=1$ at every other index. Its only
backward moves are steps back by one, from $s_m$ to the gate $2^m$.

In this family, every forward-only clock stays polynomial, and a step back accesses exponential time.

\begin{theorem}[Backward steps are necessary in the gated family]
\label{thm:gated}
Fix an integer $w\ge1$ and $0<\alpha<1$.
\begin{enumerate}[label=(\roman*),leftmargin=*]
\item If $\eta\in\mathcal K_w$ is forward-only, then surely
$\max_{t\le q}X_t\le3(w+1)q$ for every $q\in\mathbb{N}$. Hence
$q_\alpha(T)\ge T/(3(w+1))$ for every $T\ge1$, and
$\mathcal{C}_{\eta,\alpha}\subseteq\mathbf{P}$. Equality holds for the member of
$\mathcal K_1$ that moves by $+1$ at every non-gate index.
\item Every $\eta\in\mathcal K_w$ satisfies $X_t\le w(2^t-1)$ surely,
so $\mathcal{C}_{\eta,\alpha}\subseteq\mathbf{EXP}$.
\item For every $\beta\in(0,1]$,
$\mathcal{C}_{\eta^{\mathrm{gate}}_\beta,\alpha}=\mathbf{EXP}$, and the crossing budget
of $\eta^{\mathrm{gate}}_\beta$ satisfies
\[
 \frac{q_\alpha(T)}{\log_2T}\longrightarrow1+\frac1\beta
 \qquad(T\to\infty).
\]
\end{enumerate}
In particular, $\mathsf D\eta^{\mathrm{gate}}_\beta$ and the kernel
obtained from $\eta^{\mathrm{gate}}_\beta$ by replacing each step
$s_m\to s_m-1$ with $s_m\to s_m+1$ are forward-only members of
$\mathcal K_1$, so their classes are contained in $\mathbf{P}$. Within
$\mathcal K_w$, a class not contained in $\mathbf{P}$ requires backward
moves, and steps back by one suffice for $\mathbf{EXP}$.
\end{theorem}

\begin{proof}
(i) A forward-only path is nondecreasing. A gate $g$ moves to $2g+1>g$
at the next tick, so each gate is occupied at no more than one tick.
Fix $t$ and list the gates occupied before tick $t$ as
$g_1<\dots<g_N$; thus $N\le t$. Put $\lambda_0=0$ and
$\lambda_r=2g_r+1$. Between the departure from $g_{r-1}$ (or tick zero
when $r=1$) and the arrival at $g_r$, every move starts at a non-gate
index, so its displacement lies in $[0,w]$. Let
$\omega_r=g_r-\lambda_{r-1}$ be their total. For $r=1$,
$\omega_1=g_1\ge g_1/2-1$. For $r\ge2$, write $g_{r-1}=2^j$ and
$g_r=2^k$. Since $g_r>\lambda_{r-1}=2^{j+1}+1$, we have $k\ge j+2$,
so $\lambda_{r-1}\le2^{k-1}+1$ and $\omega_r\ge g_r/2-1$. In both
cases the gain at the $r$th gate satisfies
\[
 (2g_r+1)-g_r=g_r+1\le2\omega_r+3.
\]
Let $\Omega_t$ be the total displacement of the non-gate moves made
before tick $t$. Then $\Omega_t\le wt$, and the stretches counted by
the $\omega_r$ are disjoint, so $\sum_r\omega_r\le\Omega_t$.
Therefore
\[
 X_t=\Omega_t+\sum_{r=1}^N(g_r+1)
 \le3\Omega_t+3N\le3(w+1)t.
\]
Since the path is nondecreasing, this bounds $\max_{t\le q}X_t$. Thus
$\Pr(\tau_T\le q)=0$ whenever $T>3(w+1)q$, which gives the budget
bound. A witness with budget $Q$ accepts every yes-instance at an index
at most $3(w+1)Q(|x|)$. Simulating through that index, as in
\cref{prop:bounded-forward}, decides the language in polynomial time.
For the kernel that moves by $+1$ at every non-gate index, every move
increases the index, so $X_t\ge t$ surely. A decider with polynomial
logical bound $H$ is therefore observed by tick $H(n)$ with probability
one, and $\mathbf{P}\subseteq\mathcal{C}_{\eta,\alpha}$.

(ii) At a gate $X_{t+1}=2X_t+1$, and elsewhere
$X_{t+1}\le X_t+w$. Both give $X_{t+1}\le2X_t+w$, hence
$X_t\le w(2^t-1)$. A witness with budget $Q$ accepts at indices at most
$w2^{Q(|x|)}$, and simulating through that bound gives an
exponential-time decider.

(iii) Write $\eta=\eta^{\mathrm{gate}}_\beta$ and
$d_\beta=1+\beta^{-1}$. From zero the path visits $0,1,2,s_2=5$ at ticks
$0,1,2,3$. For $r\ge1$, the clock holds at $s_{r+1}$ until it steps to
the gate $2^{r+1}$, which happens after $G_r$ ticks, and the next tick
moves it to $s_{r+2}$. Every landing row has the same probabilities, so
$G_1,G_2,\dots$ are iid with law $\operatorname{Geom}(\beta)$ on
$\{1,2,\dots\}$. The arrival tick at $s_m$ is therefore
\[
 \sigma_m=3+\sum_{r=1}^{m-2}(G_r+1)\qquad(m\ge2).
\]
For $m\ge3$, before $\sigma_m$ every visited index is at most
$s_{m-1}$. For an integer $T\ge6$, put $m(T)=\lceil\log_2(T-1)\rceil\ge3$. Then
$s_{m(T)-1}<T\le s_{m(T)}$, so $\tau_T=\sigma_{m(T)}$. By the strong
law of large numbers~\cite[Theorem~2.4.1]{Durrett2019},
$\sigma_m/m\to d_\beta$ almost surely, and $m(T)/\log_2T\to1$. Hence
$\tau_T/\log_2T\to d_\beta$ almost surely, and so in probability.
For $0<\varepsilon<d_\beta$, the crossing probabilities at the budgets
$\lfloor(d_\beta-\varepsilon)\log_2T\rfloor$ and
$\lceil(d_\beta+\varepsilon)\log_2T\rceil$ tend to $0$ and $1$, so for
large $T$ the quantile $q_\alpha(T)$ lies between these budgets.
Dividing by $\log_2T$ and letting $\varepsilon\downarrow0$ proves the
limit; it does not depend on $\alpha$.

For $\mathbf{EXP}\subseteq\mathcal{C}_{\eta,\alpha}$, let a decider run within
$H(n)=2^{p(n)}$ transitions, where $p$ is an integer polynomial with
$p\ge2$ (after enlarging the bound). Since $s_{p(n)}\ge2^{p(n)}$, we
have $\tau_{H(n)}\le\sigma_{p(n)}$, and
$\mathbb{E}\sigma_{p(n)}\le3+d_\beta p(n)\le(3+d_\beta)p(n)$. Choose an integer
$k_{\alpha,\beta}\ge(3+d_\beta)/(1-\alpha)$ and put
$Q(n)=k_{\alpha,\beta}p(n)$. Markov's inequality gives
$\Pr(\tau_{H(n)}>Q(n))\le1-\alpha$. Accepting indices are at most
$H(n)$, crossing times are monotone in the target, and the decider is
sound, so $(M,Q)$ is a witness. Part~(ii) gives the reverse inclusion.

The final assertions follow from part~(i), since both modified kernels
are forward-only members of $\mathcal K_1$.
\end{proof}

\begin{remark}[Three levels of backward power]
\label{rem:three-levels}
\label{sec:what-activation-establishes}
For iid clocks, deletion never lowers the class (\cref{prop:deletion}). With fixed parity rows, deletion lowers $\mathbf{EXP}$ to $\mathbf{P}$, although forward re-entry also restores $\mathbf{EXP}$ (\cref{thm:activation,rem:forward-reentry}). In the gated family every forward-only clock stays in $\mathbf{P}$, so exceeding it requires backward steps (\cref{thm:gated}). A forward clock overshoots each gate and pays to walk to the next; a step back reopens the gate just passed. These clocks reuse indices of one path, without causal consistency or a physical speedup. For example, at $\beta=\alpha=1/2$ and $T=2^{n^2}$, \cref{thm:sharp} gives budgets asymptotic to $3(\ln2)n^2$ with returns and $2^{n^2}$ without them. Such power requires failure of monotone forward ability (\cref{thm:monotone-rows}) and persists with sparse, rare returns (\cref{sec:sparse-backward}).
\end{remark}

\subsection{Sparse and rare backward moves}
\label{sec:sparse-backward}

Returns can be supported on a density-zero set of indices. This is sparsity in
index space, not a statement that backward moves have zero frequency along
a trajectory. In the doubling family the return probabilities may vanish at
the specified landing indices. Squaring gates make the support sparser while
giving doubly exponential reach.

\label{txt:backward-indices}
For a kernel $\eta$ on $\mathbb{N}$, let
\[
 \mathrm{Back}(\eta)=\{i\in\mathbb{N}:\eta(i,j)>0\text{ for some }j<i\}
\]
be its set of \emph{backward indices}. A set $B\subseteq\mathbb{N}$ has
\emph{natural density zero} if $|B\cap[0,N)|/N\to0$. A backward move
of probability zero occurs at no tick almost surely, and the empty set
is the only counting-measure-null subset of $\mathbb{N}$. The relevant notions
of rarity are therefore density zero of $\mathrm{Back}(\eta)$ in index
space and small backward probabilities at its indices. We use
$\mathrm{Back}(\eta)$ to distinguish the index set from the parameter
set $\mathcal B$ in \cref{sec:fuzzy}.

\begin{corollary}[Density-zero backward indices suffice]
\label{cor:density-zero}
For every $\beta\in(0,1]$,
$\mathrm{Back}(\eta^{\mathrm{gate}}_\beta)=\{2^m+1:m\ge2\}$, so
$|\mathrm{Back}(\eta^{\mathrm{gate}}_\beta)\cap[0,N)|\le\log_2N$
for $N\ge2$, and this set has natural density zero. Nevertheless,
$\mathcal{C}_{\eta^{\mathrm{gate}}_\beta,\alpha}=\mathbf{EXP}$ for every
$0<\alpha<1$, whereas every $\eta\in\mathcal K_w$ with
$\mathrm{Back}(\eta)=\varnothing$ has
$\mathcal{C}_{\eta,\alpha}\subseteq\mathbf{P}$.
\end{corollary}

\begin{proof}
The rows of $\eta^{\mathrm{gate}}_\beta$ with mass below the diagonal
are exactly those at $s_m=2^m+1$, $m\ge2$. There are at most
$\log_2N$ such indices below $N$. The class statements are
\cref{thm:gated}(i) and (iii).
\end{proof}

\begin{proposition}[Vanishing backward probability]
\label{prop:vanishing-beta}
Fix $c\ge0$. Define $\eta^{\mathrm{gate}}_{(c)}\in\mathcal K_1$
like $\eta^{\mathrm{gate}}_\beta$, except that at the landing index
$s_m=2^m+1$ ($m\ge2$) it moves to $s_m-1$ with probability
$\beta_m=m^{-c}$ and holds with probability $1-\beta_m$. Thus the
backward probability at $s_m$ equals $(\log_2(s_m-1))^{-c}$, which
tends to zero when $c>0$. For every $0<\alpha<1$,
\[
 \mathcal{C}_{\eta^{\mathrm{gate}}_{(c)},\alpha}=\mathbf{EXP}.
\]
The forward-only kernels $\mathsf D\eta^{\mathrm{gate}}_{(c)}$ and
the kernel obtained from $\eta^{\mathrm{gate}}_{(c)}$ by replacing
each step $s_m\to s_m-1$ with $s_m\to s_m+1$ have classes contained
in $\mathbf{P}$.
\end{proposition}

\begin{proof}
As in the proof of \cref{thm:gated}(iii), the arrival tick at $s_m$ is
$\sigma_m=3+\sum_{r=1}^{m-2}(G_r+1)$, where now the $G_r$ are
independent and $G_r$ has law $\operatorname{Geom}(\beta_{r+1})$ on
$\{1,2,\dots\}$, so $\mathbb{E}G_r=(r+1)^c$. Hence, for $m\ge2$,
\[
 \mathbb{E}\sigma_m\le3+(m-2)(m^c+1)\le3+2m^{c+1}\le5m^{c+1}.
\]
Put $e=\lceil c\rceil+1$. Let a decider run within $2^{p(n)}$
transitions, where $p$ is an integer polynomial with $p\ge2$. Since
$s_{p(n)}\ge2^{p(n)}$, we have
$\tau_{2^{p(n)}}\le\sigma_{p(n)}$ and
$\mathbb{E}\sigma_{p(n)}\le5p(n)^e$. Choose an integer
$k_\alpha\ge5/(1-\alpha)$ and put $Q(n)=k_\alpha p(n)^e$.
Markov's inequality gives
$\Pr(\tau_{2^{p(n)}}>Q(n))\le1-\alpha$. Crossing times are
monotone in the target, and soundness is inherited from the decider,
so $\mathbf{EXP}\subseteq\mathcal{C}_{\eta^{\mathrm{gate}}_{(c)},\alpha}$.
Since $\eta^{\mathrm{gate}}_{(c)}\in\mathcal K_1$,
\cref{thm:gated}(ii) gives the reverse inclusion, and
\cref{thm:gated}(i) gives the statement about the forward-only kernels.
\end{proof}

The return probability may therefore decay polynomially in $\log_2 i$
along the landing indices. The mechanism is unchanged; only the mean
cycle lengths grow, from a constant to $O(m^c)$. At fixed confidence,
the crossing budget for $T=2^m$ is $O(m^{c+1})$, still polynomial in
$\log_2T$.

\begin{definition}[Squaring gates]
\label{def:sq-gated}
Let $\Lambda_2=\{2^{2^k}:k\ge1\}=\{4,16,256,\dots\}$. For an
integer $w\ge1$, let $\mathcal K^{(2)}_w$ be the set of
time-homogeneous stochastic kernels $\eta$ on $\mathbb{N}$ such that
$\eta(g,g^2+1)=1$ for every $g\in\Lambda_2$, and, for every
$i\notin\Lambda_2$, $\eta(i,j)=0$ whenever $j>i+w$.
\end{definition}

A squaring gate lands one step beyond the next gate $g^2$. For
$\beta\in(0,1]$, define
$\eta^{\mathrm{sq}}_\beta\in\mathcal K^{(2)}_1$ as follows. Put
$\eta^{\mathrm{sq}}_\beta(g,g^2+1)=1$ for $g\in\Lambda_2$. At
each landing index $s'_m=2^{2^m}+1$ with $m\ge2$, it moves to
$s'_m-1$ with probability $\beta$ and holds with probability
$1-\beta$. At every other index it moves to $i+1$. Thus
$\mathrm{Back}(\eta^{\mathrm{sq}}_\beta)=\{2^{2^m}+1:m\ge2\}$,
whose counting function below $N$ is at most $\log_2\log_2N$ for
$N\ge2$.

Squaring gates give doubly exponential time using only doubly sparse backward indices.

\begin{theorem}[Doubly sparse backward indices give $\mathbf{2EXP}$]
\label{thm:sq-gated}
Fix an integer $w\ge1$ and $0<\alpha<1$.
\begin{enumerate}[label=(\roman*),leftmargin=*]
\item If $\eta\in\mathcal K^{(2)}_w$ is forward-only, then surely
$\max_{t\le q}X_t\le wq+4w^2q^2$, so
$\mathcal{C}_{\eta,\alpha}\subseteq\mathbf{P}$. Equality holds for the member
that moves by $+1$ at every non-gate index.
\item Every $\eta\in\mathcal K^{(2)}_w$ satisfies
$X_t\le(2w)^{2^t}$ surely, so
$\mathcal{C}_{\eta,\alpha}\subseteq\mathbf{2EXP}$.
\item For every $\beta\in(0,1]$,
$\mathcal{C}_{\eta^{\mathrm{sq}}_\beta,\alpha}=\mathbf{2EXP}$, and the
crossing budget of $\eta^{\mathrm{sq}}_\beta$ satisfies
\[
 \frac{q_\alpha(T)}{\log_2\log_2T}\longrightarrow1+\beta^{-1}
 \qquad(T\to\infty).
\]
\end{enumerate}
\end{theorem}

\begin{proof}
(i) Argue as in \cref{thm:gated}(i). A forward-only path is
nondecreasing, and each gate is occupied at no more than one tick.
List the gates occupied before tick $t$ as $g_1<\dots<g_N$, put
$\lambda_0=0$ and $\lambda_r=g_r^2+1$, and let
$\omega_r=g_r-\lambda_{r-1}$ be the total non-gate displacement
between the departure from $g_{r-1}$ and the arrival at $g_r$, starting
at tick zero when $r=1$. For $r=1$, $\omega_1=g_1$. For $r\ge2$,
write $g_{r-1}=2^{2^j}$. Since
$g_r>\lambda_{r-1}=2^{2^{j+1}}+1$, we have
$g_r\ge2^{2^{j+2}}$, so $\lambda_{r-1}\le\sqrt{g_r}+1$ and
$\omega_r\ge g_r-\sqrt{g_r}-1\ge g_r/2$ because $g_r\ge256$.
In both cases $g_r\le2\omega_r$, so the gain at the $r$th gate
satisfies $g_r^2+1-g_r\le g_r^2\le4\omega_r^2$. With
$\Omega_t\le wt$ the total non-gate displacement before tick $t$,
the disjointness of the stretches gives
\[
 X_t\le\Omega_t+4\sum_r\omega_r^2
 \le\Omega_t+4\Omega_t^2\le wt+4w^2t^2.
\]
The path is nondecreasing, so this bounds the running maximum. A
witness with budget $Q$ therefore accepts at indices at most
$wQ(n)+4w^2Q(n)^2$, and simulation gives a polynomial-time decider.
For the member moving by $+1$ at non-gate indices, $X_t\ge t$ surely,
which gives $\mathbf{P}\subseteq\mathcal{C}_{\eta,\alpha}$.

(ii) Put $Y_t=X_t+w$. At a gate,
$X_{t+1}=X_t^2+1\le Y_t^2$; elsewhere
$X_{t+1}\le X_t+w=Y_t\le Y_t^2$. Hence
$Y_{t+1}\le Y_t^2+w\le2Y_t^2$, so
$2Y_{t+1}\le(2Y_t)^2$. Since $Y_0=w$, induction gives
$2Y_t\le(2w)^{2^t}$. Thus $X_t\le(2w)^{2^t}$. A witness with
budget $Q$ accepts at indices at most $(2w)^{2^{Q(n)}}$, and
simulating through that bound is a doubly exponential-time decider.

(iii) Write $d_\beta=1+\beta^{-1}$. From zero the path visits
$0,1,2,3,4$ at ticks $0,\dots,4$ and $s'_2=17$ at tick $5$. For
$r\ge1$, the clock waits $G_r$ ticks at $s'_{r+1}$ until it steps to
the gate $2^{2^{r+1}}$, and the next tick moves it to $s'_{r+2}$.
The $G_r$ are iid $\operatorname{Geom}(\beta)$ on $\{1,2,\dots\}$, so the
arrival tick at $s'_m$ is
\[
 \sigma'_m=5+\sum_{r=1}^{m-2}(G_r+1)\qquad(m\ge2).
\]
For $m\ge3$, before $\sigma'_m$ every visited index is at most
$s'_{m-1}$. For an integer $T\ge18$, put
$m(T)=\lceil\log_2\log_2(T-1)\rceil\ge3$. Then
$s'_{m(T)-1}<T\le s'_{m(T)}$, so $\tau_T=\sigma'_{m(T)}$.
The strong law of large numbers and $m(T)/\log_2\log_2T\to1$ give
$\tau_T/\log_2\log_2T\to d_\beta$ almost surely. The quantile
limit follows as in \cref{thm:gated}(iii).

For $\mathbf{2EXP}\subseteq\mathcal{C}_{\eta^{\mathrm{sq}}_\beta,\alpha}$,
let a decider run within $2^{2^{p(n)}}$ transitions, where $p$ is an
integer polynomial with $p\ge2$. Since
$s'_{p(n)}\ge2^{2^{p(n)}}$, we have
$\tau_{2^{2^{p(n)}}}\le\sigma'_{p(n)}$, and
$\mathbb{E}\sigma'_{p(n)}\le5+d_\beta p(n)\le(5+d_\beta)p(n)$. With an
integer $k\ge(5+d_\beta)/(1-\alpha)$ and $Q(n)=kp(n)$, Markov's
inequality gives completeness at confidence $\alpha$. Soundness is
inherited from the decider. Part~(ii) gives the reverse inclusion.
\end{proof}

Fixed iid laws also reach $\mathbf{2EXP}$ (\cref{cor:2exp-law}).
\Cref{thm:sq-gated} concerns the comparison inside $\mathcal K^{(2)}_w$,
where forward-only members stay in $\mathbf{P}$.

\begin{remark}[How sparse can backward indices be?]
\label{rem:how-sparse}
In both gated families the backward indices are the landing points, so
their sparsity is set by the spacing of the gates. Forward-only members
stay in $\mathbf{P}$ because the gain at a gate is polynomial in the walk
needed to reach it: linear for doubling gates, quadratic for squaring
gates. These constructions give counting orders $\log N$ and
$\log\log N$, respectively. A gate $g\mapsto2^g$ would give a
forward-only clock an exponential gain after a walk of order $g$,
outside this polynomial-gain argument. We leave open whether backward
indices of arbitrarily small counting order can raise the class beyond
$\mathbf{P}$ in a family whose forward-only members stay in $\mathbf{P}$.
\end{remark}

\space%
\clearpage
\section{Conclusions and open problems}
\label{sec:discussion}

Crossing reach connects a fixed random clock to its observation class.
The sparse computable constructions give one zero-drift intermediate class
and an infinite negative-drift antichain below $\mathbf{P}$. The gated constructions
give a different comparison: under their prescribed gate jumps and bounded
forward displacement away from gates, every forward-only member stays within
$\mathbf{P}$, whereas length-one returns attain $\mathbf{EXP}$ or $\mathbf{2EXP}$.
The parity example establishes a gain over deletion, with forward re-entry
also attaining $\mathbf{EXP}$. In both gated families, sparsity concerns backward-index support in index
space. The doubling family also permits vanishing return probabilities at
its specified landing indices. These classifications use the unit-cost
observation convention.

The classifications leave six questions. Each concerns an unresolved
regime or the boundary of the present constructions.
\begin{enumerate}[leftmargin=*]
\item Determine the exact class of the sparse zero-drift law of
\cref{thm:zero-drift-counterexample}.
\item Classify negative-mean laws with only some finite positive moments,
including the antichain regime of \cref{thm:antichain}.
\item Classify integrable zero-drift laws with infinite variance and
unbounded positive jumps.
\item Determine how confidence affects iid classes beyond the classified
cases. \Cref{app:calibration-7} gives a position-dependent example;
\cref{app:calibration-8} transfers iid bounds but leaves the sparse case open.
\item Determine the sparsest backward-index sets that can raise a gated
family beyond $\mathbf{P}$ (\cref{rem:how-sparse}).
\item Determine whether the law of \cref{prop:infinite-moments-p} is
computable.
\end{enumerate}
Random input access or charging configuration generation would require a
new analysis of the cost model.
\space%

\section*{Declaration of generative AI and AI-assisted technologies}
The model, the questions and the main theorems are the authors'. During the preparation of this work the authors used large language models to check proofs, to draft proofs of some lemmas, to search the literature, and to edit the text and the code of the figures. The authors checked every proof in the final manuscript and take full responsibility for its content.
\label{end:maintext}
\clearpage
\begin{appendices}
\crefalias{section}{appendix}
\crefalias{subsection}{appendix}
\crefname{appendix}{Appendix}{Appendices}
\Crefname{appendix}{Appendix}{Appendices}
\section{The computable zero-drift construction}
\label{app:zero-drift}

We prove \cref{thm:zero-drift-counterexample} by placing sparse atoms and
balancing the mean. At chosen even targets $T_j$, the budgets $q_j=T_j^j$
eventually dominate every fixed polynomial evaluated at $T_j$, yet the
crossing probability is below $2^{-j}$. Finite searches make the law
computable, and recursive choices preserve all earlier bounds.

\begin{proof}[Proof of \cref{thm:zero-drift-counterexample}]
Subscripts on probabilities and expectations in this proof specify the
proposal law.
Put $a_j=2^{-j-1}$ for $j\ge1$. Then
$\sum_{j\ge1}a_j=1/2$ and $\sum_{i\ge j}a_i=2^{-j}$.
We will choose $N_j$ recursively and set
\begin{equation}
 \mu(-1)=\frac12,\qquad
 \mu(N_j)=\frac{a_j}{N_j},\qquad
 \mu(0)=\frac12-\sum_{j\ge1}\frac{a_j}{N_j},
 \label{eq:zero-drift-law}
\end{equation}
with zero mass elsewhere. Since $N_j\ge1$, the series is at most $1/2$.
The zero mass is therefore nonnegative and the total mass is one, with
\[
 \mathbb{E}K^+=\sum_{j\ge1}a_j=\frac12=\mathbb{E}K^-,
 \qquad \mathbb{E}|K|=1,\qquad \mathbb{E}K=0.
\]

\emph{Effective finite approximations.}
Set $N_0=T_0=0$. At stage $j$, only $N_1,\ldots,N_{j-1}$ have been
chosen. Define the finite rational law $\nu_j$ below, with all atom
indices in the sums starting at one:
\[
 \nu_j(-1)=\frac12,\qquad
 \nu_j(N_i)=\frac{a_i}{N_i}\ (i<j),\qquad
 \nu_j(0)=\frac12-\sum_{i<j}\frac{a_i}{N_i}.
\]
It has drift $-1/2+\sum_{i<j}a_i=-2^{-j}<0$.
In the completed law, this is the distribution obtained by replacing every proposal
$N_i$ with $i\ge j$ by zero.

Search over $m=1,2,\ldots$ for a rational $r_j=1+2^{-m}$ such that
\[
 \lambda_j=\mathbb{E}_{\nu_j}[r_j^K]<1.
\]
Every test uses exact rational arithmetic on a finite sum. The derivative
of $r\mapsto\mathbb{E}_{\nu_j}[r^K]$ at $r=1$ is $-2^{-j}$, so the sum is
less than one for all $r>1$ sufficiently close to one. The search
therefore terminates.

For $V(i)=r_j^i$, at every $i\ge1$ all proposals of $\nu_j$ are
admissible, since its only negative proposal is $-1$. At zero that
proposal holds. Hence, with $b_j=(1-r_j^{-1})/2$, the exact identity is
\[
 \mathbb{E}_{\nu_j}[V(X_{t+1})\mid X_t=i]
 =\lambda_j V(i)+b_j\mathbf{1}_{\{i=0\}}
 \le\lambda_j V(i)+b_j.
\]
Iteration from $X_0=0$ gives the computable rational bound
$\sup_t\mathbb{E}_{\nu_j}V(X_t)\le C_j:=1+b_j/(1-\lambda_j)$.
Markov's inequality and Boole's inequality (the union bound)
\cite[Theorem~1.1.1]{Durrett2019} over ticks yield
\begin{equation}
 \Pr_{\nu_j}(\tau_T\le q)\le C_j(q+1)r_j^{-T}
 \quad(T\ge1,\ q\ge0).
 \label{eq:effective-negative-crossing}
\end{equation}

\emph{Choosing the targets and atoms.}
Search in increasing order for an even integer $T_j\ge\max(2,j)$ with
$T_j>T_{j-1}$ and
\begin{equation}
 C_j(T_j^j+1)r_j^{-T_j}<2^{-j-1}.
 \label{eq:zero-stage-target}
\end{equation}
For fixed $j$ and $r_j>1$, $(T^j+1)r_j^{-T}\to0$ along the even
integers, so this exact rational search terminates. Put $q_j=T_j^j$ and
choose the least integer
\begin{equation}
 N_j>\max\{N_{j-1},\,2q_j,\,2^{2j}\}.
 \label{eq:zero-stage-atom}
\end{equation}
These rules define computable increasing sequences independently of
the machine, input, polynomial budget, and confidence.

\emph{The completed law.}
For every $j$, monotonicity of the atoms and the sum of the $a_i$ give
\[
 \sum_{i\ge j}\mu(N_i)
 =\sum_{i\ge j}\frac{a_i}{N_i}\le\frac{2^{-j}}{N_j}.
\]
Couple a walk with law $\mu$ to one with law $\nu_j$ by replacing a
proposal $N_i$, $i\ge j$, by zero and leaving all other proposals
unchanged. The transformed sequence has independent proposals with law
$\nu_j$.
The two held walks agree until the first replacement. A union bound over
the first $q_j$ proposals, followed by
\eqref{eq:effective-negative-crossing}--\eqref{eq:zero-stage-atom}, gives
\begin{equation}
 \Pr_\mu(\tau_{T_j}\le q_j)
 \le C_j(q_j+1)r_j^{-T_j}
       +q_j\sum_{i\ge j}\mu(N_i)
 <2^{-j-1}+2^{-j-1}=2^{-j}.
 \label{eq:zero-subsequence-crossing}
\end{equation}
The sum includes every future atom of the infinite law. Later choices
therefore preserve this bound at every earlier stage.

\emph{Even length and the quantifiers.}
Fix any $0<\alpha<1$ and suppose $(M,Q)$ were a witness for
$L_{\mathrm{even}}$. Let $x$ be an accepted word of positive even length
$n$, and choose $a\in\Sigma$. The tapes for $x$ and $xa$ initially differ
only at cell $n+1$: it is blank for $x$ and contains $a$ for $xa$.
Until that cell is read, deterministic transitions, including any writes
to earlier cells, are identical on the two inputs. Accepting $x$ before
reading cell $n+1$ would therefore accept the odd-length word $xa$,
contrary to soundness. The $t$th transition can read at most cell $t$,
so reading cell $n+1$ requires at least $n+1$ transitions. Therefore $A_M(x)\ge n+1$; the construction only needs
$A_M(x)\ge n$. We use only positive even lengths below;
$\varepsilon$ needs no lower bound.

Choose integers $c\ge1$ and $d\ge0$ with $Q(n)\le c n^d$ for $n\ge1$;
the sum of the coefficients, enlarged to one, suffices for $c$.
For all sufficiently large $j$, $j\ge d$ and $2^{j-d}\ge c$.
Since $T_j\ge2$, this implies
$Q(T_j)\le cT_j^d\le T_j^j=q_j$.
The same conclusion holds for any polynomial upper bound after enlarging
it to an integer polynomial as in \cref{sec:preliminaries}.
For the yes-input $x_j=a^{T_j}$, monotonicity in the target and budget
now gives
\[
 \Pr_\mu(\tau_{A_M(x_j)}\le Q(T_j))
 \le\Pr_\mu(\tau_{T_j}\le q_j)<2^{-j}.
\]
Taking $j$ large enough that $2^{-j}<\alpha$ contradicts completeness.
This excludes every witness at every fixed confidence for the same law.
The finite first positive moment gives containment in $\mathbf{P}$ by
\cref{thm:kt-positive-moment}; $L_{\mathrm{even}}$ makes the containment
strict. The universal lower bound of \cref{lem:finite-prefix-class} gives
$\mathbf{RDEF}\subseteq\mathbf P^*_{\mu,\alpha}$, since $\Pr(K>0)>0$.

\emph{Computability and variance.}
Each stage above terminates by a proved strict-inequality search over
finite rational data. To compute $\mu(k)$ for a positive integer $k$,
generate the $N_j$ until the first $N_j\ge k$. If equality holds, return
$a_j/N_j$; otherwise return zero. This terminates because $N_j>2^{2j}$.
For $k=-1$ return $1/2$, and for $k<-1$ return zero.
At zero, after computing the first $s$ atoms, the remaining mass satisfies
\[
 0\le\sum_{i>s}\frac{a_i}{N_i}\le\sum_{i>s}a_i=2^{-s-1}.
\]
Thus $1/2-\sum_{i=1}^s a_i/N_i$ approximates $\mu(0)$ with error at
most $2^{-s-1}$. Given a requested error $2^{-p}$, choose $s\ge p$.
This gives uniform computability of the point probabilities.
Finally
\[
 \sum_{j\ge1}N_j^2\mu(N_j)=\sum_{j\ge1}a_jN_j=\infty,
 \qquad a_jN_j>2^{j-1}.
\]
The mean is zero, so this infinite second moment is also infinite
variance.
\end{proof}

\subsection{A strictly intermediate class}
\label{sec:strictly-intermediate}

The same law admits a language outside $\mathbf{RDEF}$. We use input lengths
large enough to pay for regenerating an atom, but small enough that a
polynomial observation budget crosses that atom.

\begin{proof}[Strict lower inclusion in \cref{thm:zero-drift-counterexample}]
Fix a finite nonempty alphabet $\Sigma$. The law and the sequences
$(T_j)$ and $(N_j)$ above do not depend on $\Sigma$.
Since $\mu(N_j)=2^{-j-1}/N_j$ and $N_j>2^{2j}$,
\[
 2^{j+1}<2N_j^{1/2},\qquad
 \mu(N_j)>\frac1{2N_j^{3/2}}.
\]
A proposal $N_j$ crosses threshold $N_j$ from every nonnegative position.
For any fixed $0<\alpha<1$, independence therefore gives
\begin{equation}
 \Pr_\mu(\tau_{N_j}\le q)
 \ge1-(1-\mu(N_j))^q
 \ge1-\exp(-q\mu(N_j))\ge\alpha
 \quad\text{if }q\ge2\ell_\alpha N_j^{3/2}.
 \label{eq:sparse-large-radius}
\end{equation}
Thus $R_\alpha(q)\ge N_j$ at these integer budgets.

We next bound the cost of the exact searches, including a single-tape
implementation. For $j\ge2$, put $N=N_{j-1}\ge1$ and let $m_j$ be the
first successful exponent in the search for $r_j=1+2^{-m_j}$. If
$0<\varepsilon\le1/N$, Taylor's formula gives
$(1+\varepsilon)^k\le1+k\varepsilon+(e/2)k^2\varepsilon^2$ for $1\le k\le N$.
Also $(1+\varepsilon)^{-1}\le1-\varepsilon+\varepsilon^2$. Since
$\sum_{i<j}a_i=1/2-2^{-j}$ and
$\sum_{i<j}a_iN_i\le N/2$, we obtain
\[
 \mathbb{E}_{\nu_j}(1+\varepsilon)^K-1
 \le-2^{-j}\varepsilon+\left(\frac12+\frac e4N\right)\varepsilon^2
 \le-2^{-j}\varepsilon+2N\varepsilon^2.
\]
Taking $\varepsilon\le2^{-j}/(4N)$ makes this negative. Consequently
\begin{equation}
 m_j\le j+\lceil\log_2N\rceil+2.
 \label{eq:sparse-search-exponent}
\end{equation}
Stage one has fixed finite cost.

The earlier atom has mass $a_{j-1}/N=2^{-j}/N$ in $\nu_j$.
Every term in $\lambda_j$ is nonnegative, so $\lambda_j<1$ implies
$r_j^N<2^jN$. On the other hand, \eqref{eq:zero-stage-target} and
$C_j\ge1$ imply $r_j^{T_j}>2^{j+1}$. Taking natural logarithms yields
\begin{equation}
 T_j>\frac{(j+1)\ln2}{\ln r_j}
 >\frac{(j+1)(\ln2)N}{j\ln2+\ln N}.
 \label{eq:sparse-stage-size}
\end{equation}
By \eqref{eq:zero-stage-atom}, $N>2^{2(j-1)}$, so $j=O(\log N)$.
In particular $T_j\ge cN/\log N$ for an absolute $c>0$ at large
stages. This also gives $T_j\ge\sqrt N$ eventually, and then
$\log N\le2\log T_j$. It follows that
\begin{equation}
 N_{j-1}=O(T_j\log T_j),\qquad
 j+m_j=O(\log T_j).
 \label{eq:sparse-stage-control}
\end{equation}
All constants are independent of the stage.

Here is a bit-cost bound for the searches. For a tested exponent
$m\le m_j$, the powers $(1+2^{-m})^{N_i}$ can be calculated by binary
exponentiation with exact integer numerators and denominators. A common
denominator for the finite expectation has bit length
$O(mN+j\log(N+1)+j)$: one can use the common dyadic power $2^{mN}$,
the product of the earlier $N_i$, a dyadic denominator for the $a_i$,
and $2^m+1$ for the negative atom. The rational $C_j$ has the same
order of bit length. At a tested target $T\le T_j$, the comparison in
\eqref{eq:zero-stage-target} uses integers with at most
\[
 B_j=O\bigl(m_j(N+T_j)+j\log(N+T_j+1)\bigr)
     =O\bigl(T_j(\log T_j)^2\bigr)
\]
bits. This includes $T^j$, the rational powers, and the exact strict
comparison; no numerical precision search is needed. There are at most
$m_j$ exponent tests and $T_j$ target tests, each using at most $j$
terms and binary exponentiation. Schoolbook integer arithmetic bounds
these operations by a polynomial in $N,T_j,j,m_j$ of a degree independent
of $j$. Finally setting $N_j$ by \eqref{eq:zero-stage-atom} uses
$O(j\log T_j+\log(N+1)+j)$ bits. By
\eqref{eq:sparse-stage-control}, the entire stage costs at most
$C(T_j+1)^D$ for fixed $C,D$. Polynomial single-tape simulation overhead
can be included by increasing $D$~\cite{AroraBarak2009}.

Choose an integer $j_0$ sufficiently large, as specified below, and put
\begin{equation}
 n_j=\lfloor N_j^{1/2}\rfloor,\qquad
 L_G=\{x\in\Sigma^*:|x|=n_j\text{ for some }j\ge j_0\}.
 \label{eq:sparse-gap-language}
\end{equation}
A machine first computes $n=|x|$ and then regenerates the stages by the
exact searches above. At stage $i$ it computes $\lfloor\sqrt{N_i}\rfloor$
by integer arithmetic. It accepts if this equals $n$ and $i\ge j_0$,
and rejects as soon as it exceeds $n$. The search terminates since the
$N_i$ increase without bound. This machine accepts precisely $L_G$.

For completeness, input-length computation need not incur quadratic
single-tape cost. Mark all input cells on an auxiliary symbol track, and
make repeated sweeps that retain every second surviving mark. The parity
of the number of marks before each sweep records the next binary digit of
$n$ on a separate track near the left end. There are $O(\log(n+2))$
sweeps, each costing $O(n+1)$ transitions. After recording $n$, erase the
other input tracks in one sweep and return to the left end. This takes
$O((n+1)\log(n+2))$ transitions and leaves the binary length in a local
workspace. The subsequent arithmetic is simulated in that workspace,
so its polynomial simulation cost does not square the input scan.

On a yes-input of length $n_j$, all stage comparisons with $n_j$ can
be included in a fixed polynomial in $T_j$. Indeed
$\log N_j=O(j\log T_j+\log(N_{j-1}+1)+j)$ by the least-integer
choice in \eqref{eq:zero-stage-atom}. This bounds both the length of
$n_j$ and the cost of its comparisons and of the integer square roots.
Since $T_i\le T_j$ for $i\le j$ and $j\le T_j$, there are constants
$C',D'$ such that the total single-tape arithmetic time through stage $j$
is at most $C'T_j^{D'}$. But $N_j>2T_j^j$, so
\[
 \frac{C'T_j^{D'}}{N_j}<\frac{C'}2T_j^{D'-j}\longrightarrow0.
\]
The length-computation cost is also $o(N_j)$ because
$n_j\le\sqrt{N_j}$. Choose $j_0$ large enough that the complete
accepting run, including input access and simulation overhead, satisfies
\begin{equation}
 A_M(x)\le N_j\quad\text{when }|x|=n_j, j\ge j_0.
 \label{eq:sparse-accepting-time}
\end{equation}
This choice depends on the machine implementation, not on $\alpha$.
There is no circularity. The stage generator and the comparisons do not
depend on $j_0$. The final test $i\ge j_0$ uses a counter in the finite
control that saturates at $j_0$ and costs no additional transitions.
Hence the constants $C',D'$ do not depend on $j_0$, and $j_0$ may be
chosen after them.

The accepted lengths have arbitrarily large gaps. In fact
$T_{j+1}\ge\sqrt{N_j}$ eventually by \eqref{eq:sparse-stage-control},
and $N_{j+1}>2T_{j+1}^{j+1}$, so
$N_{j+1}/N_j>2N_j^{(j-1)/2}\to\infty$.
Enlarge $j_0$ so that the $n_j$ are strictly increasing with gaps.
The language $L_G$ is infinite because $\Sigma$ is nonempty, but it
contains no cylinder $v\Sigma^*$: such a cylinder contains words of every
length at least $|v|$. An infinite reverse-definite language contains a
cylinder by \eqref{eq:reverse-definite}, since its exceptional set $F$
is finite. Thus $L_G\notin\mathbf{RDEF}$.

For confidence $\alpha$, take the integer polynomial
\[
 Q(n)=\lceil2\ell_\alpha\rceil(n+1)^4.
\]
On an accepted length $n=n_j$, we have $N_j^{3/2}<(n+1)^3$, and hence
$Q(n)\ge2\ell_\alpha N_j^{3/2}$. Combining
\eqref{eq:sparse-large-radius} and \eqref{eq:sparse-accepting-time} gives
$A_M(x)\le N_j\le R_\alpha(Q(|x|))$.
\Cref{lem:reachability} proves $L_G\in\mathbf P^*_{\mu,\alpha}$.
The same machine works at every confidence, with only $Q$ changing,
so $L_G\in\mathbf P^*_\mu$. The strict upper inclusions follow from
$L_{\mathrm{even}}$ as in the preceding proof. The law is independent
of the alphabet, proving the claimed universal quantifier.
\end{proof}
\space%
\section{Proof of \texorpdfstring{\cref{thm:antichain}}{the antichain theorem}}
\label{app:antichain}

For $n\ge1$, let $\log_2^{(0)}n=n$ and $\log_2^{(r+1)}n=\log_2(\log_2^{(r)}n)$ while the current value exceeds $1$, and put $\log^{*}n=\min\{r\ge0:\log_2^{(r)}n\le1\}$. Thus $\log^{*}n$ counts the successive binary logarithms needed to bring $n$ to at most $1$, with $\log^{*}1=0$. The smaller-input case in the proof is handled separately.

\begin{proof}
\emph{Basic properties.} The tower values satisfy
$\mathrm{tw}_j\ge2^{j+4}$, so
$\sum_{j\ge1}\mathrm{tw}_j^{-1}\le1/16$.
Thus the zero mass is positive and
$\mathbb{E}K\le-\tfrac12+\tfrac1{16}<0$.
For $s>0$,
$\mathbb{E}[(K^+)^s]=\sum_{j\in J_e}\mathrm{tw}_j^{3s-4}$,
which is finite exactly when $s<4/3$.
The nonzero point masses are rational, and for $k>0$ the value
$\mu^{\langle e\rangle}(k)$ is computed by generating
$\mathrm{tw}_j$ until $\mathrm{tw}_j^3\ge k$, then testing equality
and $j\in J_e$. At $-1$ the mass is $1/2$, and at smaller integers
it is zero. The tail of the zero-mass series after stage $j$ is at most
$2\,\mathrm{tw}_{j+1}^{-4}$. This effective remainder computes the
zero mass to any prescribed error and proves computability.
A positive atom gives
$\mathbf{RDEF}\subseteq\mathbf P^*_{\mu^{\langle e\rangle},\alpha}$
by \cref{lem:finite-prefix-class}, and the finite first positive
moment gives containment in $\mathbf{P}$ by \cref{thm:kt-positive-moment}.

\emph{A slow-stage bound.} Fix $e$ and an index $j\ge2$ with
$j\notin J_e$. We show
\begin{equation}
 \begin{split}
 \Pr_{\mu^{\langle e\rangle}}
 (\tau_{\mathrm{tw}_j}\le\mathrm{tw}_j^{\,j})
 &\le\varepsilon_j^{\mathrm{ac}}\\
 &:=6\,\mathrm{tw}_j^{\,j}
 \exp\!\left(-\frac{\mathrm{tw}_j}{(\log_2\mathrm{tw}_j)^3}\right)
 +2\,\mathrm{tw}_j^{\,j}\,2^{-4\mathrm{tw}_j},
 \end{split}
 \label{eq:slow-stage}
\end{equation}
and $\varepsilon_j^{\mathrm{ac}}\to0$.
Put $q=\mathrm{tw}_j^{\,j}$. The superscript $\mathrm{ac}$
distinguishes the auxiliary quantities here from those in the earlier
appendices.

First, the atoms $\mathrm{tw}_i^3$ with $i>j$ have total mass at
most $2\,\mathrm{tw}_{j+1}^{-4}=2\cdot2^{-4\mathrm{tw}_j}$.
Replace each such proposal by zero to obtain an iid law
$\nu_j^{\mathrm{ac}}$. The two held walks agree until the first
replaced proposal, so a union bound over the first $q$ proposals
contributes the second term of \eqref{eq:slow-stage}.

Second, the positive atoms of $\nu_j^{\mathrm{ac}}$ are
$\mathrm{tw}_i^3$ with $i<j$. Since $j\notin J_e$, all are at most
\[
 b_j^{\mathrm{ac}}=\mathrm{tw}_{j-1}^3
 =(\log_2\mathrm{tw}_j)^3\ge32^3.
\]
Put $\kappa_j=1/b_j^{\mathrm{ac}}\le1/4$.
For $0\le x\le1$, $e^x\le1+x+x^2\le1+2x$, and
$e^{-\kappa_j}\le1-\kappa_j+\kappa_j^2/2$.
Hence, using $\sum_i\mathrm{tw}_i^{-1}\le1/16$,
\[
 \begin{split}
 \mathbb{E}_{\nu_j^{\mathrm{ac}}}e^{\kappa_j K}-1
 &\le-\tfrac{\kappa_j}2+\tfrac{\kappa_j^2}4
 +2\kappa_j\sum_{i<j}\mathrm{tw}_i^{-1}\\
 &\le-\tfrac{\kappa_j}2+\tfrac{\kappa_j}{16}
 +\tfrac{\kappa_j}8\le-\tfrac{\kappa_j}4.
 \end{split}
\]
Write $g=\mathbb{E}_{\nu_j^{\mathrm{ac}}}e^{\kappa_j K}
\le1-\kappa_j/4$ and $V(i)=e^{\kappa_j i}$.
The only negative proposal is $-1$, so
$\mathbb{E}[V(X_{t+1})\mid X_t=i]=gV(i)$ for $i\ge1$, and at zero
the held proposal adds at most
$\tfrac12(1-e^{-\kappa_j})\le\kappa_j/2$.
Iterating from $V(0)=1$ gives
$\sup_t\mathbb{E}V(X_t)\le1+(\kappa_j/2)/(\kappa_j/4)=3$.
Markov's inequality and a union bound give
\[
 \Pr_{\nu_j^{\mathrm{ac}}}(\tau_T\le q)
 \le3(q+1)e^{-\kappa_j T},
\]
which at $T=\mathrm{tw}_j$ is at most the first term of
\eqref{eq:slow-stage}, since $q\ge1$.

Finally, $j\le\log_2\mathrm{tw}_j$, so
$\ln(\mathrm{tw}_j^{\,j})\le(\log_2\mathrm{tw}_j)^2$,
while $\mathrm{tw}_j/(\log_2\mathrm{tw}_j)^3$ grows faster than
this square. The term $4(\ln2)\mathrm{tw}_j$ likewise dominates
$\ln(\mathrm{tw}_j^{\,j})$. Thus both terms of
$\varepsilon_j^{\mathrm{ac}}$ tend to zero.

\emph{Membership.} A machine $M_{e,j_0}$ computes $n=|x|$ in binary
in $O((n+1)\log(n+2))$ transitions, as in
\cref{sec:strictly-intermediate}. It then generates
$\mathrm{tw}_1,\mathrm{tw}_2,\dots$ in binary in a workspace near
the left end. It rejects as soon as some $\mathrm{tw}_j>n$, or
as soon as $\mathrm{tw}_j<n$ and $\mathrm{tw}_j$ is at least the
bit length of $n$, since then $\mathrm{tw}_{j+1}>n$.
When $\mathrm{tw}_j=n$, it accepts exactly if $j\in J_e$ and
$j\ge j_0$, and otherwise rejects.
All arithmetic is on numbers of $O(\log(n+2))$ bits, and there
are at most $\log^*n+1$ stages for $n\ge2$; smaller inputs are
handled in constant time. The local arithmetic, including single-tape
simulation, has polynomial cost in $\log(n+2)$ and is absorbed
in the scan bound.

The tests on $j$ are implemented in finite control: a residue modulo
$2^{e+1}$ tests $j\in J_e$, and a stage counter saturating at $j_0$
tests $j\ge j_0$. Their updates share the existing stage transitions.
Taking the product with this finite control adds states but no
transitions to the common length-and-tower computation. Hence
\[
 A_{M_{e,j_0}}(x)\le C(n+2)\log_2(n+2)
\]
for a constant $C$ independent of both $e$ and $j_0$.
Choose $j_*\ge2$ with
$C(\mathrm{tw}_j+2)\log_2(\mathrm{tw}_j+2)\le\mathrm{tw}_j^3$
for all $j\ge j_*$. For $|x|=\mathrm{tw}_j$ with $j\in J_e$
and $j\ge j_0\ge j_*$, a proposal equal to $\mathrm{tw}_j^3$
crosses the accepting index from every nonnegative position. Hence
\[
 \begin{split}
 \Pr(\tau_{A_{M_{e,j_0}}(x)}\le q)
 &\ge1-(1-\mathrm{tw}_j^{-4})^q\ge\alpha\\
 &\hspace{12mm}\text{whenever }q\ge\ell_\alpha\mathrm{tw}_j^4.
 \end{split}
\]
The budget $Q(n)=\lceil\ell_\alpha\rceil n^4$ is therefore
adequate, and $M_{e,j_0}$ is a decider, so it is sound.
Only $Q$ depends on $\alpha$, so
$L^{\langle e\rangle}_{j_0}\in\mathbf P^*_{\mu^{\langle e\rangle}}$.

\emph{Non-membership.} Suppose $(M,Q)$ witnesses
$L^{\langle e\rangle}_{j_0}
\in\mathbf P^*_{\mu^{\langle e'\rangle},\alpha'}$ with $e'\ne e$.
Let $j\in J_e$ with $j\ge j_0$, choose $a\in\Sigma$, and put
$x=a^{\mathrm{tw}_j}$. Then $x\in L^{\langle e\rangle}_{j_0}$,
but $xa\notin L^{\langle e\rangle}_{j_0}$ because
$\mathrm{tw}_j+1$ is not a tower value. The distinguishing-cell
argument of \cref{app:zero-drift} gives
$A_M(x)\ge\mathrm{tw}_j+1$.
For all large $j$, $Q(\mathrm{tw}_j)\le\mathrm{tw}_j^{\,j}$.
Since $j\notin J_{e'}$, \eqref{eq:slow-stage} for
$\mu^{\langle e'\rangle}$ gives
\[
 \Pr_{\mu^{\langle e'\rangle}}
 (\tau_{A_M(x)}\le Q(\mathrm{tw}_j))
 \le\varepsilon_j^{\mathrm{ac}}<\alpha'
\]
for large $j\in J_e$, contradicting completeness.

\emph{Conclusion.} The language $L^{\langle e\rangle}_{j_0}$
is infinite and contains no cylinder $v\Sigma^*$, so it is not
reverse definite by \cref{def:reverse-definite}.
It belongs to $\mathbf{P}$ via $M_{e,j_0}$.
Membership and non-membership give pairwise incomparability, also
for the intersection classes. The language
$L^{\langle e\rangle}_{j_0}$ gives strict inclusion above $\mathbf{RDEF}$,
and, for any $e'\ne e$,
$L^{\langle e'\rangle}_{j_0}
\in\mathbf{P}\setminus\mathbf P^*_{\mu^{\langle e\rangle},\alpha}$
gives strict inclusion below $\mathbf{P}$.
\end{proof}
\space%
\section{Effectiveness proofs}
\label{app:effectiveness}

This appendix proves \cref{prop:computable-rec,prop:arbitrary-re} and
records the equivalence between \cref{def:kt-computable-law} and the
standard notion of a computable probability measure on the integers.

On the discrete integers, an effectively open set is a computably
enumerable set of integers. The standard characterization of a computable
probability measure requires the measures of these sets to be uniformly
lower semicomputable: an algorithm produces increasing rational lower
bounds converging to the
measure~\cite{HoyrupRojas2009}.
\Cref{def:kt-computable-law} implies this property. At stage
$s\ge1$, let $E_s$ be the finite set of distinct integers enumerated so
far, indexed by order of first appearance as $k_j$. Approximate each
$\mu(k_j)$ in $E_s$ by $d_{s,j}$ with error at most
$e_{s,j}=2^{-s-j-1}$, and put
$L_s=\sum_{k_j\in E_s}\max\{0,d_{s,j}-e_{s,j}\}$. Then
$0\le\mu(E_s)-L_s\le2^{-s}$: each clamped summand is at most $\mu(k_j)$,
and since $\max\{0,d_{s,j}-e_{s,j}\}\ge d_{s,j}-e_{s,j}\ge\mu(k_j)-2e_{s,j}$,
the total deficit is at most $\sum_j2e_{s,j}=2^{-s}$. Since
$E_s\uparrow E$, the successive maxima of $L_s$ increase to $\mu(E)$,
uniformly in the enumerator.
Conversely, singletons and their complements are
effectively open, uniformly in the integer. Lower approximations to
both measures, whose sum is one, give upper and lower approximations to
each point probability. Searching until their difference is less than
$2^{-p}$ terminates and gives the presentation in
\cref{def:kt-computable-law}. Thus the two definitions agree here.

\begin{proof}[Proof of \cref{prop:computable-rec}]
Fix a rational $a$ with $0<a<\alpha$. We use a rational $a<\alpha$
because $\alpha$ itself need not be computable, and the decider must
compare against a fixed rational. We first give an algorithm which,
on a positive integer $q$, finds an integer $b_{\rm tail}(q)\ge1$ with
\begin{equation}
 \Pr(K>b_{\rm tail}(q))<\frac{a}{2q}.
 \label{eq:computable-tail-cutoff}
\end{equation}
For $s=1,2,\ldots$, approximate every mass in $[-s,s]\cap\mathbb{Z}$ with
error at most $2^{-s}/(2s+1)$. The algorithm in
\cref{def:kt-computable-law} provides this by choosing a sufficiently
large integer precision. If $d_{k,s}$ is the resulting approximation,
compute the rational lower bound
\[
 L_s=\sum_{k=-s}^s
       \max\left\{0,d_{k,s}-\frac{2^{-s}}{2s+1}\right\}.
\]
Then
$0\le\mu([-s,s])-L_s\le2^{1-s}$.
Continuity from below gives $\mu([-s,s])\to1$, so also $L_s\to1$.
Stop at the first $s$ for which the exact rational test
$L_s>1-a/(2q)$ succeeds and put $b_{\rm tail}(q)=s$.
The search terminates, and the mass outside $[-s,s]$ is then less than
$a/(2q)$, proving \eqref{eq:computable-tail-cutoff}.

If each of the first $q$ proposals is at most $b_{\rm tail}(q)$, then
\eqref{eq:positive-sum} bounds the maximal position by $qb_{\rm tail}(q)$.
For every integer $T>qb_{\rm tail}(q)$, a union bound gives
\[
 \Pr(\tau_T\le q)
 \le q\Pr(K>b_{\rm tail}(q))<a/2<\alpha.
\]
Now fix a language witness $(M,Q)$ at confidence $\alpha$ and define
the computable integer function
\[
 H(n)=\begin{cases}
 0,&Q(n)=0,\\
 Q(n)b_{\rm tail}(Q(n)),&Q(n)\ge1.
 \end{cases}
\]
Completeness forces $A_M(x)\le H(|x|)$ on every yes-instance.
A deterministic decider reads the input to determine its length,
computes $H(|x|)$, and simulates $M$ from its original initial
configuration through that logical index, checking index zero as well.
It accepts exactly when this simulation accepts and otherwise rejects.
The simulation includes every yes-instance's accepting index, while
soundness rules out acceptance on a no-instance. The computation of $H$ and the timed simulation both terminate.
\end{proof}

\begin{proof}[Proof of \cref{prop:arbitrary-re}]
For every law, soundness and completeness imply that a witnessing machine
accepts exactly its language, so the class is contained in $\mathbf{RE}$.
For the reverse inclusion, enumerate all machines as $M_1,M_2,\ldots$.
Let $e_j$ be the maximum halting index among halting computations of the
first $j$ machines on inputs of length at most $j$, or zero if none halt.
The set is finite. Set $\mathcal E_0=0$ and
$\mathcal E_j=1+\max\{\mathcal E_{j-1},e_j\}$. This is a halting-time envelope of the
classical busy-beaver type~\cite[Section~VI]{Rado1962}. Define the law by
$\mu(\mathcal E_j)=1/(j(j+1))$ for $j\ge1$, with zero mass elsewhere.
The atoms are strictly increasing, the masses sum to one, and
$\Pr(K\ge \mathcal E_j)=1/j$.

Fix a recognizing machine $M_i$, a confidence $\alpha$, and an accepted
input of length $n$. Put $j=\max\{i,n,1\}\le i+n$.
Its accepting index is at most $e_j<\mathcal E_j$. Any proposal at least $\mathcal E_j$
therefore reaches its accepting tail. Choose an integer
$k_\alpha\ge\max\{1,\ell_\alpha\}$ and the integer polynomial
$Q(n)=k_\alpha(i+n)$. Independent proposals give failure probability at most
$(1-1/j)^{Q(n)}\le e^{-Q(n)/j}\le e^{-k_\alpha}\le1-\alpha$.
No-instance paths have no accepting index, so the same law realizes every
semidecidable language at every fixed confidence. Since such languages
include undecidable ones even on a unary alphabet,
\cref{prop:computable-rec} implies that this law is noncomputable.
\end{proof}

\begin{proof}[Proof of \cref{cor:computable-union}]
\Cref{prop:computable-rec} proves both upper inclusions. For the lower
inclusions let $M$ be a total decider for $L\in\mathbf{REC}$. For $j\ge1$,
compute the maximum $t(j)$ of $T_M(x)$ over the finite set $|x|\le j$
by running $M$ on all those words. Set
\[
 \mathcal E_0=0,\qquad \mathcal E_j=1+\max\{\mathcal E_{j-1},t(j)\},\qquad
 \mu(\mathcal E_j)=\frac1{j(j+1)}\quad(j\ge1),
\]
with zero mass elsewhere. The atoms are strictly increasing computable
positive integers. The masses telescope to one and
$\Pr(K\ge \mathcal E_j)=1/j$. Given a positive integer $k$, generate the atoms
until $\mathcal E_j\ge k$, which terminates since $\mathcal E_j\ge j$. Return
$1/[j(j+1)]$ if $\mathcal E_j=k$ and zero otherwise; nonpositive integers have
zero mass. This computes the point probabilities.

On an accepted input of length $n$, put $j=\max(1,n)\le n+1$.
Then $A_M(x)\le t(j)<\mathcal E_j$. Any proposal at least $\mathcal E_j$ reaches the
accepting tail from any nonnegative position. Choose an integer
$k_\alpha\ge\max(1,\ell_\alpha)$ and the polynomial budget
$Q(n)=k_\alpha(n+1)$. Its failure probability is at most
\[
 (1-1/j)^{Q(n)}\le e^{-Q(n)/j}\le e^{-k_\alpha}\le1-\alpha.
\]
Soundness holds because $M$ is a decider. The same computable law works at all confidences, with only the budget
changing.
\end{proof}
\space%
\section{Infinite positive moments are not sufficient}
\label{app:moment-boundary}

The moment obstruction in \cref{thm:kt-positive-moment} has no converse
for arbitrary laws. We construct a law with all positive moments infinite
and class $\mathbf{P}$. Its computability remains open.

\begin{proposition}
\label{prop:infinite-moments-p}
For every fixed finite nonempty input alphabet $\Sigma$, there is a law
$\mu$ supported on the positive integers such that
$\mathbb{E}K^r=\infty$ for every $r>0$, but
$\mathbf P^*_{\mu,\alpha}=\mathbf P^*_\mu=\mathbf{P}$ for every
$0<\alpha<1$.
\end{proposition}
\begin{proof}
For a machine $M$, let $L(M)$ be the set of words it eventually accepts.
Take a sequence $(M_j,Q_j,m_j)_{j\ge1}$ containing every triple in which
$L(M_j)\notin\mathbf{P}$, $Q_j$ is an integer polynomial budget, and
$m_j\ge2$ is an integer. Countability gives such a sequence.
Start with $W_0=1$ and $\varepsilon_0=1/2$. At stage $j$, choose an accepting
input $x_j$ such that, with $q_j=Q_j(|x_j|)$,
\[
 A_{M_j}(x_j)>W_{j-1}q_j.
\]
Such an input exists. Otherwise simulation through $W_{j-1}Q_j(n)$,
accepting exactly when $M_j$ accepts within that bound, would decide
$L(M_j)$ in polynomial time.
Choose a positive rational
\[
 \varepsilon_j\le\min\left\{\frac{\varepsilon_{j-1}}2,
                 \frac1{2m_j(q_j+1)}\right\},
 \qquad p_j=\varepsilon_j/2,
\]
and an integer $N_j>W_{j-1}$ so large that
$p_jN_j^{1/j}\ge1$. Set $W_j=N_j$.
After completing these choices, define one fixed law by
\[
 \mu(N_j)=p_j\ (j\ge1),\qquad
 \mu(1)=1-\sum_{j\ge1}p_j.
\]
All other masses are zero. Since $\varepsilon_j\le2^{-j-1}$, the sum of the
$p_j$ is at most $1/4$, so this is a probability law with positive mass
at one. Moreover $\varepsilon_\ell\le2^{-(\ell-j)}\varepsilon_j$ for $\ell\ge j$,
so $\sum_{\ell\ge j}p_\ell\le\varepsilon_j$. This bound includes all
future stages and protects every earlier allowance.

At stage $j$, every support point other than the later points
$N_\ell$, $\ell\ge j$, is at most $W_{j-1}$. If none of the first
$q_j$ proposals uses a later point, their sum is less than
$A_{M_j}(x_j)$. Thus, including when $q_j=0$,
\[
 \Pr(\tau_{A_{M_j}(x_j)}\le q_j)
 \le q_j\sum_{\ell\ge j}p_\ell
 \le q_j\varepsilon_j<1/m_j.
\]
Any language witness has $L=L(M)$, by soundness and completeness.
If it witnessed a language outside $\mathbf{P}$ at confidence $\alpha>0$,
choose $m\ge2$ with $1/m\le\alpha$. Its triple $(M,Q,m)$ appears in
the sequence and the preceding inequality contradicts completeness.
Every fixed-confidence class is therefore contained in $\mathbf{P}$.
All increments are at least one, so a polynomial-time decider reaches
its halting index within its ordinary time bound with probability one.
This proves the reverse inclusion at every confidence.

Finally, for any $r>0$ and every integer $j\ge1/r$,
$p_jN_j^r\ge p_jN_j^{1/j}\ge1$. Infinitely many terms of the moment
sum are at least one, so $\mathbb{E}K^r=\infty$.
The decidability bound for computable laws in \cref{prop:computable-rec}
does not settle the computability of this law: its language class is
already contained in $\mathbf{P}\subseteq\mathbf{REC}$.
\end{proof}
\space%
\section{Calibrating examples and confidence}
\label{app:calibration}

These examples calibrate the estimates and isolate the role of confidence and access costs. Each uses the conventions of \cref{sec:preliminaries}.

\subsection{The moment-bound constant}
\label{app:calibration-1}

A one-tick example attains the universal moment-bound coefficient.

\begin{remark}[The numerical constant in the moment bound]
\label{rem:moment-constant}
For any integer $T\ge1$ and $0<p\le1$, the law with masses $p$ at $T$
and $1-p$ at zero has $m_r=pT^r$ and
$\Pr(\tau_T\le q)=1-(1-p)^q$. At $q=1$ it attains equality in
\eqref{eq:moment-crossing}, so the universal coefficient $1$ cannot be
decreased. This is a family of laws as $T$ varies, not an asymptotic
equality for one fixed finite-moment law.
\end{remark}

\subsection{Zero-drift alternatives}
\label{app:calibration-2}

The fair walk calibrates the quadratic scale and distinguishes the two zero-drift hypotheses.

\begin{example}[The quadratic scale and the two alternatives]
\label{ex:zero-drift-calibration}
For the fair law $\mu(-1)=\mu(1)=1/2$,
\[
 \mathbb{E}_i\tau_T=T(T+1)-i(i+1)\quad(0\le i<T),\qquad
 \Pr_0(\tau_T\le q)\le\min\{1,q/T^2\}\quad(T\ge1).
\]
Indeed supported forward blocks give a finite mean hitting time.
Writing $e_i=\mathbb{E}_i\tau_T$, its first-step recurrence is $e_T=0$,
$e_0=1+(e_0+e_1)/2$ and
$e_i=1+(e_{i-1}+e_{i+1})/2$ for $0<i<T$; the displayed expression
is its unique solution. For the probability bound, the conditional gain
of $X_t^2$ is $1$ at positive states and $1/2$ at zero. Summing up to
$q\wedge\tau_T$ gives $\mathbb{E}_0X_{q\wedge\tau_T}^2\le q$, whereas a hit
contributes $T^2$. Thus
$\alpha T^2\le q_\alpha(T)\le\lceil T(T+1)/(1-\alpha)\rceil$,
so $q_\alpha(T)=\Theta_\alpha(T^2)$.
Conditional blocks sharpen the upper dependence to
$O(T^2(1+\ell_\alpha))$.

The two sufficient alternatives in \cref{prop:kt-zero-drift} are
distinct. Let $\Pr(W=j)=j^{-2}-(j+1)^{-2}$ for $j\ge1$ and
$m=\mathbb{E}W=\zeta(2)$. Telescoping gives $\Pr(W\ge j)=j^{-2}$;
the tail-sum formula (Fubini--Tonelli theorem~\cite[Theorem~1.7.2]{Durrett2019})
gives finite mean and infinite second moment.
Taking $K=1$ with probability $m/(1+m)$ and $K=-W$ otherwise gives
zero mean, bounded positive jumps, and infinite variance.
In contrast, masses $2^{-j-1}$ at each of $j$ and $-j$, $j\ge1$,
give zero mean and finite variance with unbounded positive jumps.
Both laws are computable and have class $\mathbf{P}$ by the theorem.
For the first, the series for $m$ has remainder at most $1/N$ after
$N$ terms, which supplies certified approximations.
\end{example}

\subsection{Biased nearest-neighbour walk}
\label{app:calibration-3}

A nearest-neighbour law calibrates the exponential crossing scale.

\begin{example}[Exponential scale for a biased nearest-neighbor law]
\label{ex:negative-scale}
Let $\mu(1)=p$, $\mu(-1)=1-p$, where $0<p<1/2$, and put
$\rho=p/(1-p)$. Then $q_\alpha(T)=\Theta_{p,\alpha}(\rho^{-T})$.
To see the lower bound, the probabilities $\pi_i=(1-\rho)\rho^i$
are stationary: $\pi_i p=\pi_{i+1}(1-p)$ at every edge, including
the holding boundary at zero. For these nearest-neighbor updates,
$i\mapsto\max(0,i+K)$ is nondecreasing. Coupling the walk from zero
with one started in $\pi$ therefore gives
$\Pr_0(X_t\ge T)\le\rho^T$, and a union bound gives
$\Pr_0(\tau_T\le q)\le q\rho^T$ for $T\ge1$.
For the upper bound, supported forward blocks give finite mean.
Writing $e_i=\mathbb{E}_i\tau_T$, the first-step recurrence gives
\[
 \mathbb{E}_0\tau_T=\frac{1-p}{(1-2p)^2}(\rho^{-T}-1)
              -\frac{T}{1-2p}.
\]
For example, writing $d_i=e_i-e_{i+1}$ yields
$d_0=1/p$ and $pd_i-(1-p)d_{i-1}=1$; summing the solution up to
$T-1$ gives this formula. Markov's inequality completes the quantile
upper bound. This calibrates the exponential spatial order in
\eqref{eq:exponential-crossing}, not its particular choice of $\kappa$.
The monotone coupling is specific to nearest-neighbor jumps: under the
held proposal $-2$, starts $1$ and $2$ become $1$ and $0$.
\end{example}

\subsection{Beyond exponential moments}
\label{app:calibration-4}

The all-moments classification applies beyond exponential moments.

\begin{example}[Beyond exponential moments]
\label{ex:negative-no-exp-moment}
Let
\[
 \mu(-1)=\tfrac12,\qquad
 \mu(2^j)=2^{-j^2-2j-4}\ (j\ge1),\qquad
 \mu(0)=\tfrac12-\sum_{j\ge1}2^{-j^2-2j-4}.
\]
The positive mass is less than $1/2$. Its positive mean is at most
$\sum_{j\ge1}2^{-2j-4}=1/48$, using $j^2\ge j$, whereas
$\mathbb{E}K^-=1/2$. Every positive moment is finite because
$\sum_{j\ge1}2^{-j^2+(s-2)j-4}<\infty$ for each $s>0$.
For any $\kappa>0$, the terms
$e^{\kappa2^j}2^{-j^2-2j-4}$ do not tend to zero, so
$\mathbb{E}e^{\kappa K}=\infty$. The law is computable: all nonzero
point masses are rational, and the remainder of the zero-mass series
after $N$ terms is at most
$\sum_{j>N}2^{-2j-4}=2^{-2N-4}/3$.
\Cref{thm:negative-all-moments} gives class $\mathbf{RDEF}$, while the
exponential-contraction criterion does not apply.
\end{example}

\subsection{One-jump budgets}
\label{app:calibration-5}

The next comparison identifies when a one-jump budget has the correct order.

\begin{remark}[When a one-jump budget has the right order]
\label{rem:polynomial-tail-scale}
Fix $0<a<1$ and the positive law
$\mu(j)=j^{-a}-(j+1)^{-a}$, $j\ge1$. Its exact tail is $T^{-a}$.
If $\sum_{j=1}^qK_j\ge T$, then $\sum_{j=1}^q\min(K_j,T)\ge T$.
The tail-sum formula and an integral comparison give
$\mathbb{E}\min(K,T)=\sum_{m=1}^T m^{-a}\le T^{1-a}/(1-a)$.
Markov's inequality therefore bounds the hitting probability by
$qT^{-a}/(1-a)$. The one-jump lower bound from
\cref{prop:kt-polynomial-tail} gives
$q_\alpha(T)=\Theta_{a,\alpha}(T^a)$ for this fixed law.
For the same tail form with $a>1$, the mean is finite. The first-moment
bound gives $q_\alpha(T)\ge\alpha T/\mathbb{E}K$, while $K\ge1$ gives
$q_\alpha(T)\le T$. Thus the one-jump $T^a$ budget can be nonoptimal.
Its tail hypothesis is not necessary for class $\mathbf{P}$ either:
$\delta_1$ already gives equality.
\end{remark}

\subsection{Sampling and access cost}
\label{app:calibration-6}

Exact sampling and charged access have different costs.

\begin{remark}[Effectiveness and access]
\label{rem:exp-effectiveness}
The law $\mu_{\mathrm{E}}$ has computable rational point masses and an effective tail.
It can be sampled exactly with fair bits in almost surely finite time.
For example, take $U$ uniform on $(0,1)$, determine $J=\lfloor1/U\rfloor$
by reading successive binary digits until the containing interval
$(1/(j+1),1/j]$ is determined, and output $K=2^J$.
Except at the countably many endpoints, this determination takes finitely
many digits; those endpoints have probability zero.
The interval lengths give $\Pr(J=j)=1/(j(j+1))$.

Writing a full binary sample of $K$ has infinite expected cost: its
length is $J+1$, and
$\mathbb{E}J=\sum_{m\ge1}\Pr(J\ge m)=\sum_{m\ge1}1/m=\infty$.
The $\mathbf{EXP}$ classification uses the free configuration access of
\cref{sec:preliminaries}.
\end{remark}

\subsection{Confidence for position-dependent clocks}
\label{app:calibration-7}

Different confidence levels can give different classes for a position-dependent kernel.

\begin{remark}[Confidence can matter for a position-dependent kernel]
\label{rem:confidence-lanes}
Define a kernel $\zeta$ by: from zero, move to $1$ or $2$ with
probability $1/2$ each; from $i\equiv1\pmod4$, move to $i+4$;
from even $i\ge2$, move to $i+Y$, $Y\sim\mu_{\mathrm{E}}$; from
$i\equiv3\pmod4$, hold. Every row is specified, and the last states
are unreachable from zero. The two reachable lanes never mix.
Then $\mathcal{C}_{\zeta,\alpha}=\mathbf{EXP}$ for $0<\alpha<1/2$, and
$\mathcal{C}_{\zeta,\alpha}=\mathbf{P}$ for $1/2\le\alpha<1$.

Since $X_q\ge q$ surely, polynomial-time deciders have a sure
polynomial budget. If $\alpha<1/2$, choose an integer
$k\ge\max\{1,\ell_{2\alpha}\}$ and use $Q(n)=1+kp(n)$ for an
exponential logical bound $2^{p(n)}$, $p\ge1$. Conditional on the even
lane, one of its $kp(n)$ samples crosses this bound with probability at
least $2\alpha$, by \eqref{eq:exp-crossing}.
Conversely, unused independent samples give
$\max_{t\le q}X_t\le4q+2+\sum_{j\le q}Y_j$.
The union-bound argument of \cref{lem:ceiling} forces
$T\le4q+2+q2^{q/\alpha}$ at confidence $\alpha$, giving an exponential
logical-time ceiling for a polynomial budget.
For finite $q\ge1$ and $T>4q+2$, the slow lane fails, and the fast
lane has positive failure probability: all its $q-1$ samples can equal
two. Total success is therefore strictly below $1/2$ at every such finite
horizon.
At confidence at least $1/2$, accepting time is therefore at most
$4Q(n)+2$, giving a polynomial-time simulation.
\end{remark}

\subsection{Confidence transfer for iid clocks}
\label{app:calibration-8}

For iid laws, a restart comparison transfers several radius bounds between confidence levels.

\begin{remark}[Confidence dependence for iid laws]
\label{rem:iid-confidence}
For every classified iid law in this paper, $\mathbf P^*_{\mu,\alpha}$
is independent of $\alpha$, whereas \cref{rem:confidence-lanes} gives a
position-dependent kernel whose class depends on the confidence.
Standard restart amplification for iid held walks needs care, since the
held update is not monotone in its starting point. For
$\mu=\tfrac12\delta_2+\tfrac12\delta_{-3}$,
$\Pr_1(\tau_4\le3)=3/8$, while $\Pr_0(\tau_4\le3)=1/2$.
Proposal maxima can still be amplified: a single proposal at least $T$
crosses from every nonnegative position. For $0<\alpha<\alpha'<1$,
choose an integer $k\ge1$ with $(1-\alpha)^k\le1-\alpha'$.
Independence gives
\[
 \Pr(J_{kq}\ge u)=1-(1-\Pr(J_q\ge u))^k,
\]
so $u_\alpha(q)\le u_{\alpha'}(kq)$. Together with
\cref{lem:proposal-maxima}, this yields
\[
 R_{\alpha'}(q)\le R_\alpha(q)\le qR_{\alpha'}(kq)
 \qquad(q\ge1).
\]
The left inequality follows from monotonicity in confidence; the right
follows from $R_\alpha(q)\le q u_\alpha(q)\le q u_{\alpha'}(kq)
\le qR_{\alpha'}(kq)$. At $q=0$, both radii vanish.
Polynomial upper bounds and exponential lower bounds for the radius
therefore hold at one confidence if and only if they hold at every
confidence. Exponential upper bounds are preserved as well, so the
hypotheses of \cref{lem:reachability-dictionary}(i) and (iv) are
independent of confidence. The factor $q$ is too coarse to transfer a
polynomial lower bound of arbitrary positive degree, a subpolynomial
upper bound, or an exact class below $\mathbf{P}$. For the sparse law of
\cref{thm:zero-drift-counterexample}, whose class lies between $\mathbf{RDEF}$
and $\mathbf{P}$ with strict upper containment, confidence dependence remains
open.
\end{remark}
\space%
\section{Common budgets over parameter sets}
\label{sec:fuzzy}

We ask whether one machine and one budget work for every fixed clock
in a nonempty allowed parameter set. For the two families below, the
criterion depends on the set's infimum.

\label{sec:two-layers}
\label{sec:joint-limits}
Let $\Theta$ denote the parameter domain of the clock family, and let its possibility distribution $\pi:\Theta\to[0,1]$ assign a degree of possibility to each fixed parameter. For a cut level $\gamma\in(0,1]$, write $\Theta_\gamma=\{\theta\in\Theta:\pi(\theta)\ge\gamma\}$; we consider nonempty cuts~\cite{Zadeh1978}. These grades specify allowed fixed parameter values; the probability law governs the clock at the chosen parameter throughout a run, rather than choosing a new parameter at every tick. The lazy-speed domain is $\Theta=[0,1]$, with positive-speed examples restricted to $(0,1]$; the activation domain is $\Theta=[0,1/2]$, with the positive-return example restricted to $(0,1/2]$. Only the infimum matters here: $\pi(p)=p$ gives lazy cuts $[\gamma,1]$ with linear budgets, while a constant possibility on $(0,1]$ allows speeds approaching zero and no common finite budget. For parity clocks, a cut $(0,1/2]$ has common class $\mathbf{P}$ although each member has class $\mathbf{EXP}$, whereas $\pi(\beta)=2\beta$ gives positive-infimum cuts $[\gamma/2,1/2]$ with common class $\mathbf{EXP}$.

\subsection{A lazy-clock benchmark}
\label{sec:lazy}

At every fixed positive speed, the lazy clock has class $\mathbf{P}$.
A family of such clocks can fail to admit a common finite budget.

The family is
\begin{equation}
 \mu_p=(1-p)\delta_0+p\delta_1,\qquad 0\le p\le1,
 \label{eq:lazy}
\end{equation}
with the held update \eqref{eq:hold}. Its position after $Q$ ticks is
binomial with parameters $Q$ and $p$; there are no backward proposals.
Let $\mathcal S\subseteq[0,1]$ be a nonempty allowed set, for instance a cut
$\Theta_\gamma$, put $p_*=\inf \mathcal S$, and define the common-budget class
$\mathcal U_{\mathcal S,\alpha}$ by \cref{def:kt-ac1} with the budget condition
replaced by
\[
 \inf_{p\in \mathcal S}\Pr^{(p)}\bigl(X_{Q(|x|)}\ge A_M(x)\bigr)\ge\alpha
 \quad\text{on yes-instances},
\]
one machine and one budget for all $p\in \mathcal S$. Also define
$q_{\alpha,\mathcal S}(T)=\min\{Q\in\mathbb{N}:\inf_{p\in \mathcal S}\Pr^{(p)}(X_Q\ge T)\ge\alpha\}$.
Because this walk never decreases, the event is threshold crossing by tick
$Q$. The criterion depends on $p_*$.

\begin{remark}[Uniformity criterion for lazy clocks]
\label{thm:lazy}
For the family \eqref{eq:lazy} and every $0<\alpha<1$:
\begin{enumerate}[label=(\roman*),leftmargin=*]
\item If $p_*>0$, then
 $\frac{\alpha T}{p_*}\le q_{\alpha,\mathcal S}(T)
 \le\left\lceil\frac{T}{p_*(1-\alpha)}\right\rceil$,
 for $T\ge1$. Moreover $p_*q_{\alpha,\mathcal S}(T)/T\to1$,
uniformly over all nonempty allowed sets $\mathcal S$ with $p_*>0$,
in the precise sense
 \[
 \lim_{T\to\infty}
 \sup_{\substack{\varnothing\ne \mathcal S\subseteq(0,1]\\p_*:=\inf \mathcal S>0}}
 \left|\frac{p_*q_{\alpha,\mathcal S}(T)}{T}-1\right|=0.
 \]
\item If $p_*=0$, then $q_{\alpha,\mathcal S}(T)=\infty$ for every $T\ge1$.
\item $\mathcal U_{\mathcal S,\alpha}=\mathbf{P}$ if $p_*>0$, and
 $\mathcal U_{\mathcal S,\alpha}=\{\varnothing,\Sigma^*\}$ if $p_*=0$, while every
 single $p>0$ gives $\mathbf P^*_{\mu_p,\alpha}=\mathbf{P}$.
\end{enumerate}
\end{remark}
\begin{proof}[Proof sketch]
Couple successes with common uniforms; binomial crossing probabilities are nondecreasing and continuous in $p$, so their infimum is the value at $p_*$, attained or not.
For $p_*>0$, Markov's inequality gives both displayed bounds, and $H_T$ is a sum of $T$ geometric waits with mean $T/p$ and variance $T(1-p)/p^2$.
Chebyshev's bound $\Pr(|H_T-T/p|\ge\varepsilon T/p)\le1/(\varepsilon^2T)$ is uniform in $p>0$; squeezing the quantiles and allowing rounding error $p_*/T\le1/T$ proves the uniform limit.
When $p_*=0$, $\Pr^{(p)}(X_Q\ge T)\le Qp$ for $T\ge1$, so no common finite budget exists.
Bounded jumps give the upper class bound by \cref{prop:bounded-forward}, and a positive infimum supplies a common polynomial budget to any polynomial-time decider; at zero infimum, \cref{prop:kt-stuck-walk}'s time-zero argument gives only the trivial languages.
\end{proof}

\subsection{Common budgets for the activation family}
\label{sec:joint}

For the activation clocks of \cref{sec:activation}, let the allowed
backward probabilities form a nonempty set $\mathcal B\subseteq[0,1/2]$
and put $b_*=\inf \mathcal B$. A possibility cut is one such set.

\begin{definition}[Common-budget class for the activation family]
\label{def:joint}
For $0<\alpha<1$, a language $L$ belongs to $\mathcal U_{\mathcal B,\alpha}$ if
some deterministic machine $M$ and integer polynomial $Q$ satisfy
$A_M(x)=\infty$ on every $x\notin L$, and, on every $x\in L$,
$A_M(x)<\infty$ and
\[
 \inf_{\beta\in \mathcal B}\Pr^{(\beta)}\bigl(\tau_{A_M(x)}\le Q(|x|)\bigr)\ge\alpha.
\]
One machine and one budget serve every $\beta\in \mathcal B$. Also let
$q_{\alpha,\mathcal B}(T)=\min\{q\in\mathbb{N}:\inf_{\beta\in \mathcal B}\Pr^{(\beta)}(\tau_T\le
q)\ge\alpha\}$.
\end{definition}

\begin{theorem}[Common budgets over backward-probability cuts]
\label{thm:joint}
Let $\mathcal B\subseteq[0,1/2]$ be nonempty with $b_*=\inf \mathcal B$, and fix
$0<\alpha<1$. Then:
\begin{enumerate}[label=(\roman*),leftmargin=*]
\item if $b_*>0$, then $\mathcal U_{\mathcal B,\alpha}=\mathbf{EXP}$ and
 \[
 \lim_{T\to\infty}\frac{q_{\alpha,\mathcal B}(T)}{\log_2T}
 =\left(1+\frac1{b_*}\right)\ell_\alpha;
 \]
\item if $b_*=0$, then $\mathcal U_{\mathcal B,\alpha}=\mathbf{P}$ and
 $q_{\alpha,\mathcal B}(T)=q_\alpha^{(0)}(T)$ for every $T\in\mathbb{N}$, and
 $q_{\alpha,\mathcal B}(T)/T\to1$.
\end{enumerate}
Constants may depend on $\mathcal B$ and $\alpha$; in (ii) they may be chosen to
depend only on $\alpha$. Neither the attainment of the infimum nor the
boundary convention of the cut is used.
\end{theorem}

Every fixed $\beta>0$ has $\mathcal{C}_{\eta_\beta,\alpha}=\mathbf{EXP}$ by
\cref{thm:activation}; the theorem gives common class $\mathbf{P}$ when
the allowed positive values have infimum zero. The proof rests on
finite-horizon continuity, a comparison of running maxima with the
zero-parameter clock, and the resulting expected crossing-time bound uniform in $\beta$.

\begin{lemma}[Finite-horizon parameter continuity]
\label{lem:coupling}
For every $\beta,\beta'\in[0,1/2]$, $q\ge0$, and $T\in\mathbb{N}$,
\[
 \bigl|\Pr^{(\beta)}(\tau_T\le q)-\Pr^{(\beta')}(\tau_T\le q)\bigr|
 \le\min\{1,q|\beta-\beta'|\}.
\]
Consequently, for nonempty $\mathcal B$ and $b_*=\inf \mathcal B$,
$\inf_{\beta\in \mathcal B}\Pr^{(\beta)}(\tau_T\le q)
\le\Pr^{(b_*)}(\tau_T\le q)$.
\end{lemma}
\begin{proof}
Rows at zero and at odd states agree. At each positive even state the
rows transfer mass $|\beta-\beta'|$ between the hold and the move $-1$;
their total variation distance is $|\beta-\beta'|$. Whenever the histories
agree, couple the next states maximally. A union bound bounds the chance
of disagreement by tick $q$ by $q|\beta-\beta'|$. On agreement, the
threshold-crossing events agree. Probabilities differ by at most one as well.
Since $b_*$ is in the closure of $\mathcal B$, take a sequence of allowed
parameters tending to it. Continuity at the fixed horizon gives the
upper bound on the infimum.
\end{proof}

\begin{lemma}[Comparison with the zero-parameter clock]
\label{lem:baseline-order}
For every $\beta\in[0,1/2]$, $q,T\in\mathbb{N}$,
$\Pr^{(\beta)}(\tau_T\le q)\ge\Pr^{(0)}(\tau_T\le q)$.
If $\inf \mathcal B=0$, their lower envelope is exactly
$\inf_{\beta\in \mathcal B}\Pr^{(\beta)}(\tau_T\le q)=\Pr^{(0)}(\tau_T\le q)$.
\end{lemma}
\begin{proof}
Write $X_t^{(0)}$ for the baseline. Couple both clocks' first two ticks
using the same initial heavy sample, so both positions at tick two are
$2+Y$ and are even. Subsequently use independent fair coins, one per
tick: the baseline moves $+2$ on success and holds on failure. At an
even state the $\beta$-clock also moves $+2$ on success; on failure it
returns by one with conditional probability $2\beta$ and otherwise
holds. This gives the required probabilities $1/2,\beta,1/2-\beta$.
At an odd state it takes its mandatory displacement $1+Y'\ge3$ with a
fresh independent heavy sample. All extra choices are independent of
the coins and earlier samples, so each marginal has its stated kernel.

Inductively, at every $t\ge2$ an even $\beta$-position is at least
$X_t^{(0)}$, and an odd $\beta$-position is at least $X_t^{(0)}-1$.
From an even state, coin success preserves the comparison; coin failure
holds the baseline and either holds the other clock or lowers it by one.
From an odd state, the displacement at least three restores the even
comparison even if the baseline moves by two. If the $\beta$-position
is one below the baseline, its preceding even position was at least that
baseline position: the baseline held on the return tick. Thus at every
tick its running maximum is at least the baseline position, which is
itself nondecreasing. The running maxima, and hence the hitting events,
are ordered. The equality of lower envelopes when $\inf \mathcal B=0$ follows
from this inequality and \cref{lem:coupling}.
\end{proof}

\begin{lemma}[Uniform linear expected crossing time]
\label{lem:unif-expectation}
From zero, for every $\beta\in[0,1/2]$ and integer $T\ge1$,
\[
 \mathbb{E}^{(\beta)}\tau_T\le\mathbb{E}^{(0)}\tau_T\le2T,
 \qquad
 \Pr^{(\beta)}(\tau_T>q)\le\frac{2T}{q}\quad(q\ge1).
\]
\end{lemma}
\begin{proof}
\Cref{lem:baseline-order} gives
$\Pr^{(\beta)}(\tau_T>q)\le\Pr^{(0)}(\tau_T>q)$ at every integer $q\ge0$.
Sum these tails. Tonelli's theorem for nonnegative
summands~\cite[Theorem~1.7.2]{Durrett2019} gives
\[
 \mathbb{E}^{(\beta)}\tau_T
 =\sum_{q=0}^\infty\Pr^{(\beta)}(\tau_T>q)
 \le\sum_{q=0}^\infty\Pr^{(0)}(\tau_T>q)
 =\mathbb{E}^{(0)}\tau_T\le2T,
\]
where the last inequality is \cref{lem:baseline}. Markov's inequality
proves the tail bound.
\end{proof}

\begin{proof}[Proof of \cref{thm:joint}(ii)]
First, $\mathcal B$ contains a sequence $\beta_j\to b_*=0$ (or contains zero). Let
$(M,Q)$ witness $L\in\mathcal U_{\mathcal B,\alpha}$, and fix a yes-instance $x$
with budget $q=Q(|x|)$. For every $j$,
$\Pr^{(\beta_j)}(\tau_{A_M(x)}\le q)\ge\alpha$, and
\cref{lem:coupling} gives
$\Pr^{(0)}(\tau_{A_M(x)}\le q)\ge\alpha-q\beta_j$. Letting $j\to\infty$ shows
that the same $(M,Q)$ witnesses $L\in\mathcal{C}_{\eta_0,\alpha}=\mathbf{P}$. Hence
$\mathcal U_{\mathcal B,\alpha}\subseteq\mathbf{P}$.

Conversely, let $M$ decide $L$ within the integer polynomial bound $H(n)$.
Choose an integer $k_{\alpha,\rm lin}\ge2/(1-\alpha)$ and set $Q(n)=k_{\alpha,\rm lin}(H(n)+1)$, an integer
polynomial with $Q(n)\ge1$. On a yes-instance with accepting index
$T=A_M(x)\le H(n)$: if $T=0$ then $\tau_0=0$; otherwise
\cref{lem:unif-expectation} and monotonicity in the target give, for
every $\beta\in[0,1/2]$,
\[
 \Pr^{(\beta)}(\tau_T>Q(n))\le\frac{2T}{Q(n)}
 \le\frac{2H(n)}{k_{\alpha,\rm lin}(H(n)+1)}\le1-\alpha.
\]
The same $M,Q$ therefore work uniformly over every $\mathcal B$, and
$\mathbf{P}\subseteq\mathcal U_{\mathcal B,\alpha}$.

By \cref{lem:baseline-order}, the lower envelope equals the baseline
hitting probability at every finite horizon, so
$q_{\alpha,\mathcal B}(T)=q^{(0)}_\alpha(T)$ for all $T\in\mathbb{N}$.
The limit in \cref{thm:sharp} gives the exact linear leading constant.
\end{proof}

\begin{proof}[Proof of \cref{thm:joint}(i)]
Let $b_*>0$, so every $\beta\in \mathcal B$ lies in $[b_*,1/2]$.

For $\mathbf{EXP}\subseteq\mathcal U_{\mathcal B,\alpha}$, let $M$ decide $L$ within
$H(n)=2^{p(n)}$ transitions, where $p$ has nonnegative integer coefficients and $p(n)\ge1$
for every $n\in\mathbb{N}$. Choose an integer
$k_{\alpha,\mathcal B}\ge\max\{1,b_*^{-1}\ell_\alpha\}$ and put
$Q(n)=2k_{\alpha,\mathcal B}p(n)$, an allowed polynomial. For any $\beta\in \mathcal B$ and any integer $r$ with
$2r\le Q(n)$, \cref{lem:amplify} gives, from every state below a target
$T\ge1$,
\[
 \Pr^{(\beta)}(\tau_T>2r)\le\exp\!\left(-\frac{\beta r}{h(T)}\right)
 \le\exp\!\left(-\frac{b_*r}{h(T)}\right).
\]
Take $T=H(n)$, so $h(T)=p(n)$, and
$r=k_{\alpha,\mathcal B}p(n)\in\mathbb{N}$. Then
$2r\le Q(n)$ and $b_*r\ge b_*k_{\alpha,\mathcal B}p(n)\ge p(n)\ell_\alpha$, so
$\Pr^{(\beta)}(\tau_{H(n)}>Q(n))\le1-\alpha$ for every $\beta\in \mathcal B$
simultaneously. Since $A_M(x)\le H(n)$ and crossing times are monotone in
the target, $(M,Q)$ witnesses $L\in\mathcal U_{\mathcal B,\alpha}$.

For the reverse inclusion, $\mathcal B$ is nonempty; fix any $\beta_0\in \mathcal B$, so
$\beta_0\ge b_*>0$. A common witness for $\mathcal B$ also witnesses membership
for the single clock $\eta_{\beta_0}$, and
$\mathcal{C}_{\eta_{\beta_0},\alpha}=\mathbf{EXP}$ by \cref{thm:activation}, so
$\mathcal U_{\mathcal B,\alpha}\subseteq\mathbf{EXP}$.

It remains to identify the leading constant. Couple the geometric
waiting times of \cref{thm:sharp} using common iid uniforms $U_{r,s}$:
put $G_r^{(\beta)}=\min\{s\ge1:U_{r,s}\le\beta\}$.
Then $G_r^{(\beta)}$ and the cycle lengths
$D_r^{(\beta)}=G_r^{(\beta)}+1$ are nonincreasing in $\beta>0$.
The resulting completed-heavy-sample counts $N_q^{(\beta)}$ are
nondecreasing in $\beta$. For each parameter this is the renewal representation in
\cref{thm:sharp}, independent of the heavy samples.

For $T>2q$ and $q\ge1$, the lower bound in
\eqref{eq:sharp-sandwich} and this count ordering give
\[
 L_{b_*}(T,q):=1-\mathbb{E}\left(1-\frac1{h(T)}\right)^{N_q^{(b_*)}}
 \le\inf_{\beta\in \mathcal B}\Pr^{(\beta)}(\tau_T\le q)
 \le\Pr^{(b_*)}(\tau_T\le q),
\]
where the last inequality is \cref{lem:coupling}. For
$q=\lfloor c\log_2T\rfloor$ with fixed $c>0$, the fixed-parameter
renewal argument in \cref{thm:sharp} gives
$N_q^{(b_*)}/\log_2T\to c/(1+1/b_*)$ in probability and
$h(T)/\log_2T\to1$. For natural logarithms, the expansion
$N_q^{(b_*)}\ln(1-1/h(T))=-N_q^{(b_*)}/h(T)+o(1)$
holds in probability, since $N_q^{(b_*)}\le q=O(\log_2T)$.
The powers are bounded by one, so the expectation argument in
\cref{thm:sharp} applies and their expectations tend to
$\exp(-c/(1+1/b_*))$. Both bounding expressions
converge to $1-\exp(-c/(1+1/b_*))$, so the lower envelope has that limit.

Set $c_*=(1+1/b_*)\ell_\alpha$. For any
$0<\varepsilon<c_*$, the limiting lower-envelope probabilities at
$c_* -\varepsilon$ and $c_*+\varepsilon$ are respectively below and
above $\alpha$. Fix error margins smaller than half these two positive
gaps. For all sufficiently large $T$ the lower-envelope probabilities
fall below and above $\alpha$, respectively.
Monotonicity in the integer budget squeezes the common
quantile between the corresponding floor budgets for all sufficiently
large $T$. Dividing by $\log_2T$ and letting $\varepsilon\downarrow0$
proves the limit.
\end{proof}
\space%
\end{appendices}
\bibliographystyle{unsrtnat}
\bibliography{references}
\end{document}